\pdfoutput=1
\documentclass[11pt,oneside,leqno]{article}
\usepackage{lmodern}
\usepackage{dsfont}
\usepackage[T1]{fontenc}
\usepackage{geometry}
\usepackage{framed}
\usepackage{enumitem}
\usepackage{algpseudocode,algorithm}
\usepackage{soul}
\usepackage{showlabels}
\usepackage{ulem}
\usepackage{authblk}

\usepackage{float}
\usepackage{amstext}
\usepackage{amsthm}
\usepackage{amsmath,amsfonts}
\usepackage{amssymb}
\usepackage{mathabx}
\usepackage{bm}
\usepackage{color}
\usepackage{xfrac}
\usepackage{bbm}
\usepackage{todonotes}

\usepackage{amssymb}
\usepackage{graphicx}
\usepackage[numbers]{natbib}
\usepackage[unicode=true,
 breaklinks=false,
 pdfborder={0 0 1},
 colorlinks=false]
 {hyperref}
\usepackage{mdframed}
\usepackage{subcaption}
\usepackage[page]{appendix}

\usepackage{fancyhdr}

\renewcommand{\mathbb}[1]{\mathbbmss{#1}}

\makeatletter
\numberwithin{equation}{section}
\numberwithin{figure}{section}

\newcommand{\rE}{\mathbb{E}}
\newcommand{\rR}{\mathbb{R}}
\newcommand{\rZ}{\mathbb{Z}}
\newcommand{\RR}{\mathcal{R}}
\newcommand{\rP}{\mathbb{P}}
\newcommand{\rd}{\mathrm{d}}
\newcommand{\re}{\mathrm{e}}
\newcommand{\ZZ}{\mathcal{Z}}
\newcommand{\FF}{\mathcal{F}}
\newcommand{\GG}{\mathcal{G}}
\newcommand{\lref}{\lambda_{\textrm{ref}}}
\newcommand{\xx}{\bm{x}}

\newcommand{\ZZZ}{\bm{Z}}
\newcommand{\zz}{\bm{z}}
\newcommand{\vv}{\bm{v}}
\newcommand{\AAA}{\mathcal{A}}
\newcommand{\lann}{\langle\!\langle}
\newcommand{\rann}{\rangle\!\rangle}

\newcommand{\tx}{\widetilde{X}}
\newcommand{\tv}{\widetilde{V}}
\newcommand{\grad}{\nabla}
\newcommand{\var}{\mathrm{Var}}
\newcommand{\inner}[2]{\ensuremath{%
		\left\langle#1,#2\right\rangle%
}}
\renewcommand{\r}{\right}
\renewcommand{\l}{\left}
\newcommand{\Tr}{\mathrm{Tr}}

\newcommand{\aos}{Ann. Statist.}
\newcommand{\aap}{Ann. Appl. Probab.}

\usepackage{amsthm}
\usepackage{enumitem}
\usepackage{cleveref}

\setlist[enumerate,1]{label=(\roman*),ref=(\roman*)}
\setlist[enumerate,2]{label=(\alph*),ref=(\roman{enumi}-\alph*)}
\setlist[enumerate,3]{label=(\Alph*),ref=(\roman{enumi}-\alph{enumii}-\Alph*)}
\setlist[enumerate,4]{label=(\arabic*),ref=(\roman{enumi}-\alph{enumii}-\Alph{enumiii}-\arabic*)}

\usepackage{color}
\allowdisplaybreaks

\usepackage{epstopdf}
\newtheorem{theorem}{Theorem}

\newtheorem{proposition}[theorem]{Proposition}

\newtheorem{lemma}{Lemma}
\newtheorem{assumption}{Assumption}
\newtheorem*{assumption*}{Assumptions}
\theoremstyle{definition}
\newtheorem{remark}{Remark}

\usepackage{apptools}
\AtAppendix{\counterwithin{lemma}{section}}

\begin{document}

\title{Randomized Hamiltonian Monte Carlo as Scaling Limit of the Bouncy Particle Sampler and Dimension-Free Convergence Rates}

\author[1]{George Deligiannidis}
\author[2]{Daniel Paulin}
\author[3]{Alexandre Bouchard-C\^ot\'e}
\author[4]{Arnaud Doucet}
\affil[1,4]{\small Department of Statistics, University of Oxford, UK.}
\affil[2]{\small School of Mathematics, University of Edinburgh, UK.}
\affil[3]{\small Department of Statistics, University of British Columbia, Canada.}
{

    \makeatletter
    \renewcommand\AB@affilsepx{, \protect\Affilfont}
    \makeatother

\affil[1]{\small\texttt{deligian@stats.ox.ac.uk}}
\affil[2]{\texttt{dpaulin@ed.ac.uk}}
\affil[3]{\texttt{bouchard@stat.ubc.ca}}
\affil[4]{\texttt{doucet@stats.ox.ac.uk}}
}
\maketitle
\thispagestyle{fancy}

\begin{abstract}
The Bouncy Particle Sampler is a Markov chain Monte Carlo method based on a non-reversible piecewise deterministic Markov process.
In this scheme, a particle explores the state space of interest by evolving according to a linear dynamics which is altered by bouncing on the hyperplane perpendicular to the gradient of 
the negative log-target density at the arrival times of an inhomogeneous Poisson Process (PP) and by randomly perturbing its velocity at the arrival times of a homogeneous PP.
Under regularity conditions, we show here that the process corresponding to the first component of the particle and its corresponding velocity converges weakly towards a
Randomized Hamiltonian Monte Carlo (RHMC) process as the dimension of the ambient space goes to infinity. RHMC is another piecewise deterministic non-reversible
Markov process where a Hamiltonian dynamics is altered at the arrival times of a homogeneous PP by randomly perturbing the momentum component.
We then establish dimension-free convergence rates for RHMC for strongly log-concave targets with bounded Hessians using coupling ideas and hypocoercivity techniques.
We use our understanding of the mixing properties of the limiting RHMC process to choose the refreshment rate parameter of BPS. This results in significantly better performance in our simulation study than previously suggested guidelines.\end{abstract}

{\small{}{}Keywords: }Bouncy particle sampler; Coupling; Randomized Hamiltonian Monte Carlo; Weak Convergence; Hypocoercivity.

\section{Introduction}
Assume one is interested in sampling from a target probability density on $\mathbb{R}^d$ which can be evaluated pointwise up to an intractable normalizing constant.
In this context one can use Markov chain Monte Carlo (MCMC) algorithms to sample from, and compute expectations with respect to the target measure.
Despite their great success, standard MCMC methods, such as the ubiquitous Metropolis--Hastings algorithm,
tend to perform poorly on high-dimensional targets. To address this issue,
several new methods have been proposed over the past few decades. Popular alternatives include the Metropolis-adjusted Langevin algorithm (MALA)
\cite{rossky1978brownian, roberts1996exponential}, Hamiltonian, or Hybrid, Monte Carlo (HMC) \cite{duane1987hybrid} and slice sampling \cite{neal2003slice}.

Recently, a novel class of non-reversible, continuous-time MCMC algorithms based on piecewise-deterministic Markov processes (PDMP) has appeared in applied probability \cite{monmarche2016piecewise,bierkens2017piecewise}, automatic control \cite{mesquita2012jump},
physics \cite{peters2012rejection,michel2014generalized,nishikawa2016event} statistics and machine learning
\cite{bouncy2018,bierkens2016zig,RHMC,vanetti2017piecewise,B_BC_D_D_F_R_J_16,pakman2016stochastic,wu2017generalized}.
Most of the current literature revolves around two piecewise-deterministic MCMC (PDMCMC) schemes: the Bouncy Particle Sampler (BPS) \cite{peters2012rejection,bouncy2018} and the Zig-Zag sampler \cite{bierkens2016zig}.
A practical advantage of the  BPS and Zig-Zag algorithms is that in many models it is possible to simulate their piecewise linear paths without time-discretization \cite{bouncy2018}. In contrast, methods based on either diffusions or Hamiltonian paths require time discretization and moreover their performance is known to collapse if the discretization is too coarse.
Despite the increasing interest in these piecewise linear PDMCMC algorithms, our theoretical understanding of their properties remains limited, although
a fair amount of progress has been achieved recently in establishing geometric ergodicity, see \cite{deligiannidis2017exponential,durmus2018geometric} for BPS and
\cite{fetique2017long,bierkens2017ergodicity} for Zig-Zag. However, all of these results tend to provide convergence rates that deteriorate with the dimension and thus fail to capture the
empirical performance of these PDMCMC algorithms on high-dimensional targets.

Scaling limits have become a very popular tool for analysing and comparing MCMC algorithms in high-dimensional scenarios since their introduction in the seminal paper \cite{roberts1997weak}; see, e.g., \cite{roberts1998optimal,beskos2013optimal}. They have been used to establish the computational complexity  of
the most popular MCMC algorithms, which is $O(d^2)$ for Random Walk Metropolis (RWM), $O(d^{4/3})$ for MALA and $O(d^{5/4})$ for HMC; here computational complexity is defined in terms of the expected squared jump distance.
In this direction, the recent work of \citet{bierkens2018high} has established scaling limits for both Zig-Zag and global BPS for high-dimensional standard Gaussian targets.
They obtain the scaling limits of several finite dimensional statistics, namely the angular velocity, the log-density and the first coordinate.
In this context, it is shown that Zig-Zag has algorithmic complexity $O(d)$ for all three types of statistics, whereas global BPS has complexity $O(d)$ for angular momentum and
$O(d^2)$ for the other two types of statistics. Benefits of Zig-Zag over global BPS are to be expected in this scenario. Indeed, when applied to a product target, the Zig-Zag sampler
factorises into independent components and is closely related to Local-BPS (LBPS); see
\cite{peters2012rejection,bouncy2018}. The standard (global) BPS studied herein and in \citet{bierkens2018high}, just like RWM, MALA and HMC, is an algorithm whose dynamics do not distinguish between product and non-product targets.

In the present paper, we also study scaling limits for BPS on a very general class of targets that greatly extends the i.i.d.\ scenario, and its variants, often considered in the literature, see e.g. \cite{roberts1997weak,roberts1998optimal,beskos2013optimal,bierkens2018high}.
We concentrate on the first coordinate and its corresponding velocity in a regime which differs from the one considered in \cite{bierkens2018high} in the following three ways:
(a) \cite{bierkens2018high} considers BPS with the location evolving at unit speed, whereas in our scenario the velocity is Gaussian, therefore with speed scaling like $\sqrt{d}$ in the dimension;
(b) \cite{bierkens2018high} considers scaling limits for the first coordinate of the location process only,
whereas we look at both location and velocity; and finally (c)\cite{bierkens2018high} rescales time with a factor $d$, whereas we obtain our limiting process on the natural time scale. As a result we obtain a different scaling limit
which suggests that BPS has algorithmic complexity $O(d^{3/2})$ if one is interested on low-dimensional projections, at least on weakly dependent targets. This is in agreement with the empirical results reported in \cite{bouncy2018}.   Given the different regimes and different objects studied in \cite{bierkens2018high} and the present paper, it is not surprising that the
two scaling limits differ significantly, with our bound being tighter and seemingly better at capturing the empirical behaviour of the process. In \cite{bierkens2018high} the first location coordinate converges to a Langevin diffusion, whereas in the present paper the process
tracking the first location and velocity components converges to a piecewise deterministic Markov process known as Randomized Hamiltonian Monte Carlo (RHMC).
Although the corresponding Fokker-Planck equation  was studied in \citet{dolbeault2015hypocoercivity}, using a related approach to ours,
RHMC was first studied in a Monte Carlo context in \cite{RHMC}.

To the best of our knowledge, our result is the first in the literature establishing a direct link between BPS and Hamiltonian dynamics.
It is our understanding that the Langevin diffusion obtained in \cite{bierkens2018high} can be obtained from RHMC by a further limiting procedure similar to the \textit{overdamped} regime of the Langevin equation.
In addition, the assumptions under which our scaling limit is obtained allow much more complex dependence structures than those considered in the literature, see e.g. \cite{bedard2007weak,bedard2019hierarchical,yang2019optimal,breyer2000metropolis,roberts1997weak,roberts1998optimal,beskos2013optimal}, where the target is assumed to factorise or to possess a hierarchical structure. In addition, in the scenario we consider all dimensions have an impact, in contrast with the Hilbert-space setting, see e.g.\ \cite{mattingly2012diffusion}, where only a fixed, finite number of dimensions is significant.

The second part of the paper is concerned with the convergence properties of RHMC. This process was studied in \cite{RHMC} where it was established that it is \textit{geometrically ergodic}. However, it is not clear whether such an approach  can provide dimension independent convergence rates. The earlier work of \cite{dolbeault2015hypocoercivity} studies the corresponding Fokker-Planck equation, tracking the evolution of densities rather than conditional expectations.
In recent years, there has been great success in obtaining dimension-free convergence rates of MCMC schemes for strongly log-concave targets  with bounded Hessians; see for example \cite{dalalyan2017theoretical,durmus2017nonasymptotic,mangoubi2017rapid,bou2018coupling,dwivedi2018log}.
In particular, in relation to HMC, the papers \cite{mangoubi2017rapid,bou2018coupling} use coupling techniques to obtain convergence
rates in terms of Wasserstein or total variation distances, but these usually leverage independent momentum refreshment to obtain a Markov process in the location components only.
We establish here these convergence rates in weighted Wasserstein distance using coupling ideas, and also
in $L^2$ using \textit{hypocoercivity}; see, e.g., \cite{dric2009hypocoercivity,pavliotis2014stochastic}.
The rates we provide may generally not be the optimal ones for specific scenarios. However, the optimal rates for a specific scenario can be obtained by solving a multivariate optimisation problem. \citet{dolbeault2015hypocoercivity} also uses hypocoercivity, albeit with a much different flavour, and does not seem to provide explicit rates. After the first version of the present paper appeared online, the approach of \cite{dolbeault2015hypocoercivity} was extended in \citet{andrieu2018hypercoercivity} to cover several PDMPs, including BPS, Zig-Zag and RHMC. Even more recently, the paper \cite{lu2020explicit} appeared online, proving $L^2$ rates for three PDMPs (BPS, Zig-Zag and RHMC).

The approach in  \cite{dolbeault2015hypocoercivity} and \cite{andrieu2018hypercoercivity} is quite distinct to ours.
In particular \cite{andrieu2018hypercoercivity} also obtain dimension-free bounds for RHMC under similar assumptions; their explicit rates have a complex dependency on various parameters of the problem and therefore a detailed comparison with the explicit rates in our Theorem~\ref{thm:hypoco} was not performed in \cite{andrieu2018hypercoercivity}. In Remark \ref{rem:RHMCcomparison} we perform a comparison, and find that in the strongly convex and smooth setting, neither of these two approaches outperforms the other in all cases, sometimes the bound of \cite{andrieu2018hypercoercivity} is sharper, while in other scenarios our bound is sharper. Their approach is quite general but much less direct for RHMC than ours, as they rely on generic results by Dolbeault, Mouhot, and Schmeiser.
The approach in \cite{lu2020explicit} is entirely different from \cite{andrieu2018hypercoercivity} and ours, using sophisticated PDE methods to analyse the Fokker-Planck equations of the PDMP directly. In Remark \ref{rem:RHMCcomparison}, we include a detailed comparison with our results. In general, we find that the bounds in \cite{lu2020explicit} for RHMC are sharper than ours in the condition number $M/m$, but the constant of proportionality is not explicitly stated, and might be non-trivial to obtain reasonably small constants.

In addition the bounds of \cite{andrieu2018hypercoercivity} and \cite{lu2020explicit} for BPS suggest that its computational cost scales like $O(d^2)$. This seems to capture the worst case scenario and agrees for example with results \cite{bierkens2018high} for the log-density of the target, which recommends scaling the refreshment rate with the dimension. Our results suggest that when one is interested in low-dimensional projections, then it is computationally more efficient to not scale the refreshment rate with the dimension, achieving computational cost of order $O(d^{3/2})$. Empirical results in  Section \ref{sec:empirical} seem to suggest that this may also be the case for certain classes of functions depending on all the coordinates, such as the sum of all coordinates. A common scenario where this type of scaling limit is extremely relevant is for example that of Bayesian inference where typically one may only be interested in estimating the posterior means, variances and covariances of the high-dimensional state components (this is a set of one and two dimensional marginals). Finally, it is intuitively clear that the log-density will not mix well in a high-dimensional target for the global BPS, see \cite{bierkens2018high} for a detailed study. We conjecture that the functions that exhibit this type of behaviour form a low-dimensional sub-space of $L^2(\pi)$.  Recently \cite{bierkens2019spectral} has obtained very detailed results on the whole spectrum of the one-dimensional Zig-Zag process, it would be interesting if similar results could be obtained for BPS in high dimensional scenarios.

Apart from the intrinsic interest of the RHMC process, our motivation for studying its convergence rates is as follows. In the scaling literature for MCMC the limiting processes are usually Langevin diffusions. These have very well understood convergence rates which, at least under additional assumptions, are dimension-free. Therefore, in high-dimensions the cost of running the (time-rescaled) algorithm serves as a proxy for its computational complexity. In our case, the algorithm ran on its natural time scale converges to RHMC, which
as we establish here, also enjoys dimension-free convergence rates under appropriate assumptions. Therefore the cost of running BPS for a unit of process time serves as a proxy
for its algorithmic complexity.

The next section contains the statements of the main results of the paper along with necessary notation and definitions.
The remaining sections contain the proofs of the main results.

\section{Main results}
\subsection{Notation}
For $x\in \mathbb{R}$, let $x_+=\max\{x,0\}$.
Let $k\geq 1$. For vectors $u,v\in \mathbb{R}^k$ we write $|v|$ and $(u, v)$ for the Euclidean norm and inner product respectively. For matrices $A, B \in \mathbb{R}^{k\times k}$ we write
$A\preceq B$ if $B-A$ is positive-definite. For a function $f:\mathbb{R}^k \mapsto \mathbb{R}$ we write $\nabla f, \nabla^2 f$ for its (weak) gradient and Hessian respectively.
When considering functions $f=f(a,b)$, where $a,b\in\mathbb{R}^k$, that is $f:\mathbb{R}^{2k}\mapsto \mathbb{R}$, we will write $\nabla_a f$, $\nabla_b f$ to denote the gradient with respect to the variables $a\in\mathbb{R}^{k}$ and $b\in\mathbb{R}^{k}$  respectively.
Allowing a slight abuse of notation, for vector valued functions $f:\mathbb{R}^{d}\to \mathbb{R}^k$, we will also write $\nabla f$ for the Jacobian matrix of derivatives.

For $\ZZ=\mathbb{R}^k$, with $k\in \mathbb{N}$, let $C_0(\ZZ)$ denote the space of continuous functions $f:\mathcal{Z}\mapsto \mathbb{R}$ that vanish at infinity.
Recall that $C_0(\ZZ)$ is a Banach space with respect to the $\|\cdot\|_\infty$ norm, which is defined as usual through $\|f\|_\infty = \sup |f|$.
Also let $C_c^\infty (\ZZ)$ be the space of infinitely differentiable functions $f:\ZZ\mapsto \mathbb{R}$ with compact support.

For a measure $\pi$ on $\ZZ$, we will write $L^2(\pi)$ for the usual, real Hilbert space, and
$\langle \cdot,\cdot\rangle, \| \cdot\|$
to denote the inner product and norm in $L^2(\pi)$ respectively, whereas $L_0^2(\pi)$ will denote the orthogonal complement of the constant functions, i.e., functions with mean zero under the distribution $\pi$. Finally for $f:\mathcal{Z}\to \mathbb{R}^d$ and $g:\mathcal{Z}\to \mathbb{R}^d$, with $d\geq 1$,  we also write
$$\langle f, g\rangle = \int \pi(\rd z) ( f(z), g(z) ).$$
It will be clear from the context whether $\langle \cdot, \cdot\rangle$ is applied to $\mathbb{R}-$ or $\mathbb{R}^d$-valued functions.
We also define
$$H^1:= H^1(\pi):= \left\{h\in L_0^2(\pi): \nabla_x h, \nabla_v h\in L^2(\pi) \right\},$$
the Sobolev space of centred functions in $L^2(\pi)$ with weak derivatives in $L^2(\pi)$ and for $f,g\in H^1(\pi)$
we will denote the inner product and norm on $H^1(\pi)$ with $\lann \cdot ,\cdot\rann_{H^1(\pi)}$ and $\|\cdot \|_{H^1(\pi)}$ respectively, where
$$\lann f, g\rann_{H^1(\pi)}= \langle \nabla_x f, \nabla_x g\rangle+ \langle \nabla_v f, \nabla_v g\rangle.$$

\subsection{The Bouncy Particle Sampler}
Let $\ZZ:=\mathbb{R}\times \mathbb{R}$ and for $n\geq1$, define the Borel probability measure $\pi_n(\rd \zz)$ on $\ZZ^n$ with density w.r.t.\ Lebesgue measure given by
$$\pi_n(\zz) = \pi_n(\xx,\vv)\propto \exp\left\{ -U_n(\xx) - |\vv|^2/2\right\},\quad (\bm{x},\bm{v}) \in \ZZ^n,$$
where $U_n:\mathbb{R}^{n}\mapsto \mathbb{R}_+$ is a potential.

For $(\bm{x},\bm{v}) \in \ZZ^n$, define
\begin{equation}\label{eq:bounceoperator}
R_n(\bm{x})\bm{v} := \bm{v} - 2 \frac{( \nabla U_n(\bm{x}), \bm{v} )}{|\nabla U_n(\bm{x})|^2}\nabla U_n(\bm{x}).
\end{equation}
The vector $R_n(\bm{x})\bm{v}$ can be interpreted as a Newtonian collision on the hyperplane orthogonal to the gradient of the potential $U_n$, hence the interpretation of $\bm{x}$ as a position, and $\bm{v}$, as a velocity.

The Bouncy Particle Sampler (BPS), first introduced in \cite{peters2012rejection} and in a statistical context in \cite{bouncy2018}, defines a ${\pi_n}$-invariant, non-reversible, piecewise deterministic Markov process (PDMP)
$\{\ZZZ_n(t):t\geq 0\}=\{(\bm{X}(t), \bm{V}(t)): t\geq 0\}$ taking values in $\mathcal{Z}^n$ whose generator
$\mathcal{A}_n$, for smooth enough functions $f:\mathcal{Z}^n \mapsto \mathbb{R}$, is given by
\begin{align*}
 \AAA_n f(\bm{x},\bm{v})
 &= ( \nabla f(\bm{x}, \bm{v}), \bm{v} ) +  \max\{0, ( \nabla U_n(\bm{x}), \bm{v})\}
 \left[\mathfrak{R}_nf\left(\bm{x},\bm{v}\right) - f\left(\bm{x},\bm{v}\right)  \right]\\
 &\qquad
 + \lref \left[ Q_{\alpha,n} f \left(\bm{x},\bm{v}\right)-f\left(\bm{x},\bm{v}\right)\right],
 \end{align*}
where
\begin{equation*}
 \mathfrak{R}_n f\left(\bm{x},\bm{v}\right)
 := f\left(\bm{x},R_n(\bm{x})\bm{v}\right),
 \quad Q_{\alpha,n} f\left(\bm{x},\bm{v}\right)
 := \frac{1}{(2\pi)^{n/2}}\int_{\mathbb{R}^{n}} \re^{-|\boldsymbol{\xi}|^2/2} f\left(\bm{x},\alpha\bm{v}+\sqrt{1-\alpha^2} \boldsymbol{\xi}\right) \rd \boldsymbol{\xi},
\end{equation*}
for $0\leq\alpha<1$ and a positive refreshment rate $\lref>0$.
We also write $\ZZZ_n(t)=\left(Z_n^{(1)}(t), \dots, Z_n^{(n)}(t)\right)$ where $Z_n^{(k)}(t)=(X_n^{(k)}(t), V_n^{(k)}(t)) \in \ZZ$ is the $k$-th component.
The original formulation of the BPS algorithm corresponds to $\alpha=0$, that is refreshment occurs independently. The generalization $\alpha > 0$ \cite{vanetti2017piecewise} consists in refreshments that are performed according to an auto-regressive process.

\subsection{Randomized Hamiltonian Monte Carlo}
We define here RHMC as this is the process we will obtain as the weak limit of  $Z_n^{(1)}(t)=(X_n^{(1)}(t), V_n^{(1)}(t)) \in \ZZ$ as $n\to\infty$.    Define the Hamiltonian
\begin{equation}\label{eq:Hamiltonian}
H(x,v)=W(x) + |v|^2/2,
\end{equation}
for $(x,v)\in \ZZ$ and the corresponding probability density on $\ZZ$
\begin{equation}\label{targetRHMC}
\pi(x,v)=\bar{\pi}(x) \cdot \psi(v)  \propto \exp\{-W(x)-|v|^2/2\}.
\end{equation}
The \textit{Hamiltonian dynamics} associated to (\ref{eq:Hamiltonian}) is an ordinary differential equation in $\ZZ$ of drift $(\nabla_v H,-\nabla_x H)=(v,-\nabla W)$.
The RHMC process, denoted $\{Z_t:t \geq 0\}$, can then be defined following \citet[Section~24]{D_93}, as a PDMP with deterministic dynamics given by Hamiltonian dynamics with respect to $H$, fixed jump rate $\lref>0$ and jump kernel
\begin{equation}
Q_\alpha f(x,v):= \frac{1}{(2\pi)^{n/2}}\int \re^{-|{\xi}|^{2}/2}
f\left(x,\alpha v+\sqrt{1-\alpha^2} \xi\right) \rd \xi,
\end{equation} for some $0\leq\alpha<1$.
We will write $\{P^t: t\geq 0\}$ for the semi-group corresponding to $\{Z_t: t\geq 0\}$, that is
$$P^t f(z) = \rE\left[ \left. f(Z_t) \right| Z_0=z \right].$$
It has been shown, \cite{RHMC},  that RHMC admits $\pi$ as an invariant distribution.

It can also be shown that for $f\in C^{\infty}_c (\ZZ)$, the generator of the semigroup $\{P^t: t \geq 0\}$ is given by
\begin{equation}\label{eq:RHMC_generator}
\mathcal{A}f(x,v) = ( \nabla_x f, v) - ( \nabla_v f, \nabla W)
+ \lref \left[Q_\alpha f(x,v) - f(x,v)\right].
\end{equation}

The refreshment is done in an auto-regressive manner. From now on, we will restrict ourselves for BPS and RHMC to $0<\alpha<1$. The reason for using $\alpha>0$ is that it allows us to establish the Feller property which greatly simplifies the rest of the proofs.
 Since the autoregressive process mixes exponentially fast there is no loss in terms of mixing potentially at the cost of more frequent refreshments, something which has also been observed empirically.
\begin{remark}
As one of the referees kindly suggested, one may attempt to couple the process with $\alpha=0$ with the process at $\alpha_n=o(1)$ in order to extend the result to the case $\alpha=0$. Unfortunately, the obvious line of attack requires one to couple the full $n$-dimensional velocity vector at refreshments, so the maximal coupling deteriorates with the growing dimension; this approach would require a quantitative version of Theorem~\ref{thm:weakconv}. It is possible that a different coupling can be used, but we did not pursue this issue further.
\end{remark}

\subsection{Main results}
\subsubsection{RHMC as Scaling Limit of BPS}
Before stating our weak convergence result, we will make some assumptions.
We consider a sequence of targets $\pi_n$ on $\mathbb{R}^n\times \mathbb{R}^n$ where $\pi_n(\bm{x}, \bm{v}) = \bar{\pi}_n(\bm{x})  \psi_n(\bm{v})$, with $\psi_n$ a standard $n$-dimensional Gaussian and $\bar{\pi}_n(\bm{x}) = \exp\left[ - U_n(\bm{x}) \right]$ for a sequence of potentials $U_n: \mathbb{R}^n \to [0,\infty)$ satisfying the following assumptions.
\begin{assumption}\label{ass:potential}
The potential $U_n \in C^2(\mathbb{R}^n)$ is $m$-strongly convex with $M$-Lipschitz gradient
\begin{equation}
mI \preceq \nabla^2 U_n (\bm{x}) \preceq M I, \qquad x \in \mathbb{R}^n, \text{ with }0<m\le M<\infty,
\end{equation}
and $U_n$ achieves its minimum at 0, that is $U_n(0)=0$ and $\nabla U_n(0)=0$.
\end{assumption}
\begin{assumption}\label{ass:marginal}
The marginal density of the first component of $\bar{\pi}_n$ is fixed and is given by
$$f(x) := \int \bar{\pi}_n(x, \bm{x}_{2:n}) \rd \bm{x}_{2:n}.$$
We assume that $f(x) = \exp[- W(x)]$ for a potential  $W\in C^{\infty}(\mathbb{R}; [0,\infty))$ such that $\lim_{|x|\to\infty} W(x)=\infty$ and
$$\int \re^{-W(x)} \left( |W''(x)| + |W'(x)|^2\right) \rd x < \infty.$$
\end{assumption}
Let $\{Z_t: t \geq 0\}$ be the RHMC process with potential $W$ and write $\mathcal{A}$ for its  generator given in \eqref{eq:RHMC_generator}.
The following theorem is our first main result.
\begin{theorem}\label{thm:weakconv}
Suppose Assumptions~\ref{ass:potential} and \ref{ass:marginal} hold,  $0<\alpha<1$, $\lref>0$ and that the BPS process  $\{\ZZZ_n(t):t\geq 0\}$  is initialized
at stationarity, i.e., $\ZZZ_n(0)\sim \pi_n$. Then the process $\{Z_n^{(1)}(t):t\geq 0\}$ corresponding to the first location and velocity components of the BPS process converges weakly to the RHMC process  $\{Z_t:t\geq 0\}$ as $n\to\infty$.
\end{theorem}

We would like to stress that there is no time-rescaling in the above result, and that the sequence of targets is not assumed to factorise into independent components, or to converge towards an infinite dimensional measure as the dimension $n\to \infty$.
\begin{remark}
Notice that Assumption~\ref{ass:potential} allows for the standard scenario where the target factorises in $n$ i.i.d.\ copies which corresponds to
$U_n(\bm{x}) = \sum_{i=1}^n U(x_i)$, for an $m$-strongly convex potential $U\in C^2(\mathbb{R})$ with $U''\leq M$. Indeed in this case the Hessian matrix is diagonal and given by
$\left(\nabla^2 U_n(\bm{x})\right)_{i,j} = U''(x_i) \delta_{i,j} \geq 0$. This was the scenario considered in an earlier version of the present paper. In fact in this i.i.d.\ scenario the convexity assumption can be removed and the upper bound on $U''$ can be replaced by an upper bound on $U^{(k)}$ for any $k$, at the expense of additional technical complexity.
\end{remark}

\begin{remark}
From the proof (in particular, the bounds \eqref{eq:bndRn}, \eqref{eq:bndeps1}, \eqref{eq:bndeps21}, \eqref{eq:bndeps22}, \eqref{eq:bndeps3}) it is clear that the result remains true when $m, M$ in Assumption~\ref{ass:potential} are allowed to depend on $n$, if in addition we assume that
\begin{equation}
 m_n n \to \infty,\quad \frac{M_n}{m_n}= o(n^{1/4}),\quad  \frac{M_n^3}{m_n^{3/2}}=o(n^{1/2}),\quad \frac{M_n^3}{m_n^{2}}=o(n^{1/2}),\quad  \frac{M_n^2}{m_n}=o\left(\frac{n^{1/2}}{(\log(n))^{1/2}}\right).
\end{equation}
\end{remark}

\begin{remark}
Scaling limits for non i.i.d.\ targets have appeared in the past. \citet{bedard2007weak} studied targets that factorise into independent, but not identically distributed components; results on hierarchical targets can be found in \cite{bedard2019hierarchical, yang2019optimal} and references therein. The case of Gibbs measures with finite range interactions was studied in \cite{breyer2000metropolis}. \citet{mattingly2012diffusion} proved that a sequence of algorithms targeting finite dimensional projections of a measure admitting a density with respect to a reference Gaussian measure on a Hilbert space, converge to a Hilbert space-valued stochastic differential equation.
\end{remark}

\begin{remark}
		To illustrate Theorem~\ref{thm:weakconv} in Figure \ref{fig:weakconv} we have plotted the paths of the BPS process and the equi-energy contours of the Hamiltonian corresponding to the deterministic dynamics of RHMC.
		The target distribution has potential $U(\bm{x})=\sum_{i=1}^n |x_i|^b/2$ and we have tested two values of $b$, $b=2$ (Gaussian) and $b=4$. 	These figures show the first coordinate of the position and velocity vectors. As we can see, as the dimension increases, the paths of BPS indeed appear more and more similar to the contours of the Hamiltonian.
\begin{figure}
	\centering
	\begin{subfigure}[t]{0.49\textwidth}
		\raisebox{-\height}{\includegraphics[width=\textwidth]{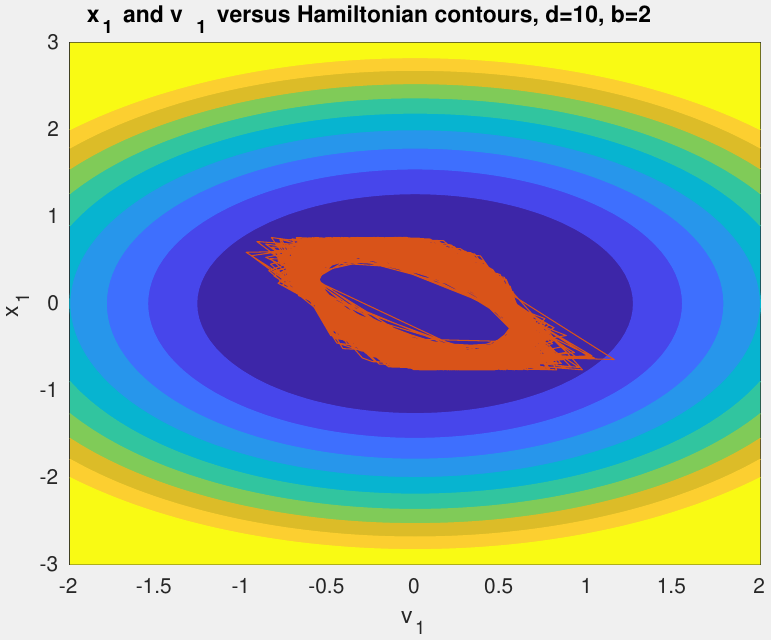}}
		\caption{$b=2$, $d=10$}
	\end{subfigure}
	\hfill
	\begin{subfigure}[t]{0.49\textwidth}
		\raisebox{-\height}{\includegraphics[width=\textwidth]{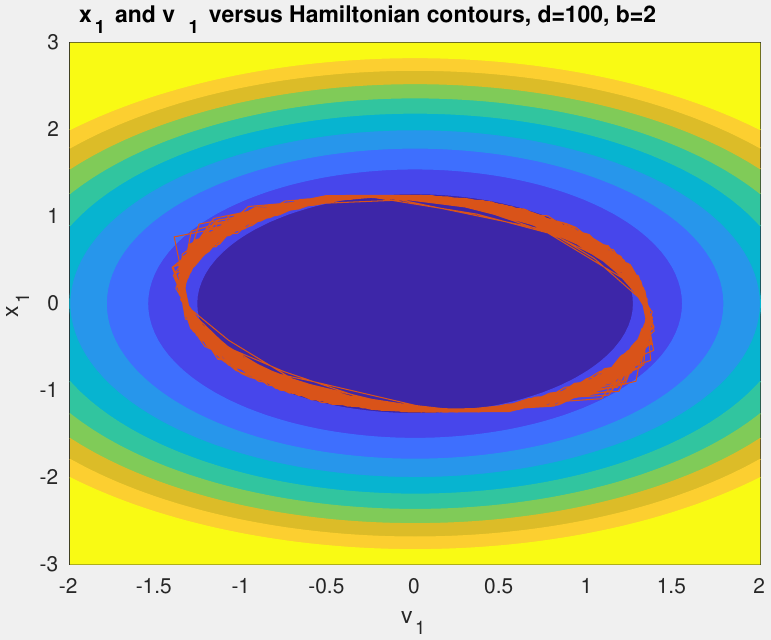}}
		\caption{$b=2$, $d=100$}
	\end{subfigure}
	\begin{subfigure}[t]{0.49\textwidth}
		\raisebox{-\height}{\includegraphics[width=\textwidth]{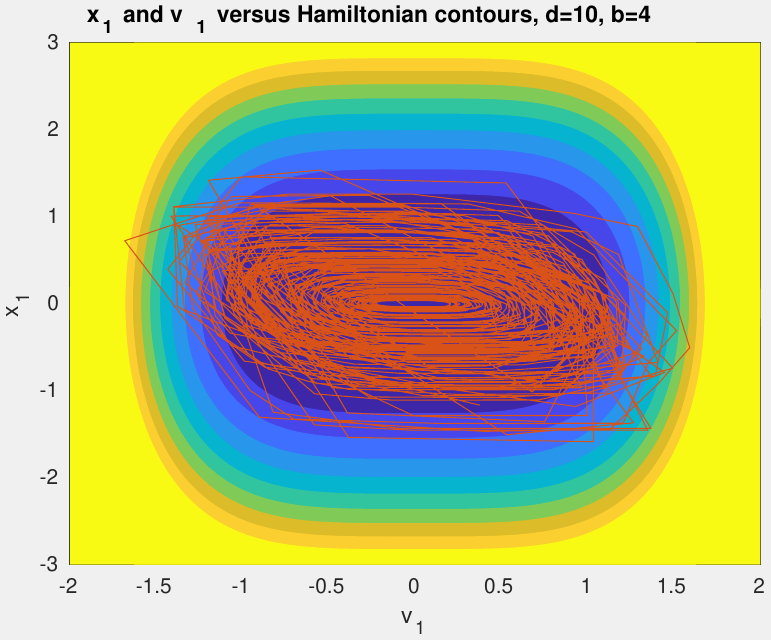}}
		\caption{$b=4$, $d=10$}
	\end{subfigure}
	\hfill
	\begin{subfigure}[t]{0.49\textwidth}
		\raisebox{-\height}{\includegraphics[width=\textwidth]{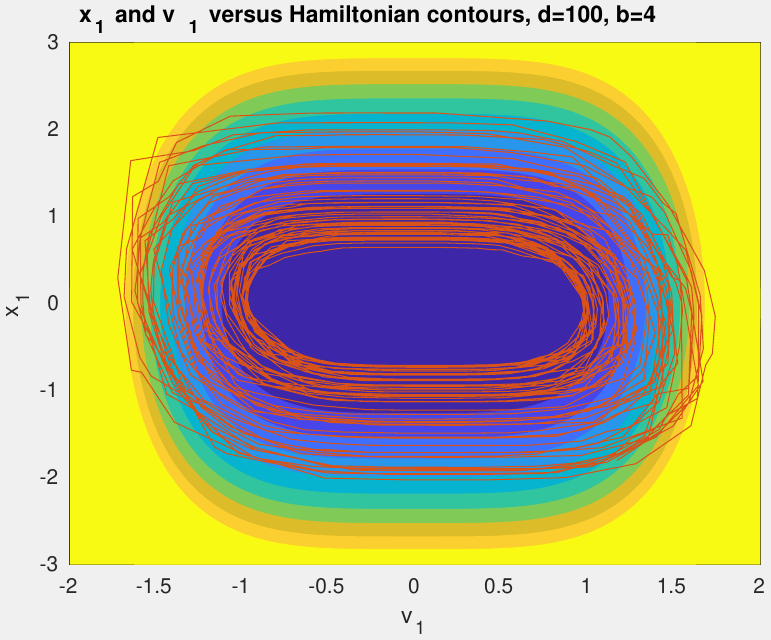}}
		\caption{$b=4$, $d=100$}
	\end{subfigure}
	\caption{ Convergence of the BPS process to RHMC in high dimensions for $U(x)=|x|^b/2$.}
	\label{fig:weakconv}
\end{figure}			
				
\end{remark}
\begin{remark}
Theorem~\ref{thm:weakconv} can be straightforwardly extended to any fixed, finite number of coordinates $d>1$. In this case the limiting process will be RHMC in $\mathbb{R}^d\times \mathbb{R}^d$ with respect to the potential $W:\mathbb{R}^d \to \mathbb{R}$ given by
$$W(\bm{x}) = -\log \int \bar{\pi}_n(\bm{x}, \bm{x}_{d+1:n}) \rd \bm{x}_{d+1:n}, \quad \bm{x}\in \mathbb{R}^d,$$
with $W$ satisfying a $d$-dimensional version of Assumption~\ref{ass:marginal}.
%
\end{remark}
\paragraph{Sketch of Proof.}
The full proof of this result is quite lengthy and will be given in Section~\ref{sec:weakconvproof}. However, we now give the key idea without going into technical details, for the simpler i.i.d.\ scenario where $U_n(\bm{x}) = \sum_{i=1}^n U(x_i)$, for $U:\mathbb{R}\mapsto [0,\infty)$. In this case the limiting process has potential $W\equiv U$.
Under the assumptions of Theorem~\ref{thm:weakconv} let
 $\bm{Z}_n = (\bm{X}_n, \bm{V}_n) \sim \pi_n$  and let $f:\mathbb{R}\times \mathbb{R} \to \mathbb{R}$ be smooth.
We now consider the generator $\mathcal{A}_n$ of BPS targeting $\pi_n$ and the generator $\mathcal{A}$ of RHMC targeting $\pi$, the marginal of the first location momentum pair under $\pi_n$, applied to the function $f$.
By inspecting $\mathcal{A}_n(f)$, $\mathcal{A}(f)$ we find that the terms corresponding to the deterministic flow of BPS and the refreshment events coincide exactly with corresponding terms in $\mathcal{A}(f)$. We therefore only have to consider the term corresponding to the ``bounce events'', that is
$$\max\left\{0, \left( \nabla U_n(\bm{X}), \bm{V}\right) \right\}
\left[ f\left(X_1, V_1 - 2\frac{\left( \nabla U_n(\bm{X}), \bm{V}\right)}{\left|\nabla U_n\left(\bm{X}\right)\right|^2} U'(X_1) \right) - f(X_1, V_1) \right],$$
and show that on average it is close to $-\left(\nabla_v f, \nabla U\right) = - U'(X_1) \partial_v f(X_1, V_1)$.

To see why this is true, after a Taylor expansion we can see that the bounce part of the BPS generator is close to
$$- 2\max\left\{0, \frac{\left( \nabla U_n(\bm{X}), \bm{V}\right)}{\left|\nabla U_n\left(\bm{X}\right)\right|} \right\}  \frac{\left( \nabla U_n(\bm{X}), \bm{V}\right)}{\left|\nabla U_n\left(\bm{X}\right)\right|} \partial_v f(X_1, V_1) U'(X_1).$$
Looking closer one can see that
$$\frac{\left( \nabla U_n(\bm{X}), \bm{V}\right)}{\left|\nabla U_n\left(\bm{X}\right)\right|}
= \frac{\sum_{i=1}^n U'(X_i) V_i}{\sqrt{\sum_{i=1}^n U'(X_i)^2}},$$
and since the $(V_i)_i$ are i.i.d.\ standard Gaussians it is easily seen that
$$\left.\frac{\sum_{i=1}^n U'(X_i) V_i}{\sqrt{\sum_{i=1}^n U'(X_i)^2}}\right| (X_i)_{i=1}^n \sim \mathcal{N}(0,1).$$
It now seems plausible that, letting $\xi \sim \mathcal{N}(0,1)$, we have
\begin{multline*}
\rE\left\{\left. \max\left\{0, \left( \nabla U_n(\bm{X}), \bm{V}\right) \right\}
\left[ f\left(X_1, V_1 - 2\frac{\left( \nabla U_n(\bm{X}), \bm{V}\right)}{\left|\nabla U_n\left(\bm{X}\right)\right|^2} U'(X_1) \right) - f(X_1, V_1) \right] \right| X_1, V_1 \right\}
\\\approx- 2\rE\left[\max\left\{0, \xi \right\}  \xi \right]\partial_v f(X_1, V_1) U'(X_1)
= -\partial_v f(X_1, V_1)  U'(X_1) = - ( \nabla_v f, \nabla U).
\end{multline*}
\subsubsection{Dimension-free Convergence Rates for RHMC}
We consider the RHMC process on the target
$$\pi(\bm{x},\bm{v})=\bar{\pi}(\bm{x}) \cdot \psi(\bm{v})  \propto \exp\{-U(\bm{x})-|\bm{v}|^2/2\},$$
 defined on  $\ZZ:=\mathbb{R}^d\times \mathbb{R}^d$ for $\bar{\pi}(\cdot)$ a strongly log-concave target distribution on $\mathbb{R}^d$  having a potential with bounded Hessian. This is a standard assumption adopted in  \cite{bou2018coupling,mangoubi2017rapid,dalalyan2017theoretical,dwivedi2018log,durmus2017nonasymptotic}.
\begin{assumption}\label{ass:hypoco}
Assume that  $U\in C^2(\mathbb{R}^d)$ and that for some $0<m<M$, and all $\bm{x}, \bm{v}\in \mathbb{R}^d$
\begin{equation}
m ( \bm{v}, \bm{v}) \leq ( \bm{v}, \nabla^2 U(\bm{x}) \bm{v}) \leq M ( \bm{v}, \bm{v}).
\end{equation}
\end{assumption}

The following proposition, whose proof is given in Appendix A,  shows that the expected number of bounces per unit time for BPS in stationary distribution is $O(\sqrt{d})$.
\begin{proposition}\label{proposition:BPSexpectednbofbounces}
Suppose that $\bar{\pi}(\bm{x})\propto \exp(-U(\bm{x}))$ is a probability density on $\mathbb{R}^d$.  Then the BPS process on $\ZZ$ targeting $\bar{\pi}\otimes\psi$ and   initialized 
at stationarity, has the following expected number of bounces per unit time: 
\[\Lambda_{b}:=\rE_{X\sim \pi, V\sim N(0,\mathbf{I}_d)}\l[ ( \grad U(\bm{X}),\bm{V})_+\r],\]
for any choice of refreshment rate $\lref$ and auto-regressive parameter $\alpha$. Moreover, if $\bar{\pi}$ satisfies Assumption \ref{ass:hypoco}, then we have
\begin{equation}
\frac{\sqrt{m (d-1/2)}}{\sqrt{2\pi}}\le \Lambda_{b}\le \frac{\sqrt{M d}}{\sqrt{2\pi}}.
\end{equation}
\end{proposition}

\paragraph{Wasserstein distance.}
For $t\ge 0$, let $Z^{(1)}(t)=(X^{(1)}(t),V^{(1)}(t))$ denote a path of the RHMC process. We couple this with another path $Z^{(1)}(t)=(X^{(2)}(t), V^{(2)}(t))$ such that their refreshments happen simultaneously and the same multivariate normal random variables are used for updating their velocities. Therefore
the difference between the paths $Z^{(1)}(\cdot)$ and $Z^{(2)}(\cdot)$ stems only from the different initialisations.
Then the coupled process $\left(Z^{(1)}(t),Z^{(2)}(t)\right)$ is Markov and we write $L_{1,2}$ for the corresponding generator.
Notice that the $2\times 2$ real valued matrix
\begin{equation}
A:=\l(\begin{matrix}a &b\\ b &c\end{matrix}\r), \label{eq:matrixA}
\end{equation}
is positive definite, denoted $A\succeq 0$, if and only if $a>0$, $c>0$ and $b^2<ac$.  For such a matrix, let
\begin{align*}&d_A^2(Z_1(t),Z_2(t)):=\\
&a \|X^{(2)}(t)-X^{(1)}(t)\|^2 + 2 b \inner{X^{(2)}(t)-X^{(1)}(t)}{V^{(2)}(t)-V^{(1)}(t)}+c\|V^{(2)}(t)-V^{(1)}(t)\|^2
\end{align*}
denote a distance function called weighted distance. It is equivalent up to constant multiplicative factors to the standard Euclidean distance on $\mathbb{R}^{2d}$ and the standard Euclidean distance corresponds to the special case $a=1$, $b=0$, $c=1$. However, due to the effect of the generator $L_{1,2}$ on $d_A^2(Z_1(t),Z_2(t))$, it will never be a contraction when $b=0$, and thus weighting this distance is essential for obtaining convergence rates.
Note that for every $p\ge 1$, the $W_p$-Wasserstein distance of two distributions $\nu_1,\nu_2$ on $\mathbb{R}^{2d}$ is defined as $W_{p}(\nu_1,\nu_2)=(\inf_{X_1\sim \nu_1,X_2\sim \nu_2}\rE(|X_1-X_2|^p))^{1/p}$, where the infimum is taken over all couplings with marginals $\nu_1$ and $\nu_2$.

Our main result in this section is the following.
\begin{theorem}\label{thm:Wasserstein}
Suppose that $0\le \alpha<1$, Assumption~\ref{ass:hypoco} holds and let
$$\lref=\frac{1}{1-\alpha^2}\l(2\sqrt{M+m}-\frac{(1-\alpha) m}{\sqrt{M+m}}\r), \qquad \mu=\frac{(1+\alpha)m}{\sqrt{M+m}}-\frac{\alpha m^{3/2}}{2(M+m)}.$$
Then there exist constants $a$, $b$ and $c$ depending on $m$, $M$ and $\alpha$, stated explicitly in \eqref{eq:Wassabcdef}, such that the corresponding matrix $A$ is positive definite, and for any $t \geq 0$ we have

\begin{equation}\label{eq:Wassersteincontraction1}
L_{1,2} \ d_A^2(Z_1(t),Z_2(t))\le -\mu\cdot d_A^2(Z_1(t),Z_2(t)).
\end{equation}
This directly implies that for any initial distribution $\nu$ on $\mathbb{R}^{2n}$, for all $t\geq 0$, we have the following bounds on the 2-Wasserstein distance to the stationary distribution,
\begin{equation}\label{eqWass2conv}
W_2(P^t \nu, \pi)^2 \leq C_2 e^{-\mu t} W_2(\nu, \pi)^2,
\end{equation}
for $C_2=\frac{a+c+\sqrt{(a+c)^2-4(ac-b^2)}}{a+c-\sqrt{(a+c)^2-4(ac-b^2)}}$. Moreover, for every $f\in L_0^{2}(\pi)$, for all $t\geq 0$
\begin{equation}\label{eqL2bndWass}
\|P^t f \|^2 \leq \min(C e^{-\mu t},1) \|f\|^2,
\end{equation}
where $C=\frac{ac+b^2+2\sqrt{ac b^2}}{ac-b^2}$.
\end{theorem}
\begin{remark}
Due to the non-reversibility of RHMC, the convergence rates in Wasserstein distance do not directly imply bounds on the asymptotic variance for every function in $L^2(\pi)$,
but only for Lipschitz functions. The argument for extending this contraction rate to all of $L^2(\pi)$,
can be found in the second half of the proof of Theorem~\ref{thm:hypoco}. This is based on the fact that Lipschitz functions are dense in $L^2(\pi)$.
\end{remark}
\begin{remark}\label{rem:optimalpha}
These results seem to suggest that choosing $\alpha$ close to 1 increases the convergence rate $\mu$ approximately by a factor of 2, at the expense of a higher refreshment rate. Hence in practice some tradeoff needs to be made between additional computational cost and the increased convergence rate. By Proposition \ref{proposition:BPSexpectednbofbounces}, we know that the rate of bounces according to the stationary distribution is at least $\frac{\sqrt{m (d-1/2)}}{\sqrt{2\pi}}$, which will be significantly higher than the rate $\frac{1}{1-\alpha^2}\l(2\sqrt{M+m}-\frac{(1-\alpha) m}{\sqrt{M+m}}\r)$ in high dimensions, provided that  $\sqrt{\frac{M+m}{m} \cdot \frac{1}{d-1/2}} \cdot \frac{1}{1-\alpha^2}\ll 1$. The choice $\alpha=0.9$ is reasonable in most scenarios.
\end{remark}
\begin{remark}
We have been able to verify using Mathematica that if $M/m\ge 5$, and we choose $\lref \le \frac{1}{2} \cdot \frac{1}{1-\alpha^2}\l(2\sqrt{M+m}-\frac{(1-\alpha) m}{\sqrt{M+m}}\r)$ (half  the value recommended in Theorem \ref{thm:Wasserstein}), then the contraction \eqref{eq:Wassersteincontraction1} cannot hold for any choice of $a$, $b$ and $c$. In general, if we choose $\lref=\frac{r}{1-\alpha^2}\l(2\sqrt{M+m}-\frac{(1-\alpha) m}{\sqrt{M+m}}\r)$ for some $r>1$ (that is, $r$ times the refreshment rate recommended in Theorem \ref{thm:Wasserstein}), then it seems based on extensive experiments that the rate
$\mu=\frac{1}{r}\l(\frac{(1+\alpha)m}{\sqrt{M+m}}-\frac{\alpha m^{3/2}}{2(M+m)}\r)$ is attained (i.e.\ $\mu$ drops by a factor $r$); no values of $a$, $b$ and $c$ result in double the same rate. Obtaining a formula that describes sharp rates $\mu$ for a general choice of $\lref$ seems difficult with our method of proof, as the inequalities that need to be checked in this case depend on many variables, and the calculations become intractable. We include in the electronic supplementary material Mathematica code that checks, for given values of $m, M, \alpha, \lref, \mu$, whether there exist $a$, $b$ and $c$ such that \eqref{eqL2bndWass} holds, and returns a possible choice of these parameters if they exist.
\end{remark}

As we shall see in the next proposition, it is possible to obtain faster convergence rates, that is larger $\mu$, for Gaussian target distributions. For this result, we consider a \textit{weighted} distance of the form
\begin{equation}
d_D^2(Z_1(t),Z_2(t)):=\left\langle Z_2(t)-Z_1(t), D (Z_2(t)-Z_1(t))\right\rangle,
\end{equation}
where $D$ is a real valued $2d\times 2d$ positive definite matrix.
\begin{proposition}\label{prop:WassersteinGauss}
Suppose that $\bar{\pi}$ is Gaussian and its inverse covariance matrix $H$ satisfies $m I\preceq H\preceq MI$.
Let $$\lref=\frac{2\sqrt{m}}{1-\alpha}, \qquad \mu=\frac{\sqrt{m}}{3}.$$
Then there exists a $2d\times 2d$ real valued matrix $D$ such that for any $t \geq 0$ we have
\begin{equation}
L_{1,2} \ d_D^2(Z_1(t),Z_2(t))\le -\mu\cdot d_D^2(Z_1(t),Z_2(t)).
\end{equation}
Moreover, for every $f\in L_0^{2}(\pi)$, we have
\begin{equation}\label{eqL2bndGaussian}
\|P^t f \|^2 \leq \min(C e^{-\mu t},1) \|f\|^2,
\end{equation}
where $C=\frac{ac+b^2+2\sqrt{ac b^2}}{ac-b^2}$.
\end{proposition}

\paragraph{Hypocoercivity.}
Our next convergence result is based on the hypocoercivity approach; see, e.g., \cite{monmarche2014hypocoercive,nier2005hypoelliptic,dric2009hypocoercivity,dolbeault2015hypocoercivity,roussel2017spectral}.
Our result will be stated in terms of the modified Sobolev norm $\lann h , h\rann^{1/2}$, where
\begin{equation}\label{eq:newnorm}
\lann h, h\rann:=a\|\nabla_v h\|^2 - 2b \, \langle \nabla_x h, \nabla_v h \rangle + c\|\nabla_x h\|^2,
\end{equation}
which again for $a,c>0$ and $b^2<ac$ defines a norm equivalent to the $H^1$ norm.
In particular following the calculations in \cite{dric2009hypocoercivity}, by Young's inequality we get
\begin{align*}
\left( 1+ \frac{|b|}{\sqrt{ac}}\right) \left[ a\|\nabla_v h\|^2+c\|\nabla_x h\|^2\right]\ge \lann h, h\rann
&\geq \left( 1- \frac{|b|}{\sqrt{ac}}\right) \left[ a\|\nabla_v h\|^2+c\|\nabla_x h\|^2\right].
\end{align*}
By the Efron-Stein-Steele inequality (\cite{steele1986efron}) and the fact that $\pi(x,v)=\overline{\pi}(x)\psi(v)$ is the product of two independent distributions, we have
\begin{align*}
\|h\|^2=\var_{\pi}(h)\le \var_{\psi}(\rE_{\overline{\pi}}(h))+\var_{\overline{\pi}}(\rE_{\psi}(h)),
\end{align*}
for any $h \in L_0^2(\pi)$.
Now by using the Poincar\'e inequality (\cite{BrascampLieb}) and the strong log-concavity of the distributions $\overline{\pi}$ and $\psi$, it is not difficult to show that
\[a\|\nabla_v h\|^2+c\|\nabla_x h\|^2\ge a\cdot 1 \cdot \var_{\psi}(\rE_{\overline{\pi}}(h))+c\cdot m\cdot  \var_{\overline{\pi}}(\rE_{\psi}(h))\ge \min(a, cm) \|h\|^2.
\]
Therefore convergence in the $\lann \cdot, \cdot \rann$ norm implies convergence in $L_0^2(\pi)$.

%

\begin{theorem}\label{thm:hypoco}
Suppose that Assumption~\ref{ass:hypoco} holds and let $\alpha \in [0, 1)$ and
$$\lref=\frac{1}{1-\alpha^2}\l(2\sqrt{M+m}-\frac{(1-\alpha) m}{\sqrt{M+m}}\r), \qquad
\mu=\frac{(1+\alpha)m}{\sqrt{M+m}}-\frac{\alpha m^{3/2}}{2(M+m)}.$$
Then there exist constants $a,b,c$ depending on $m$, $M$ and $\alpha$ such that $a>0, c>0, b^2<a c$, and for every  $f \in \mathcal{D}({B})\subset H^1(\pi)\subset L_0^2(\pi)$, with $B$, $\mathcal{D}({B})$ as defined in \eqref{eq:definitionofB},
\begin{equation}\frac{\rd }{\rd t} \lann P^t f, P^t f\rann\le -\mu \lann P^t f, P^t f\rann.
\label{eq:hypoco_rate}
\end{equation}
Moreover,  for every $f\in L_0^{2}(\pi)$ and $t\ge 0$, we have
\begin{equation}
\|P^t f \|^2 \leq \min(C e^{-\mu t},1) \|f\|^2,
\end{equation}
where $C=\frac{ac+b^2+2\sqrt{ac b^2}}{ac-b^2}$.
\end{theorem}

\begin{remark}
Although \eqref{eq:hypoco_rate} only implies variance bounds for functions in $H^1$, we are able to extend this to functions in $L^2(\pi)$ in the second half of the proof of Theorem ~\ref{thm:hypoco}, given in Section~\ref{sec:h1tol2}. As our rates are the same as in Theorem \ref{thm:Wasserstein}, the optimal choice of $\alpha$ can be done as discussed in Remark \ref{rem:optimalpha}.

Since the first-coordinate process of BPS converges to RHMC, whose mixing we established above, in the natural time-scale
the computational cost of running BPS for one time unit serves as a proxy for its algorithmic complexity. This cost is proportional to the number of total events per time unit, including bounces and refreshments. Proposition \ref{proposition:BPSexpectednbofbounces} shows that the expected number of bounces per unit time under Assumption \ref{ass:hypoco} is at least $\frac{\sqrt{m(d-1/2)}}{2\sqrt{\pi}}$, which is much larger than the expected number of refreshments ($\lref$) if the refreshment rate is chosen as recommended by Theorems \ref{thm:Wasserstein} and \ref{thm:hypoco} (as long as $M/m\ll d$ and $\alpha$ is not too close to 1).  Therefore in these cases it is justified to choose $\lref$ in order to maximize the contraction rate $\mu$ of the limiting RHMC process.

Since each bounce has a computational cost of order $O(1)$ in terms of gradient evaluations, our results suggests that BPS scales like $O(d^{1/2})$ in gradient evaluations under our assumptions.  This is the scaling observed in the simulations presented in the next section.
\end{remark}
\begin{remark}\label{rem:RHMCcomparison}
We state here the rates for RHMC obtained by \cite{andrieu2018hypercoercivity} and \cite{lu2020explicit} under the same set of assumptions on the potential, i.e. $m I_d \preceq \grad^2 U(\mathbf{x})\preceq M I_d$. Both papers show $L^2$ bounds of the form
\[\|P^t f \| \leq C e^{-\mu t} \|f\| \text{ for every }f\in L^2_0(\pi).\]
The convergence rate $\mu$ in \cite{andrieu2018hypercoercivity} in this setting is shown to satisfy the inequality $\alpha(\epsilon_0)\le \mu\le 3\alpha(\epsilon_0)$. After some calculations with Mathematica, we were able to show that
\[ \frac{m^2}{30}\le \alpha(\epsilon_0)\le \frac{m^2}{5} \text{ for }0<m<1, \text{ and } 0.03\le \alpha(\epsilon_0)\le 0.11 \text{ for }m>1,\]
when the optimal choice of refreshment rate is chosen as
\[\lref^{\mathrm{opt}}=\frac{8-2\sqrt{2}+4\sqrt{3}}{\sqrt{2}}\approx 8.5583.\]
Assuming $\alpha=0$ (no autoregressive part in the velocity refreshments), our results yield
$$\mu=\frac{m}{\sqrt{M+m}}\text{ for the choice }\lref=2\sqrt{M+m}-\frac{m}{\sqrt{M+m}}, $$
We can see that for large values of $M/m$, the convergence rate of \cite{andrieu2018hypercoercivity} is sharper, while for smaller values, our rates are sharper. We note that the conditions in \cite{andrieu2018hypercoercivity} are quite general, and only require a Poincar\'e inequality, hence they are applicable even without strong convexity.
\cite{lu2020explicit} shows that for RHMC, the convergence rate is $\mu=\Theta(\frac{m\lref}{(\sqrt(m)+\lref)^2})$, which is maximized when $\lref=\Theta(\sqrt{m})$, yielding $\mu=\Theta(\sqrt{m})$. The dependence of these results on the parameters $m, M$ improves upon \cite{andrieu2018hypercoercivity} and our paper, but the constant of proportionality is not known.

In the case of BPS, both \cite{andrieu2018hypercoercivity} and \cite{lu2020explicit} shows rates of the form $\mu=\Theta(\sqrt{d})$. The dependence on the  parameters $m$ and $M$ is sharper in \cite{lu2020explicit} compared to \cite{andrieu2018hypercoercivity}, but the constant of proportionality is unknown. In contrast with these results, our high dimensional limit argument (Theorem \ref{thm:weakconv}) shows that for functions that only depend on a single coordinate (or on a fixed number of coordinates), in high dimensions, the convergence occurs according to a dimension independent rate $\mu$ as long as we choose the refreshment rate appropriately, at $\lref=\Theta(1)$. This is useful in particular for situations where we are interested in estimating the posterior mean.
\end{remark}

\subsection{Empirical results for different functions}\label{sec:empirical}
In this section, we show some simulation results about the computational cost of the BPS for a $d$ dimensional standard normal target, and seven different test functions defined as follows,
\begin{align*}
f_1(x)&=x_1 \quad (\text{first coordinate}),\\
f_2(x)&=\sum_{i=1}^{d}x_i \quad (\text{sum of all coordinates}),\\
f_3(x)&=\sum_{i=1}^{d-1}\sin(x_i+x_{i+1}) \quad (\text{a sum of sines depending on two component each}),\\
f_4(x)&= |x| \quad (\text{radius}),\\
f_5(x)&=\frac{|x|^2}{2}=\sum_{i=1}^{d}\frac{x_i^2}{2} \quad (\text{log-density}),\\
f_6(x)&=x_1^2 \quad (\text{square of first coordinate}),\\
f_7(x)&=x_1 x_2 \quad (\text{product of first and second coordinates}).
\end{align*}
In order to estimate the effective sample sizes, we have run 100 parallel BPS simulations with $10^6$ events per simulation,  starting from the Gaussian target distribution. The autoregressive parameter $\alpha$ was set as $\alpha=0$.
Figure \ref{fig:empirical} shows the number of events required for one  effective sample for dimensions $d=10$, $100$, $1000$ and $10000$ for these 7 functions, with refreshment parameter choices $\lref=1$ (as suggested by Theorems \ref{thm:Wasserstein} and \ref{thm:hypoco}) and $\lref=\sqrt{d}$ (as suggested by \cite{bierkens2018high} and Table 1  of \cite{andrieu2018hypercoercivity}). The number of events is a correct proxy for the computational cost as each event requires one gradient evaluation (see Section 2.3 of \cite{bouncy2018} for the description of the implementation of BPS for Gaussian targets). As we can see, if the refreshment rate is chosen as $\lref=1$, these simulation results show $O(\sqrt{d})$ scaling in the number of events required for an effective sample for all of the functions except the radius and the log-density ($f_4$ and $f_5$). In contrast, the choice $\lref=\sqrt{d}$ seems to require significantly more events per effective sample, with $O(d)$ scaling observed empirically. In the cases of the radius and the log-density, the choice $\lref=\sqrt{d}$ still seems to  require $O(d)$ events per effective sample, while $\lref=1$ is doing worse, approximately $O(d^{4/3})$ events per effective sample is required. The scaling limits for this function were studied in \cite{bierkens2018high}, who has recommended choosing $\lref=O(\sqrt{d})$ to obtain the best mixing for the log-density, consistently with our empirical results.

To sum up, we can see that if the goal of the simulation is to estimate the posterior mean or posterior covariance matrix, or other quantities only depending a small subset of the coordinates, then choosing $\lref$ as recommended by Theorems \ref{thm:Wasserstein} and \ref{thm:hypoco}  yield good empirical performance ($O(\sqrt{d})$ scaling in the number of events required for an effective sample). For functions depending on all of the coordinates the situation is more complicated, and the best choice of $\lref$ is strongly function dependent in this case.

	\begin{figure}
		\begin{subfigure}{.5\textwidth}
			\centering
			\includegraphics[width=.8\linewidth]{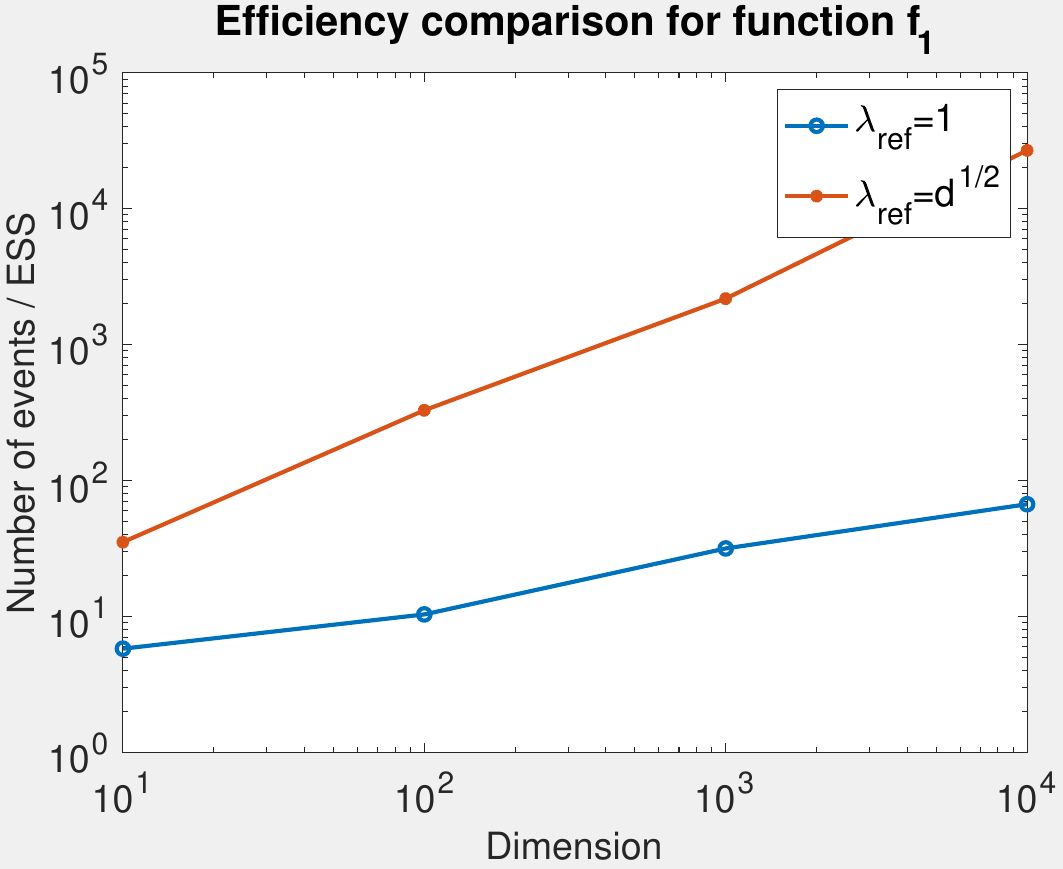}
		\end{subfigure}%
		\begin{subfigure}{.5\textwidth}
			\centering
			\includegraphics[width=.8\linewidth]{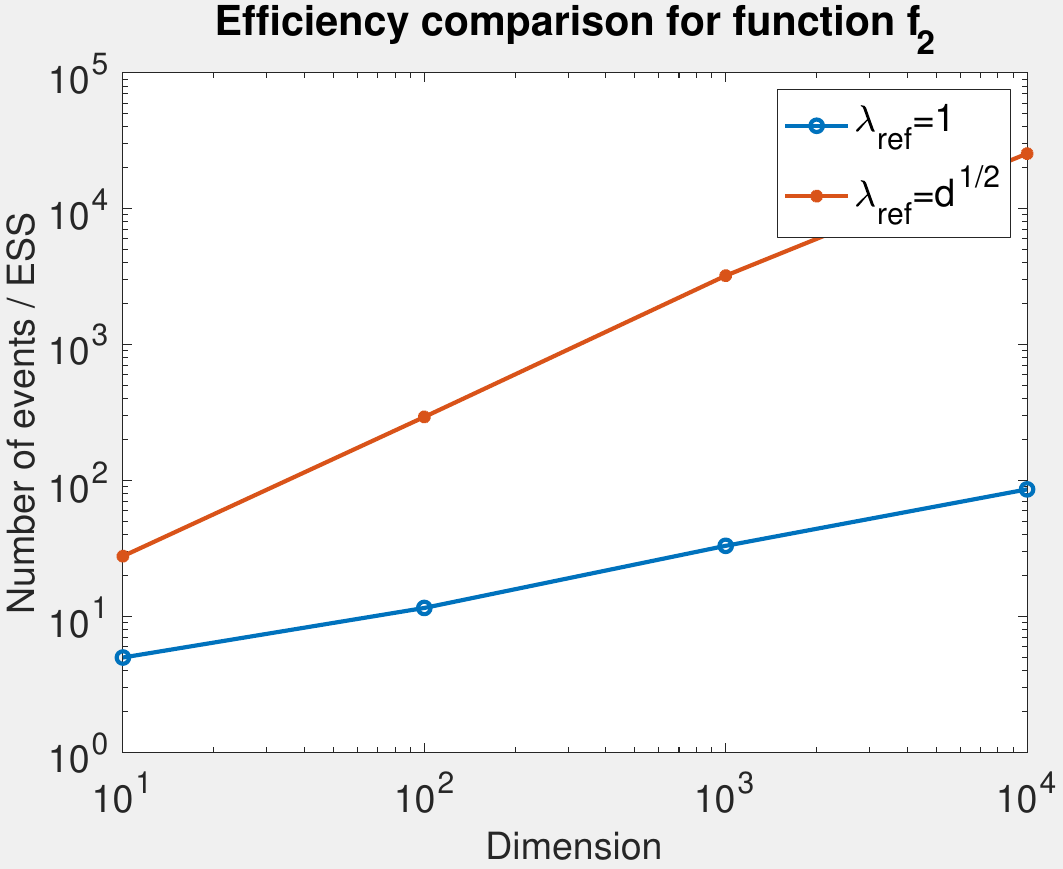}
		\end{subfigure}\\

		\begin{subfigure}{.5\textwidth}
		\centering
	\includegraphics[width=.8\linewidth]{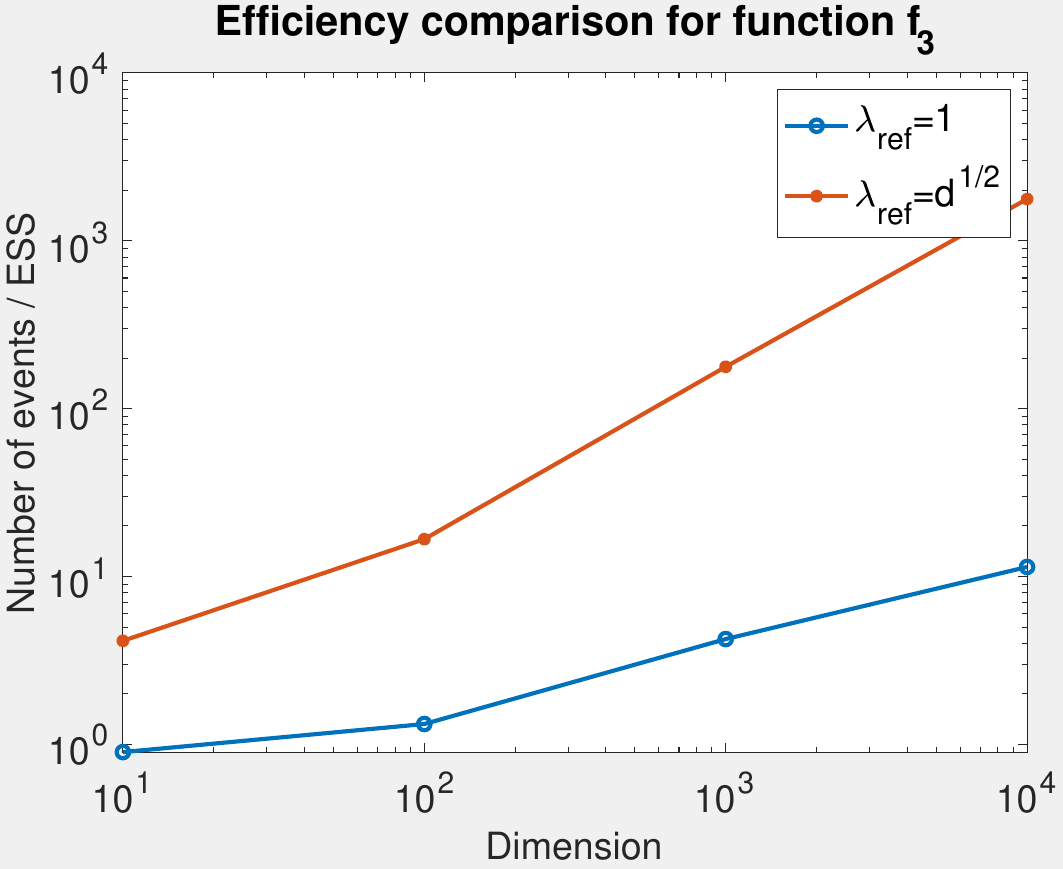}
\end{subfigure}%
\begin{subfigure}{.5\textwidth}
	\centering
	\includegraphics[width=.8\linewidth]{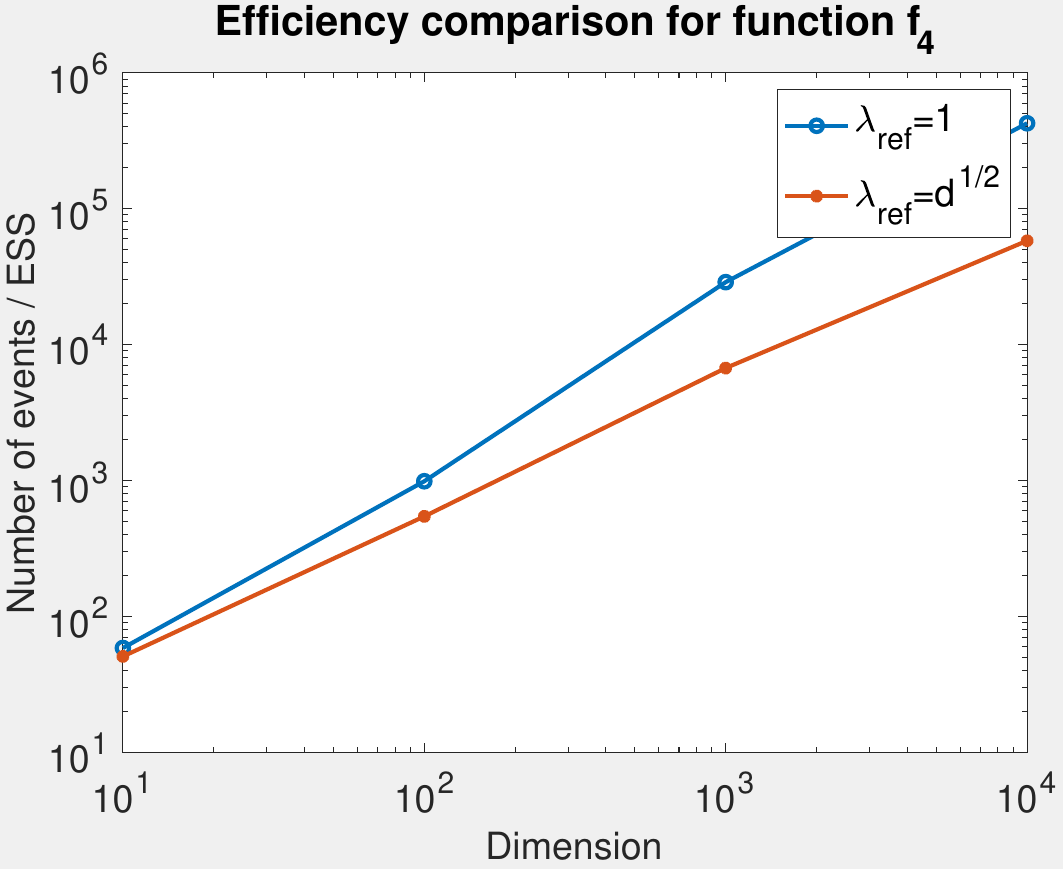}
\end{subfigure}\\

		\begin{subfigure}{.5\textwidth}
	\centering
	\includegraphics[width=.8\linewidth]{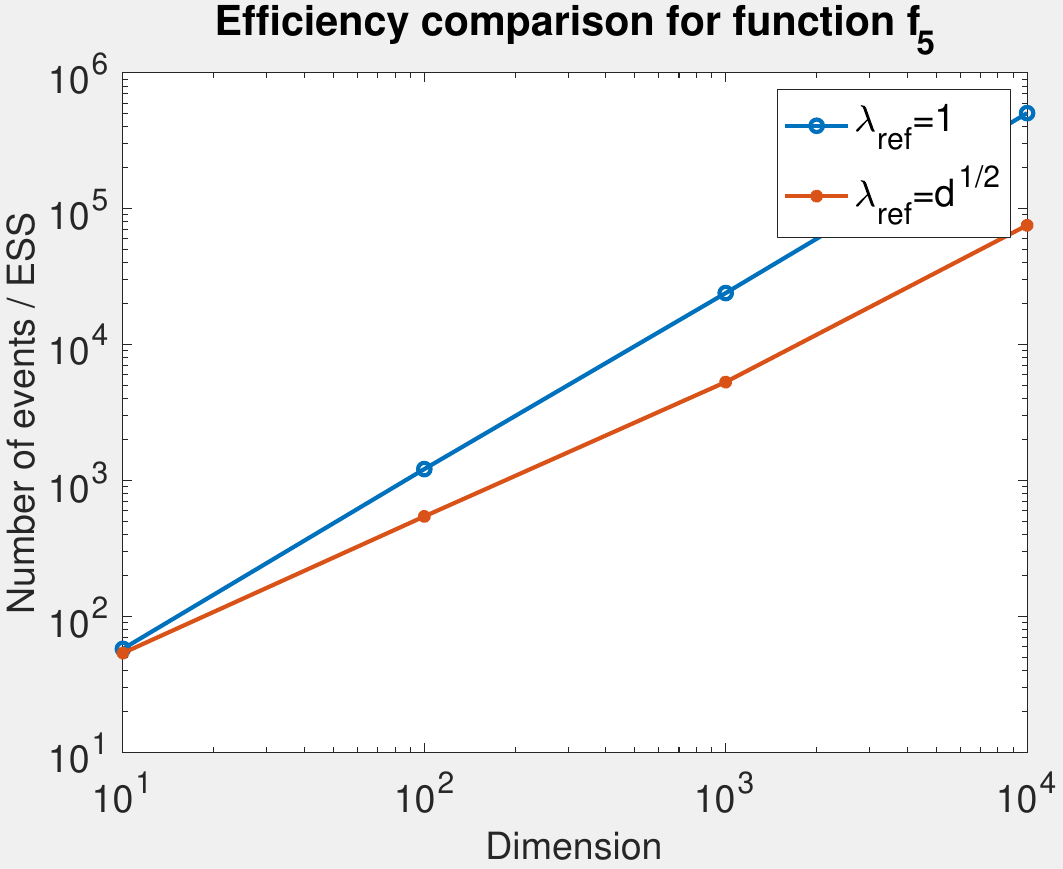}
\end{subfigure}%
\begin{subfigure}{.5\textwidth}
	\centering
	\includegraphics[width=.8\linewidth]{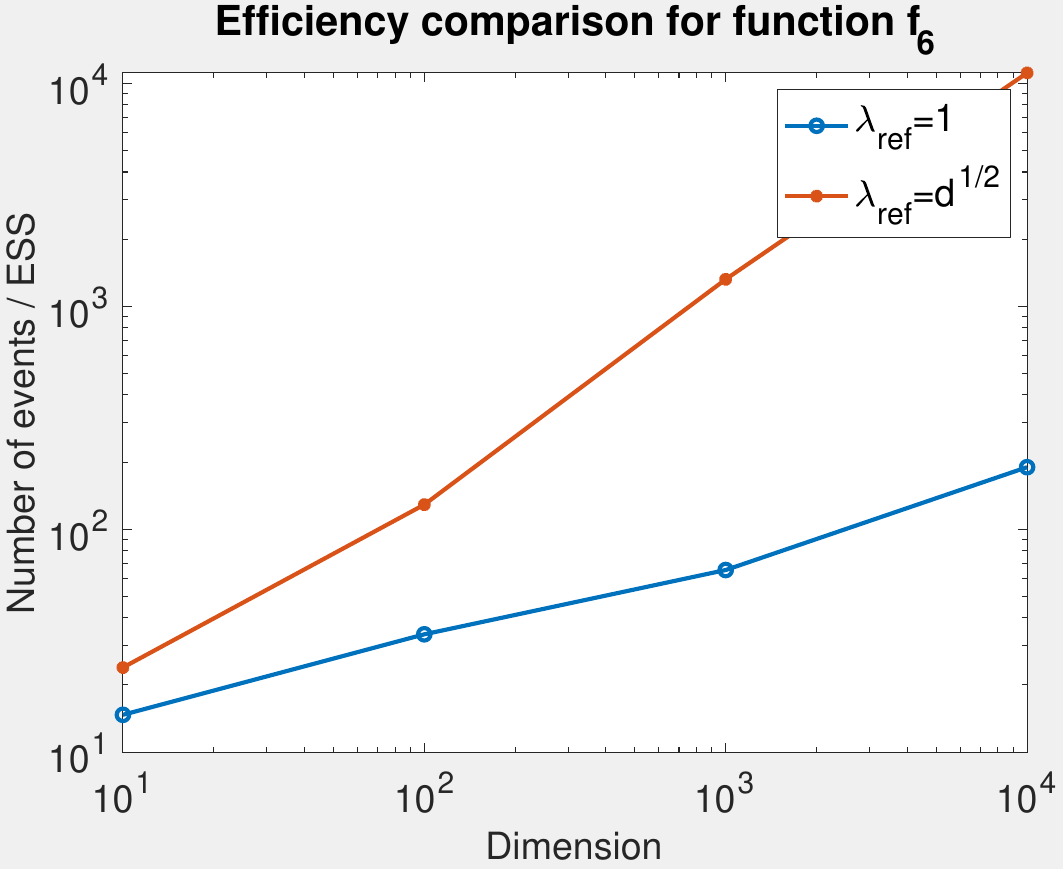}
\end{subfigure}\\

	\begin{subfigure}{.5\textwidth}
	\centering
	\includegraphics[width=.8\linewidth]{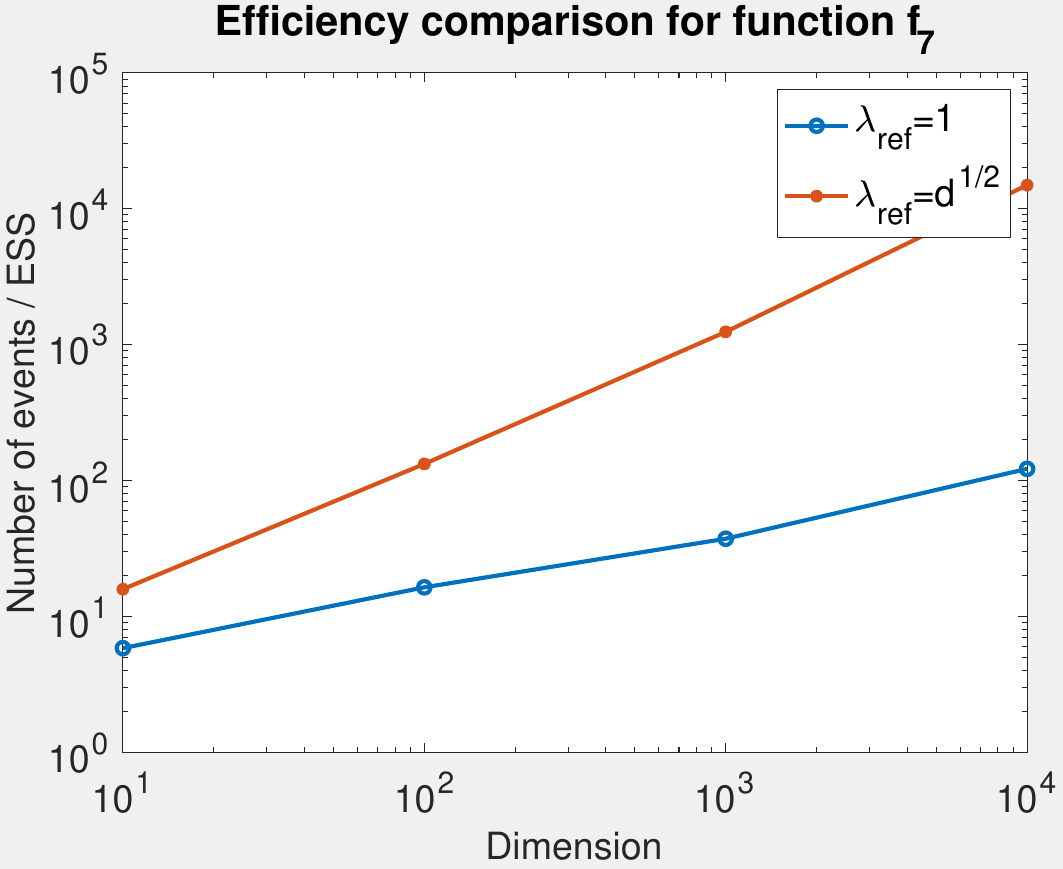}
\end{subfigure}%

		\caption{Number of BPS events per effective sample for 7 different functions for standard Gaussian target as a function of the dimension, with two different scalings of the refreshment rate $\lref$ in terms of the dimension}
		\label{fig:empirical}
	\end{figure}

\section{Proof of Weak Convergence Result - Theorem~\ref{thm:weakconv}\label{sec:weakconvproof}}
The proof will be based on a sequence of auxiliary results.
First we will show that the RHMC semigroup $\{P^t: t \geq 0\}$, acting on the Banach space $C_0(\mathcal{Z})$ with the sup-norm is Feller, and that the space $C_c^\infty (\mathcal{Z})$ is a \textit{core} for its generator given in \eqref{eq:RHMC_generator}, in the sense that
$C_c^\infty$ is dense in $\mathcal{D}(\mathcal{A})$ with respect to the norm $\vvvert \cdot \vvvert := \|f\|_\infty + \|\mathcal{A} f \|_\infty$.
This, and a sequence of auxiliary results, will allow us to
apply \cite[Corollary~8.6]{EK_05} to prove Theorem~\ref{thm:weakconv}.
\subsection{Feller property}

Recall that in the context of Theorem ~\ref{thm:weakconv}, we have  $d=1$ and $\ZZ=\mathbb{R}^{2}$. A Markov process taking values in $\ZZ$,
with transition semigroup $\{P^{t}:t\geq0\}$, is called a Feller
process and $\{P^t:t\geq 0\}$ a Feller semigroup, if it satisfies the following two properties
\begin{description}
\item [{Feller property:}] for all $t\geq0$ and $f\in C_{0}(\mathcal{Z})$ we have
$P^{t}f\in C_{0}(\mathcal{Z})$,  and
\item [{Strong continuity:}] $\|P^{t}f - f\|_\infty \to 0$ as $t\to0$ for $f\in C_{0}(\mathcal{Z})$.
\end{description}

\begin{proposition}\label{lem:feller}
Suppose that $W:\mathbb{R}\mapsto [0,\infty)$ is continuously differentiable and $\lim_{|x|\to\infty} W(x)=\infty$. Then the RHMC process $\{Z_t\}_{t\geq 0}$ with generator $\mathcal{A}$ given by \eqref{eq:RHMC_generator} with Hamiltonian $H(x,v) = W(x) + |v|^2/2$, $\alpha \in (0,1)$ and $\lref>0$ is a Feller process. If in addition $W\in \mathcal{C}^\infty(\mathbb{R})$, then $C_c^\infty (\mathbb{R})$ is a core for its generator.
\end{proposition}
Note that a more technical approach proposed recently in \citet{holderrieth2019core} requires weaker assumptions.

\subsubsection{Proof of Proposition~\ref{lem:feller}}
%
%
%
Before we proceed let us first define the \textit{resolvent operator} for $\lambda > 0$
$$\RR_\lambda f(z) :=  \int_0^\infty \re^{-\lambda s} P^s f(z) \rd s
=\int_0^\infty \re^{-\lambda s} \rE^z \left[ f(Z_s)\right] \rd s.$$

The proof will proceed as follows. First we will first show that
$\mathcal{R}_\lambda: C_0(\mathcal{Z}) \to C_0(\mathcal{Z})$,
and then use \cite[Corollary~1.23]{Bottcher13} to establish that $\{P^t:t\geq 0\}$ has the Feller property, that is for all $t\geq 0$ $P^t: C_0(\mathcal{Z}) \to C_0(\mathcal{Z})$.
Once the Feller property is established by \cite[Lemma~1.4]{Bottcher13} to prove strong continuity, it suffices to prove the weaker statement
$P^t f(z) \to f(z)$, for all $f\in C_0(\mathcal{Z})$ and $z \in \mathcal{Z}$.
We now establish this property.
Let $T_1, T_2, \dots$ be the arrival times of the jumps. Then we have for $h> 0$
\begin{align*}
P^h f(z) - f(z)
&= \rE^z \left[ f(Z_h)\right] - f(z)\\
&= \rE^z \left[ f(Z_h) \mathbb{1} \{T_1 \geq h \}\right] - f(z)
 + \rE^z \left[ f(Z_h) \mathbb{1} \{T_1 < h \}\right]\\
&= f\left( \Xi(h,z)\right)\re^{-\lref h}  - f(z)
 + \mathcal{E},
\end{align*}
where we write $\Xi(z,t)$ for the solution of the Hamiltonian dynamics at time $t$ initialized at $z_0=z$. It is well-known that if $H:\rR\times \rR \to \rR$ is continuously differentiable everywhere then $\Xi(z,s)$ is well defined for all $s>0$ (see for example \cite[Theorem~1.186]{Chi}), $H\big(\Xi(z,s)\big)=H(z)$ for all $s>0$ and
$\Xi(z,h) \to z$ as $h\to 0$.
Since $f$ is bounded it easily follows that as $h\to 0$
$$|\mathcal{E}| \leq \|f\|_\infty (1-\re^{-\lref h}) \to 0.$$
Since $\Xi(z,h)\to z$ as $h\to 0$, the result follows.

\paragraph{Proof of the Feller property.} From \cite[Equation~2.6]{DuCo_08} we know that we can express the resolvent kernel as follows for a measurable set $A$
\begin{equation}\label{eq:resolvent_sum}
\RR_\lambda(z,A) = \sum_{j=0}^\infty J_\lambda^j K_\lambda (z,A),
\end{equation}
where
\begin{align}
K_\lambda (z,A)
&:=\int_0^\infty \re^{-\lambda s -\lref s} \mathbb{1}_A \left( \Xi(z,s)\right) \rd s,\\
J_\lambda (z,A)
&:=\int_0^\infty \lref \, \re^{-\lambda s -\lref s} Q_\alpha\left(\Xi(z,s), A \right)
\rd s,
\end{align}
with $\Xi(z,s)=\Xi\big((x,v),s\big)$ as defined above.

We will now show that  $\RR_\lambda f \in C_0(\mathcal{Z})$ for any $f\in C_0(\mathcal{Z})$.
This follows from the next result.
%
\begin{lemma}\label{lem:resolvent}
$W\in C^{1}(\mathbb{R}; [0,\infty))$, $W(x)\to \infty$ as $|x|\to \infty$ and let $f\in C_0(\mathcal{Z})$. Then, for any $\lambda>0$, we have $J_\lambda f, K_\lambda f \in C_0(\mathcal{Z})$ and $\|J_\lambda f\|_\infty\leq \lref/(\lambda+ \lref)\|f\|_\infty$. In particular
$$\mathcal{R}_\lambda f = \sum_{j=0}^\infty J_\lambda^j K_\lambda f \in C_0(\mathcal{Z}).$$
 \end{lemma}
\begin{proof}[Proof of Lemma~\ref{lem:resolvent}]
Let $\lambda >0$ and let us first look at $K_\lambda$.
Suppose now that $f\in C_0(\mathcal{Z})$ and that $z_n\to z$. Then
\begin{align*}
|K_\lambda f(z) - K_\lambda f(z_n)|
&\leq \int_0^\infty \lref \, \re^{-\lambda s -\lref s} |f \left(\Xi(z,s)\right)-f \left(\Xi(z_n,s)\right)| \rd s \to 0,
\end{align*}
by the bounded convergence theorem, since $f$ is bounded and the functions
$s\mapsto |f \left(\Xi(z,s)\right)-f \left(\Xi(z_n,s)\right)|$ vanish pointwise by the continuity of $f$ and the continuous dependence of the solution $\{\Xi(z,s):s\geq 0\}$ on the initial condition; see, e.g., \cite[Theorem~1.3]{Chi}.
This establishes that $K_\lambda f$ is continuous.

Next we prove that $K_\lambda f$ vanishes at infinity. Let $\epsilon>0$ be arbitrary.
Since $W(x) \to \infty$ as $|x|\to \infty$, the level sets $\mathcal{H}_L:=\{z: H(z)\leq L\}$ are compact and $\ZZ = \cup_{L>0}\{z: H(z)\leq L\}$. Therefore we can find $L=L(\epsilon)$ such that
$|f(z)|<\epsilon(\lambda + \lref)$ for $z \notin \mathcal{H}_L$. For all such $z$, since $H(\Xi(z,s))=H(z)$ for all $s>0$,  we have that
\begin{align*}
|K_\lambda f(z)|
&\leq \int_0^\infty \re^{-\lambda s -\lref s} \left|f \left( \Xi(z,s)\right)\right| \rd s\\
&< \epsilon(\lambda + \lref) \int_0^\infty \re^{-\lambda s -\lref s}  \rd s = \epsilon.
\end{align*}
Thus we conclude that for all $\lambda >0$ we have $K_\lambda :C_0(\mathcal{Z}) \to C_0(\mathcal{Z})$.

Now we move on to $J_\lambda$.  First notice that for any $f\in C_0(\mathcal{Z})$ we have $Q_\alpha f$ is also continuous. To see why let $z_n=(x_n,v_n)\to z=(x,v)$
and notice that as $d=1$
\begin{multline*}
|Q_\alpha f \left(z_n\right)-Q_\alpha f \left(z\right)|
\\ \leq \frac{1}{\sqrt{2\pi}} \int_{-\infty}^\infty
	\left| f\left(x_n, \alpha v_n + \sqrt{1-\alpha^2} \xi\right)
	-f\left(x, \alpha v + \sqrt{1-\alpha^2} \xi\right) \right|\re^{-\xi^2/2} \rd \xi \to 0,
	\end{multline*}
by the bounded convergence theorem, since $f$ is continuous and bounded, and therefore $Q_\alpha f$ is continuous.
Next, for any $\delta>0$ we can choose a compact set $K_\delta$ such that $|f(z)|<\delta$ for $z\notin K_\delta$. In particular, since $K_\delta$ is compact, for any $\delta>0$ we can also find $M_\delta>0$ such that
$$K_\delta \subset \{(x,v): |x|, |v|\leq M_\delta\}.$$
Fix $\epsilon \in (0, 1/2)$ and choose $z_\epsilon$ such that $\Phi(z_{\epsilon})=1-\epsilon/2$, where $\Phi$ is the cumulative distribution function of the standard normal distribution. Then
\begin{align*}
|Q_\alpha f \left(z\right)|
&\leq \epsilon \|f\|_\infty
+\frac{1}{\sqrt{2\pi}} \int_{\xi=-z_{\epsilon}}^{z_\epsilon}
	\left| f\left(x, \alpha v + \sqrt{1-\alpha^2} \xi\right)
	\right|\re^{-\xi^2/2}\rd \xi.
\end{align*}
Then for all $z=(x,v)$ and $\xi$ such that $|x|>M_\epsilon$,  $|v|>(M_\epsilon+z_\epsilon)/\alpha$ and $|\xi|<z_\epsilon$ we have
\begin{align*}
\Big|\alpha v + \sqrt{1-\alpha^2}\xi\Big|
&\geq \alpha |v| - \sqrt{1-\alpha^2}|\xi|
\geq \alpha |v| -|\xi|\geq M_\epsilon + z_\epsilon - z_\epsilon > M_\epsilon.
\end{align*}
Therefore for such $z$ we have that
\begin{align*}
|Q_\alpha f \left(z\right)|
&\leq \epsilon \|f\|_\infty
+\frac{1}{\sqrt{2\pi}} \int_{\xi=-z_{\epsilon}}^{z_\epsilon}
	\left| f\left(x, \alpha v + \sqrt{1-\alpha^2} \xi\right)
	\right|\re^{-\xi^2/2}\rd \xi\\
&< \epsilon \|f\|_\infty
+\frac{\epsilon}{\sqrt{2\pi}} \int_{\xi=-z_{\epsilon}}^{z_\epsilon}\re^{-\xi^2/2}\rd \xi,
\end{align*}
and since $\epsilon>0$ is arbitrary it follows that $Q_\alpha f \in C_0(\mathcal{Z})$.

Observe that $J_\lambda f(z) =\lref K_\lambda Q_\alpha f(z)$. Therefore
if $f\in C_0(\mathcal{Z})$, since we have already shown that $Q_\alpha :C_0(\mathcal{Z})\to C_0(\mathcal{Z})$ and $K_\lambda :C_0(\mathcal{Z})\to C_0(\mathcal{Z})$, it follows that $J_\lambda f \in C_0(\mathcal{Z})$.
%

Finally, since clearly $\|Q_\alpha f\left( \Xi(z,s)\right)\|_\infty \leq \|f\|_\infty$
\begin{align*}
\|J_\lambda f\|_\infty
&= \sup_{z} \Big|
\int_0^\infty \lref \, \re^{-\lambda s -\lref s} Q_\alpha f \left(\Xi(z,s)\right) \rd s
\Big|\\
&\leq \int_0^\infty \lref \, \re^{-\lambda s -\lref s} \|Q_\alpha f\left(\Xi(z,s)\right)\|_\infty \rd s\\
&\leq \int_0^\infty \lref \, \re^{-\lambda s -\lref s} \|f\|_\infty \rd s
=\frac{\lref}{\lambda + \lref} \|f\|_\infty,
\end{align*}
and since $\lambda >0$ we can see that this is a strict contraction.
From this, it follows that the sequence
$$\sum_{j=0}^n J_\lambda^j K_\lambda f, $$
is Cauchy in the Banach space $\left(C_0(\mathcal{Z}), \| \cdot \|_\infty\right)$, whence the conclusion follows.
\end{proof}
\paragraph{$C_c^\infty$ is a core.}
Define the semigroup $\{Q^t:t \geq 0\}$, where for each $t\geq 0$
$Q^t: C_0(\mathcal{Z})\mapsto C_0(\mathcal{Z})$ is defined through
$Q^t f(z)=f\left(\Xi (z,t) \right)$, with $\Xi(z,t)$ denoting as before the solution of the Hamiltonian dynamics started from $z$ at time $t$.
It can be easily shown that the generator of $Q^t$ is given for $f\in C_c^\infty(\mathcal{Z})$ by
$$Bf(x,v) = ( \nabla_x f, v ) - ( \nabla_v f , \nabla U(x)),$$
that is the first two terms of the generator $\mathcal{A}$ of RHMC.

Let $f$ be supported on a compact set $K$. By our assumptions on the Hamiltonian $H$, there exists $L>0$ such that $K\subseteq \mathcal{H}_L:=\{(x,v): H(x,v) \leq L\}$. Letting $z \notin \mathcal{H}_L$,
for all $t\geq 0$, we have by definition $H(\Xi(z,t))=H(z)$ and thus $\Xi(z,t) \notin \mathcal{H}_L$. Therefore $Q^tf$ will have compact support.

Notice next, since $W\in C^\infty(\mathbb{R})$, that for any $t\geq 0$ the mapping $z\mapsto \Xi(z,t)$ is infinitely differentiable, see e.g. \cite[Exercise~1.185]{Chi}. From this and the above discussion we conclude that for any $f\in C_c^\infty (\mathcal{Z})$ and $t\geq 0$ we have $Q^t f \in C_c^\infty$.
Therefore from \citet[Theorem~1.9]{davies1980one}, and since $C_c^\infty (\mathcal{Z})\subset C_0 (\mathcal{Z})$ is dense, we conclude that $C_c^\infty$ is a core for $B$, and in particular that for any $f\in \mathcal{D}(B)$, there exists a sequence $\{f_n:n\geq 0\} \subset C_c^\infty (\mathcal{Z})$ such that
$$\| f_n -f \|_\infty + \|Bf_n - Bf \|_\infty \to 0.$$
Since the operator $\lref [Q_\alpha - I]$ is clearly bounded on $C_0(\mathcal{Z})$ for any $\alpha$, it follows that $\mathcal{D}(\mathcal{A})= \mathcal{D}(B)$, and that for the sequence $\{f_n\}$ above we also have
$$\| f_n -f \|_\infty + \|\mathcal{A}f_n - \mathcal{A} f \|_\infty \to 0,$$
proving that $C_c^\infty(\mathcal{Z})$ is a core for $\mathcal{A}$.

\subsection{Proof of Theorem~\ref{thm:weakconv}}
Recall that we write $\{\ZZZ_n(s):s\geq 0\}$ for BPS initialized from $\pi_n$, the generator of which we denote with $\mathcal{A}_n$, and write $\{Z^{(1)}_n(s):s\geq 0\}$ for its first component.
In addition let
$$\FF^n_t:=\sigma\{\ZZZ_n(s): s\leq t\}, \quad \text{and} \quad \GG^n_t:=\sigma\left\{Z^{(1)}_n(s): s\leq t\right\}.$$

Let $\epsilon_n\to 0$ be monotone and to be specified later on. All expectations will be with respect to the path measure of BPS started from $\pi_n$.
We proceed with the usual construction. For some function $f:\ZZ \to \mathbb{R}$, that is $f$ is a function only of $Z_n^{(1)}$, such that $f\in C^{\infty}_c$, smooth with compact support, we define
\begin{align}
 \xi_n(t)&:= \epsilon_n^{-1} \int_0^{\epsilon_n}\rE\left[ f\big( Z_n^{(1)}(t+s) \big) \big| \GG_t^n \right] \rd s,\label{eq:xidef}\\
 \phi_n(t)&:=\epsilon_n^{-1} \rE\left[\left. f\left( Z_n^{(1)}(t+\epsilon_n) \right) - f\left( Z_n^{(1)}(t) \right) \right| \GG_t^n \right].\label{eq:phidef}
\end{align}
Abusing notation, we will also write $f$ for the mapping $\mathcal{Z}^n\mapsto \mathbb{R}$ given by $f(z_1,\dots, z_n)=f(z_1)$.
We have already established that $(\mathcal{A}, C_c^\infty)$ generates the strongly continuous semigroup $\left\{ P^{t}:t\geq0\right\}$ corresponding to RHMC.
To apply \citep[Corollary~8.6 of Chapter~4]{EK_05} we need to check
the following:
\begin{itemize}
\item \textbf{Strongly Separating algebra:} the closure of the linear span of $C_c^\infty$
contains an algebra that strongly separates points, see \cite[Section~3.4]{EK_05} for the definition. This is obvious since $C_c^\infty(\mathcal{Z})$ strongly separates points and is dense in the algebra $C_0(\mathcal{Z})$,
since any function in $C_0(\mathcal{Z})$ can be approximated arbitrarily well by functions in $C_c(\mathcal{Z})$ by multiplying with, and then convolving with appropriate mollifiers.
\item \textbf{Generator~~convergence:} for each $f\in C_c^\infty(\mathcal{Z})$ and $T>0$, for $\xi_{n},\phi_{n}$ as defined in \eqref{eq:xidef},\eqref{eq:phidef}
\begin{align}
\sup_{n}\sup_{t\leq T}\mathbb{E}[|\xi_{n}(t)|] & <\infty\label{eq:firstcondition}\\
\sup_{n}\sup_{t\leq T}\mathbb{E}[|\phi_{n}(t)|] & <\infty\label{eq:secondcondition}\\
\lim_{n\to\infty}\mathbb{E}\left[\left|\xi_{n}(t)-f\left(Z_n^{(1)}(t)\right)\right|\right] & =0,\label{eq:thirdcondition}\\
\lim_{n\to\infty}\mathbb{E}\left[\left|\phi_{n}(t)-\mathcal{A}f\left(Z_n^{(1)}(t)\right)\right|\right] & =0,\label{eq:fourthcondition}
\end{align}
and in addition
\begin{equation}
\lim_{n\to\infty}\mathbb{E}\left\{ \sup_{t\in\mathbb{Q}\cap[0,T]}|\xi_{n}(t)-f(Z_n^{(1)}(t))|\right\} =0,\label{eq:833}
\end{equation}
and for some $p>1$
\begin{equation}
\sup_{n\to\infty}\mathbb{E}\left[\left(\int_{0}^{T}|\phi_{n}(s)|^{p}\mathrm{d}s\right)^{1/p}\right]<\infty.\label{eq:834}
\end{equation}
\end{itemize}
\subsubsection{Proof of Equations (\ref{eq:833}) and (\ref{eq:thirdcondition}).}
Since condition (\ref{eq:thirdcondition}) is implied by (\ref{eq:833}), we will establish \eqref{eq:833}.

{First recall that for each $n$, BPS is non-explosive. To see why, for each $\bm{x}, \bm{v}$, let $L>|v|>0$ and consider
$$\tau_{n, L} := \inf\{t\geq 0: \bm{Z}_n(t) \notin B(x, L^2)\times B(0, L)\}.$$
Letting
$$\sigma^x_{n, L^2}:= \inf\{t\geq 0: \bm{X}_n(t) \notin B(x, L^2)\}, \quad
\sigma^v_{n, L}:= \inf\{t\geq 0: \bm{V}_n(t) \notin B(0, L)\},
$$
we have
\begin{align*}
\tau_{n, L} &= \sigma^x _{n, L^2} \mathds{1}\{\sigma^x _{n,L^2} < \sigma^v_{n, L} \}
+ \sigma^v _{n, L} \mathds{1}\{\sigma^x _{n, L^2} \geq \sigma^v_{n, L} \}\\
&\geq L \mathds{1}\{\sigma^x _{n, L^2} < \sigma^v_{n, L} \}+ \sigma_{n, L}^v \mathds{1}\{\sigma^x _{n, L^2} \geq \sigma^v_{n, L} \}\geq L \vee \sigma_{n, L}^v,
\end{align*}
where the first inequality follows, since on the event $\{\sigma^x _{L^2} \geq \sigma^v_L \}$ the maximum speed up to $\sigma^x_{L^2}$ is less than $L$. Since $|\bm{V}_n(t)|$ only changes at the arrivals of a homogeneous Poisson process with rate $\lref>0$, it is clear that as $L\to \infty$, $\sigma^v_{n, L}\to \infty$ and therefore $\tau_{n, L} \to \infty$.
}

Fix $T>0$. Since BPS is non-explosive for every $n$ and $\delta>0$ we can find a $K_{n,\delta}>0$ such that
$$\rP\left[ \sup_{t\leq T+1} |\bm{Z}_n(t)| \geq K_{n,\delta} \right] \leq \delta.$$
For $\delta_n\to 0$ and by a diagonal argument, we can find a sequence $K_{n,\delta_n}$ such that
$$\rP\left[ \sup_{t\leq T+1}|\bm{Z}_n(t)| \geq K_{n,\delta_n} \right] \leq \delta_n \to 0.$$
We will write $G_n$ for the event
$$G_n := \left\{  \sup_{t\leq T+1}|\bm{Z}_n(t)| \leq K_{n,\delta_n} \right\}.$$
Then we have for $\epsilon_n \to 0$, to be specified later on,
\begin{align*}
\lefteqn{\rE\left[ \sup_{t\in[0,T] \cap \mathbb{Q}} \left| \xi_n(t) - f\left( Z_n^{(1)}(t) \right)\right| \right]}\\
&=\rE\left[\sup_{t\in[0,T] \cap \mathbb{Q}} \left| \epsilon_n^{-1}\int_0^{\epsilon_n}\rE\left[\left.   f\left( Z^{(1)}_n(t+r)\right)-  f\left( Z^{(1)}_n(t)\right)\right| \GG_t^n \right]\rd r \right| \right]\\
&=\rE\left[\sup_{t\in[0,T] \cap \mathbb{Q}} \left| \epsilon_n^{-1}\int_0^{\epsilon_n}\rE\left[\left. \rE\left\{\left.   f\left( Z^{(1)}_n(t+r)\right)-  f\left( Z^{(1)}_n(t)\right)
  \right| \FF_t^n \right\} \right| \GG_t^n \right]\rd r \right| \right]\\
&\leq\rE\left[\sup_{t\in[0,T] \cap \mathbb{Q}} \left| \epsilon_n^{-1}\int_0^{\epsilon_n}\rE\left[\left. \rE\left\{\left.  \left( f\left( Z^{(1)}_n(t+r)\right)-  f\left( Z^{(1)}_n(t)\right)\right) \mathbb{1}_{G_n}
  \right| \FF_t^n \right\} \right| \GG_t^n \right]\rd r \right| \right]\\
&\qquad +\rE\left[\sup_{t\in[0,T] \cap \mathbb{Q}} \left| \epsilon_n^{-1}\int_0^{\epsilon_n}\rE\left[\left. \rE\left\{\left.  \left( f\left( Z^{(1)}_n(t+r)\right)-  f\left( Z^{(1)}_n(t)\right)\right) \mathbb{1}_{G_n^\mathtt{c}}
  \right| \FF_t^n \right\} \right| \GG_t^n \right]\rd r \right| \right]\\
&:= J_1 + J_2.
\end{align*}
For the term $J_2$ we have for $p>1$
\begin{align}
J_2
&\leq 2\|f\|_\infty \rE\left[\sup_{t\in[0,T] \cap \mathbb{Q}} \rE\left[\left.  \mathbb{1}_{G_n^\mathtt{c}}
   \right| \GG_t^n \right]\right]
   \leq  2\|f\|_\infty\rE\left[ \left(\sup_{t\in[0,T] \cap \mathbb{Q}} \rE\left[\left.  \mathbb{1}_{G_n^\mathtt{c}}
   \right| \GG_t^n \right] \right)^p\right]^{1/p}\notag\\
&\leq 2\|f\|_\infty\frac{p}{p-1} \rE\left[ \rE\left[\left.  \mathbb{1}_{G_n^\mathtt{c}} \right| \GG_T^n \right]^p\right]^{1/p}
\leq 2\|f\|_\infty\frac{p}{p-1} \rE\left[ \mathbb{1}_{G_n^\mathtt{c}}^p\right]^{1/p} = 2\|f\|_\infty\frac{p}{p-1} \delta_n^{1/p},\label{eq:j2_doob}
\end{align}
where we used Jensen's inequality, the fact that for each $n$, $\left\{\rE[\mathbb{1}_{G_n^\mathtt{c}}\mid \mathcal{G}_t^n]: t\geq 0\right\}$ is a $\mathcal{G}_t^n$-martingale and Doob's martingale inequality.

We proceed with the term $J_1$ as follows
\begin{align*}
J_1
&\leq \rE\left[\sup_{t\in[0,T] \cap \mathbb{Q}} \left| \epsilon_n^{-1}\int_0^{\epsilon_n}\rE\left[\left. \rE\left\{\left.  \left[ f\left( Z^{(1)}_n(t+r)\right)-  f\left( Z^{(1)}_n(t)\right)\right] \mathbb{1}_{G_n}
\mathbb{1}\{\tau_1^\text{ref}(t)> \epsilon_n\} \right| \FF_t^n \right\} \right| \GG_t^n \right]\rd r \right| \right]\\
&\qquad + \rE\left[\sup_{t\in[0,T] \cap \mathbb{Q}} \left| \epsilon_n^{-1}\int_0^{\epsilon_n}\rE\left[\left. \rE\left\{\left.  \left[ f\left( Z^{(1)}_n(t+r)\right)-  f\left( Z^{(1)}_n(t)\right)\right] \mathbb{1}_{G_n}
\mathbb{1}\{\tau_1^\text{ref}(t)\leq  \epsilon_n\} \right| \FF_t^n \right\} \right| \GG_t^n \right]\rd r \right| \right]\\
&=: J_{1,1} + J_{1,2},
\end{align*}
where we denote by $\tau_1^\text{ref}(t)$ the first refreshment time after time $t$. Since refreshment happens independently we can bound $J_{1,2}$
\begin{align*}
 J_{1,2}
 &\leq 2\|f\|_\infty
 \rE\left[\sup_{t\in[0,T] \cap \mathbb{Q}} \left| \epsilon_n^{-1}\int_0^{\epsilon_n} (1-\re^{-\lref \epsilon_n})\rd r \right| \right]\leq 2\|f\|_\infty \lref\, \epsilon_n \to 0.
\end{align*}

We control the term $J_{1,1}$ in two steps. To keep notation short we introduce the notation $G_n'(t):= \{ \tau_1^\text{ref}(t)> \epsilon_n\}$. Then
\begin{align*}
J_{1,1}
&\leq
\rE\Bigg[\sup_{t\in[0,T] \cap \mathbb{Q}} \bigg| \epsilon_n^{-1}\int_0^{\epsilon_n}\rE\Big[ \rE\Big\{  \Big[ f\left( X^{(1)}_n(t+r), V^{(1)}_n(t+r)\right) \\
&\qquad\qquad -  f\left( X^{(1)}_n(t), V^{(1)}_n(t+r)\right)\Big] \mathbb{1}_{G_n} \mathbb{1}_{G'_n(t)} \Big| \FF_t^n \Big\} \Big| \GG_t^n \Big]\rd r \bigg| \Bigg]\\
&\qquad + \rE\Bigg[\sup_{t\in[0,T] \cap \mathbb{Q}} \bigg| \epsilon_n^{-1}\int_0^{\epsilon_n}\rE\Big[ \rE\Big\{  \Big[ f\left( X^{(1)}_n(t), V^{(1)}_n(t+r)\right) \\
&\qquad\qquad\qquad -  f\left( X^{(1)}_n(t), V^{(1)}_n(t)\right)\Big] \mathbb{1}_{G_n} \mathbb{1}_{G'_n(t)} \Big| \FF_t^n \Big\} \Big| \GG_t^n \Big]\rd r \bigg| \Bigg]=: J_{1,1,1}+ J_{1,1,2}.
\end{align*}
For the first term, since only the location component changes we have
\begin{align*}
J_{1,1,1}
&\leq \|\partial_x f\|_\infty
\rE\Bigg[\sup_{t\in[0,T] \cap \mathbb{Q}} \bigg| \epsilon_n^{-1}\int_0^{\epsilon_n}
\rE\Big[ \rE\Big\{  \Big| X^{(1)}_n(t+r) - X^{(1)}_n(t)\Big|\times \mathbb{1}_{G_n} \mathbb{1}_{G'_n(t)} \Big| \FF_t^n \Big\} \Big| \GG_t^n \Big]\rd r \bigg| \Bigg]\\
&\leq \|\partial_x f\|_\infty
\rE\Bigg[\sup_{t\in[0,T] \cap \mathbb{Q}} \bigg| \epsilon_n^{-1}\int_0^{\epsilon_n}
\rE\Big[ \epsilon_n \Big| V^{(1)}_n(t)\Big|\times \mathbb{1}_{G_n} \mathbb{1}_{G'_n(t)} \Big| \FF_t^n \Big\} \Big| \GG_t^n \Big]\rd r \bigg| \Bigg],
\end{align*}
where the second inequality follows from the linear dynamics of BPS, since on the event $G_n'(t)$ there is no refreshment event and therefore the norm of the velocity component does not change.
Finally, recalling the definition of the event $G_n$ we obtain
\begin{align*}
J_{1,1,1}
&\leq \|\partial_x f\|_\infty \epsilon_n K_{n,\delta_n}.
\end{align*}
Next we have to control the term $J_{1,1,2}$ for which we point out that, since there is no refreshment event, the velocity will remain constant on the interval $[t, t+\epsilon_n]$ unless there is a bounce.
Writing $\sigma_1(t)$ for the arrival time of the first bounce after time $t$ we thus have
\begin{align*}
&J_{1,1,2}\\
&=  \rE\Bigg[\sup_{t\in[0,T] \cap \mathbb{Q}} \bigg| \epsilon_n^{-1}\int_0^{\epsilon_n}\rE\Big[ \rE\Big\{  \Big[ f\left( X^{(1)}_n(t), V^{(1)}_n(t+r)\right) \\
&\qquad\qquad\qquad -  f\left( X^{(1)}_n(t), V^{(1)}_n(t)\right)\Big] \mathbb{1}\{\sigma_1(t)< \epsilon_n\}\mathbb{1}_{G_n} \mathbb{1}_{G'_n(t)} \Big| \FF_t^n \Big\} \Big| \GG_t^n \Big]\rd r \bigg| \Bigg]\\
&\leq 2\|f\|_\infty \rE\Bigg[\sup_{t\in[0,T] \cap \mathbb{Q}} \bigg| \epsilon_n^{-1}\int_0^{\epsilon_n}\rE\Big[ \rE\Big\{  \mathbb{1}\{\sigma_1(t)< \epsilon_n\}\mathbb{1}_{G_n} \mathbb{1}_{G'_n(t)} \Big| \FF_t^n \Big\} \Big| \GG_t^n \Big]\rd r \bigg| \Bigg]\\
&\leq 2\|f\|_\infty \rE\Bigg[\sup_{t\in[0,T] \cap \mathbb{Q}} \bigg| \epsilon_n^{-1}\int_0^{\epsilon_n}\rE\Big[ \rE\Big\{  \mathbb{1}\{\sigma_1(t)< \epsilon_n\} \Big| \FF_t^n \Big\} \Big| \GG_t^n \Big]\rd r \bigg| \Bigg]\\
&\leq 2\|f\|_\infty \rE\Bigg[\sup_{t\in[0,T] \cap \mathbb{Q}} \bigg| \epsilon_n^{-1}\int_0^{\epsilon_n}
\rE\Big[\rE\Big\{ \Big[ 1-\exp\Big(-\int_0^{\epsilon_n} ( \nabla U_n(\bm{X}_n(t+s)), \bm{V}_n(t+s) )_+
\rd s  \Big) \Big] \Big| \FF_t^n \Big\} \Big| \GG_t^n \Big]\rd r \bigg| \Bigg],
\end{align*}
where we dropped the indicators in order to be able to compute the probability of no bounce. We again decompose according to the event $G_n$ in order to proceed
\begin{align*}
&J_{1,1,2}\\
&\leq 2\|f\|_\infty \rE\Bigg[\sup_{t\in[0,T] \cap \mathbb{Q}} \bigg| \epsilon_n^{-1}\int_0^{\epsilon_n}
\rE\Big[\rE\Big\{ \Big[ 1-\exp\Big(-\int_0^{\epsilon_n} ( \nabla U_n(\bm{X}_n(t+s)), \bm{V}_n(t+s) )_+
\rd s  \Big) \Big]\mathbb{1}_{G_n} \Big| \FF_t^n \Big\} \Big| \GG_t^n \Big]\rd r \bigg| \Bigg]\\
&\,+ 2\|f\|_\infty \rE\Bigg[\sup_{t\in[0,T] \cap \mathbb{Q}} \bigg| \epsilon_n^{-1}\int_0^{\epsilon_n}
\rE\Big[\rE\Big\{ \Big[ 1-\exp\Big(-\int_0^{\epsilon_n} ( \nabla U_n(\bm{X}_n(t+s)), \bm{V}_n(t+s) )_+
\rd s  \Big) \Big]\mathbb{1}_{G_n^\mathtt{c}} \Big| \FF_t^n \Big\} \Big| \GG_t^n \Big]\rd r \bigg| \Bigg].
\end{align*}
Since the integrand is bounded above by 1, a calculation similar to the one for the term $J_2$ in \eqref{eq:j2_doob} shows that the second term above vanishes as $n\to \infty$, and therefore using the inequality
$1-\exp(-x) \leq x$ for $x>0$ we have for $p>1$
\begin{align*}
&J_{1,1,2}\\
&\leq C \delta_n^{1/p}+
2\|f\|_\infty \rE\Bigg[\sup_{t\in[0,T] \cap \mathbb{Q}} \bigg| \epsilon_n^{-1}\int_0^{\epsilon_n}
\rE\Big[ \rE\Big\{ \int_0^{\epsilon_n} |\nabla U_n(\bm{X}_n(t+s))| |\bm{V}_n(t+s)| \rd s  \mathbb{1}_{G_n}\Big| \FF_t^n \Big\} \Big| \GG_t^n \Big]\rd r \bigg| \Bigg]\\
&\leq C \delta_n^{1/p}+\\
 &\quad +2\|f\|_\infty \rE\left[
 \sup_{t\in[0,T] \cap \mathbb{Q}}
 \left| \epsilon_n^{-1}\int_0^{\epsilon_n}\rE\left[ \left. \rE \left\{ \left.\int_0^{\epsilon_n}\left( \frac{1}{2}|\nabla U_n(\bm{X}_n(t+s))|^2 + \frac{1}{2} |\bm{V}_n(t+s)|^2 \right)\rd s  \times
 \mathbb{1}_{G_n}\right| \FF_t^n \right\} \right| \GG_t^n \right]\rd r \right| \right]\\
 &\leq C \delta_n^{1/p}+\\
 &\quad 2 C \|f\|_\infty \rE\left[
 \sup_{t\in[0,T] \cap \mathbb{Q}}
 \left| \epsilon_n^{-1}\int_0^{\epsilon_n}\rE\left[ \left. \rE \left\{ \left.\int_0^{\epsilon_n}
 \left( {M|\bm{X}_n(t+s)|^2} + |\bm{V}_n(t+s)|^2 \right)\rd s  \times
 \mathbb{1}_{G_n}\right| \FF_t^n \right\} \right| \GG_t^n \right]\rd r \right| \right]\\
 \intertext{since $|\nabla U_n(\bm{x}| =|\nabla U_n(\bm{x} - \nabla U_n(0)|\leq M|\bm{x}|$ by Assumption~\ref{ass:potential}}
 &\leq C \delta_n^{1/p}+\\
 &\quad 2C{M}\|f\|_\infty \rE\Bigg[\sup_{t\in[0,T] \cap \mathbb{Q}} \bigg| \epsilon_n^{-1}\int_0^{\epsilon_n}
 \rE\Big[ C\epsilon_n |\bm{Z}_n(t+s)|^2\mathbb{1}_{G_n}\Big| \FF_t^n \Big\} \Big| \GG_t^n \Big]\rd r \bigg| \Bigg]\\
&\leq C \delta_n^{1/p}+ 2C\|f\|_\infty \epsilon_n K_{n,\delta_n}^2.
\end{align*}
We choose $\epsilon_n$ such that $\epsilon_n K^2_{n,\delta_n} \to 0$.
\subsubsection{Proof of \eqref{eq:fourthcondition}.}
Next we prove \eqref{eq:fourthcondition}. First, by stationarity notice that we can equivalently check
$$\rE\left[ \left| \phi_n(0) - \mathcal{A}f\left( Z_n^{(1)}(0) \right)\right| \right] \to 0.$$
Notice first that $f\in \mathcal{D}\big(\widetilde{\AAA}_n\big)$, the domain of the \textit{extended generator}, since $f$ is smooth and bounded (see \cite[Theorem~26.14]{D_93})
\begin{align*}
 \phi_n(0)
 &=\epsilon_n^{-1} \rE\left[\left. f\left( Z_n^{(1)}(\epsilon_n) \right) - f\left( Z_n^{(1)}(0) \right) \right| \GG_0^n \right]\\
 &=\epsilon_n^{-1} \rE\left[\left. \int_0^{\epsilon_n} \widetilde{\AAA}_n f \left(\ZZZ_n(s) \right)\rd s + \mathcal{R}_n(s)  \right| \GG_0^n \right]\\
 &=\epsilon_n^{-1} \rE\left[\left. \int_0^{\epsilon_n} \widetilde{\AAA}_n f \left(\ZZZ_n(s) \right)\rd s \right| \GG_0^n \right],
\end{align*}
where we used the facts that $\mathcal{R}_n(t)$ is an $\mathcal{F}_t^n$-martingale and $\mathcal{F}_t^n \subseteq \mathcal{G}_t^n$, whence
$$\rE\left[\left. \mathcal{R}_n(s)\right| \GG_0^n \right]
=\rE\left\{\left. \rE\left[\left. \mathcal{R}_n(s)\right| \FF_0^n \right]\right|\GG_0^n \right\}=0.
$$
We also notice that $g_n:=\widetilde{\AAA}_n f\in \mathrm{Dom}\big(\widetilde{\AAA}_n\big)$ the domain of the extended generator.
Therefore
\begin{align*}
 \phi_n(0)
  &=\epsilon_n^{-1} \int_0^{\epsilon_n} \rE\left[\left. g_n \left(\ZZZ_n(s) \right) \right| \GG_0^n \right]\rd s\\
  &=\epsilon_n^{-1} \int_0^{\epsilon_n} \rE\left[\left. \widetilde{\AAA}_n f \left(\ZZZ_n(0) \right) + \int_0^s \widetilde{\AAA}_n g_n \left(\ZZZ_n(r) \right) + \mathcal{R}_n'(s) \rd r \right| \GG_0^n \right]\rd s\\
  &=\epsilon_n^{-1} \int_0^{\epsilon_n} \rE\left[\left. \widetilde{\AAA}_n f \left(\ZZZ_n(0) \right) + \int_0^s \widetilde{\AAA}_n g_n \left(\ZZZ_n(r) \right)\rd r \right| \GG_0^n \right]\rd s,
\end{align*}
where, from \cite[Theorem~26.12]{D_93}, it follows that the local martingale $\{\mathcal{R}_n'(s):s\geq 0\}$ is actually a proper martingale,
and therefore using the same arguments as before, for $s>0$,
$$\rE\left[\left. \mathcal{R}_n'(s) \right| \GG_0^n \right]=0.$$
Then we have
\begin{align}
  \rE\left[ \left| \phi_n(0) - \mathcal{A}f\left( Z_n^{(1)}(0) \right)\right| \right]
 &\leq \rE\left[ \left|\rE\left[\left. \widetilde{\AAA}_n f \left(\ZZZ_n(0) \right)\right| \GG_0^n\right] - \mathcal{A}f\left( Z_n^{(1)}(0) \right)\right| \right] \notag\\
  &\qquad + \rE\left\{ \left| \epsilon_n^{-1} \int_0^{\epsilon_n}\rE\left[ \left. \int_0^s \widetilde{\AAA}_n g_n \left(\ZZZ_n(r) \right)\rd r  \right| \GG_0^n \right]\rd s\right|\right\}
  \notag\\
&\leq \rE\left[ \left|\rE\left[\left. \widetilde{\AAA}_n f \left(\ZZZ_n(0) \right)\right| \GG_0^n\right] - \mathcal{A}f\left( Z_n^{(1)}(0) \right)\right| \right]
\notag\\
  &\qquad + \epsilon_n^{-1} \int_0^{\epsilon_n}\int_0^s\rE\left\{  \rE\left[ \left.  \big| \widetilde{\AAA}_n g_n \left(\ZZZ_n(r) \right)\big| \right| \GG_0^n \right]\right\}\rd r \rd s
  \notag\\
&:=\rE\left[ \left|\rE\left[\left. \widetilde{\AAA}_n f \left(\ZZZ_n(0) \right)\right| \GG_0^n\right] - \mathcal{A}f\left( Z_n^{(1)}(0) \right)\right| \right]
+ \mathcal{R}_n,\label{eq:Rndef}
  \end{align}
applying Jensen's inequality conditionally. Finally by the tower law and by stationarity of $\{\ZZZ_n(t):t\geq 0\}$ when initialized from $\pi_n$
\begin{align*}
\mathcal{R}_n
&= \epsilon_n^{-1} \int_0^{\epsilon_n}\int_0^s\rE\left\{  \rE\left[ \left.  \big| \widetilde{\AAA}_n g_n \left(\ZZZ_n(r) \right)\big| \right| \GG_0^n \right]\right\}\rd r \rd s\\
 &=\epsilon_n^{-1} \int_0^{\epsilon_n}\int_0^s\rE\left\{  \big| \widetilde{\AAA}_n g_n \left(\ZZZ_n(r) \right)\big| \right\}\rd r \rd s =\frac{\epsilon_n}{2}\rE\left\{  \big| \widetilde{\AAA}_n g_n \left(\ZZZ_n(0) \right)\big| \right\}.
\end{align*}
\paragraph{Error term.}
We will now control this error term. Recall first that for $f \in C^\infty_c(\mathcal{Z})\subset \mathcal{D}(\AAA_n)$ we have
\begin{align*}
 \AAA_n f(\bm{x},\bm{v})
 &= ( \nabla f(\bm{x}), \bm{v})) +  \max\{0, ( \nabla U_n(\bm{x}), \bm{v})\}
 \left[\mathfrak{R}_nf\left(\bm{x},\bm{v}\right) - f\left(\bm{x},\bm{v}\right)  \right]
 + \lref \left[ Qf \left(\bm{x},\bm{v}\right)-f\left(\bm{x},\bm{v}\right)\right],\\
 \mathfrak{R}_nf\left(\bm{x},\bm{v}\right)
 &:= f\left(\bm{x},\bm{v}- 2 \frac{( \nabla U_n(\bm{x}), \bm{v}))}{| \nabla U_n(\bm{x})|^2} \nabla U_n(\bm{x})\right),\\
 Q_{\alpha, n}f\left(\bm{x},\bm{v}\right)
 &:= \frac{1}{(2\pi)^{n/2}}\int_{\mathbb{R}^n} \re^{-|\boldsymbol{\xi}|^2/2} f\left(\bm{x},\alpha\bm{v}+\sqrt{1-\alpha^2} \boldsymbol{\xi}\right) \rd \boldsymbol{\xi}.
\end{align*}
Potentially abusing notation, for $n\geq 1$ and $\bm{x}\in \mathbb{R}^n$ we define a mapping $R_n(\bm{x}): \mathbb{R}^n \mapsto \mathbb{R}^n$ through
$$R_n(\bm{x}) \bm{v}:= \bm{v}- 2 \frac{( \nabla U_n(\bm{x}), \bm{v})}{| \nabla U_n(\bm{x})|^2} \nabla U_n(\bm{x}),$$
{with the convention that $R_n(x) \bm{v}=0$, when $\nabla U_n (\bm{x})=0$. }

We decompose the generator $\AAA_n$ into three parts
$$\AAA_n = \AAA_n^{(1)}+\AAA_n^{(2)}+\AAA_n^{(3)},$$
where
\begin{align*}
 \AAA_n^{(1)} f (\bm{x},\bm{v})
 &=\left.\frac{\rd}{\rd t}f\left(\bm{x}+t \bm{v}, \bm{v}\right) \right|_{t=0}
 ,\\
 \AAA_n^{(2)} f (\bm{x},\bm{v})
 &=\max\{0, ( \nabla U_n(\bm{x}), \bm{v})\}
 \left[\mathfrak{R}_nf\left(\bm{x},\bm{v}\right) - f\left(\bm{x},\bm{v}\right)  \right],\\
 \AAA_n^{(3)} f (\bm{x},\bm{v})
 &=\lref \left[ Qf \left(\bm{x},\bm{v}\right)-f\left(\bm{x},\bm{v}\right)\right].
\end{align*}
\begin{remark}\label{rem:abscontinuity}
{Notice that when $f$ is differentiable we have
$$ \AAA_n^{(1)} f (\bm{x},\bm{v})
 =\langle \nabla f(\bm{x}), \bm{v}\rangle,$$
however for $\AAA_n^{(1)}f(x,v)$ to be well defined we only need that $t\mapsto f(\bm{x}+t\bm{v}, \bm{v})$ is absolutely continuous, see \citet[Chapter~2.22]{D_93}.}
\end{remark}

Therefore when considering $\AAA_n g_n=\AAA_n\AAA_n f_n$ we will need to consider all possible combinations $\AAA_n^{(i)} \AAA_n^{(j)}$ since the operators do not necessarily commute.
\paragraph{Case $i=1$.}
Using the fact that $f(\bm{x},\bm{v})=f(x_1,v_1)$, where we write $(x_1,v_1)$ for the first location and velocity components of $(\bm{x},\bm{v})$, the first term reduces to
\begin{align*}
 \AAA_n^{(1)} \AAA_n^{(1)}f (\bm{x},\bm{v})
 &= \left.\frac{\rd }{\rd t} ( \nabla f(\bm{x}), \bm{v})\right|_{t=0}= \left.\frac{\rd }{\rd t}  \frac{\partial}{\partial x}f(x_1+tv_1, v_1)v_1\right|_{t=0}\\
 &=\frac{\partial^2 f}{\partial x^2}(x_1,v_1) v_1^2.
\end{align*}
Since $f\in C^\infty_c (\mathbb{R}\times \mathbb{R})$, it follows that $\partial^2_x f(x,v)$ is also continuous and compactly supported and therefore bounded. Thus
\begin{align*}
\rE\left|\frac{\partial^2 f}{\partial x^2}(X^{(1)},V^{(1)}) \left(V^{(1)}\right)^2 \right|
&\leq \left\|\frac{\partial^2 f}{\partial x^2}\right\|_\infty\rE\left[\left(V^{(1)}\right)^2 \right]\leq \left\|\frac{\partial^2 f}{\partial x^2} \right\|_\infty = {O(1)},
\end{align*}
since under $\pi_n$, $V^{(1)}$ is centered Gaussian with unit variance.

The second term, see Remark~\ref{rem:abscontinuity}, takes the form
\begin{align*}
 \AAA_n^{(1)} \AAA_n^{(2)}f (\bm{x},\bm{v})
 &= \left.\frac{\rd }{\rd t} \AAA_n^{(2)}f (\bm{x}+t\bm{v},\bm{v})\right|_{t=0}\\
 &= \left.\frac{\rd }{\rd t} \max\{0, ( \nabla U_n(\bm{x}+t\bm{v}), \bm{v}))\}  \right|_{t=0}
  \left[\mathfrak{R}f\left(\bm{x},\bm{v}\right) - f\left(\bm{x},\bm{v}\right)  \right]\\
  &\qquad +\left.\frac{\rd }{\rd t} \left[\mathfrak{R}f\left(\bm{x}+t\bm{v},\bm{v}\right) - f\left(\bm{x}+t\bm{v},\bm{v}\right)  \right]  \right|_{t=0}
 \max\{0, ( \nabla U_n(\bm{x}), \bm{v}))\}\\
 &=: J_1 + J_2.
\end{align*}
For $J_1$, since by Assumption~\ref{ass:potential} $\nabla U_n$ is $M$-Lipschitz
\begin{align*}
 \lefteqn{\left|\max\left\{0, \left( \nabla U_n(\bm{x}+h\bm{v}), \bm{v}\right)\r\}-\max\l\{0, \l( \nabla U_n(\bm{x}), \bm{v}\r)\r\}\right|}\\
 &\leq \left|( \nabla U_n(\bm{x}+h\bm{v}), \bm{v})-( \nabla U_n(\bm{x}), \bm{v}))\right|
 \leq {M h |\bm{v}|^2.}
\end{align*}
Therefore we have that, for $h\in (0,1)$
\begin{align*}
 h^{-1} \left|\max\{0, ( \nabla U_n(\bm{x}+h\bm{v}), \bm{v})\}-\max\{0, ( \nabla U_n(\bm{x}), \bm{v})\}\right|
 &\leq M|\bm{v}|^2
 \in L^1(\pi),
\end{align*}
since the $V_i$ are standard normal random variables.
In addition since $f$ is bounded it follows that
$\mathfrak{R}f\left(\bm{x},\bm{v}\right) \leq \|f\|_\infty$. Therefore by the dominated convergence theorem, we can exchange the $h\to 0$ limit and expectation to  obtain
\begin{align*}
\pi \left[J_1\right]
  &\leq 2\|f\|_\infty \rE \left[ \left|\frac{\rd }{\rd t} \max\{0, ( \nabla U_n(\bm{X}+t\bm{V}), \bm{V}))\}  \Big|_{t=0}\right| \right]\\
  &\leq 2\|f\|_\infty \rE \left[\lim_{h\to 0} h^{-1}\left| ( \nabla U_n(\bm{X}+h\bm{V}), \bm{V}))  -( \nabla U_n(\bm{X}), \bm{V}))\right| \right]\\
  &\leq {2\|f\|_\infty M\rE  \left[  \left| \bm{V}\right|^{2} \right]}
  ={O( M n)}.
\end{align*}
For $J_2$ a lengthy but straightforward calculation shows that
\begin{align*}
\lefteqn{\left.\frac{\rd }{\rd t} \mathfrak{R}f\left(\bm{x}+t\bm{v},\bm{v}\right)\right|_{t=0}}\\
&=\left.\frac{\rd }{\rd t}
f\left(x_1+tv_1,v_1- 2 \frac{( \nabla U_n(\bm{x}+t\bm{v}), \bm{v}))}{| \nabla U_n(\bm{x}+t\bm{v})|^2}  \partial_1 U_n(\bm{x}+t\bm{v})\right)
 \right|_{t=0}\\
&=\partial_x f\left(x_1, v_1-2\frac{( \nabla U_n(\bm{x}), \bm{v})}{| \nabla U_n(\bm{x})|^2}  \partial_1 U_n(\bm{x})\ \right)v_1 \\
&\quad -2 \partial_v f\left(x_1, v_1-2\frac{( \nabla U_n(\bm{x}), \bm{v})}{| \nabla U_n(\bm{x})|^2}  \partial_1 U_n(\bm{x}) \right)
\left.\frac{\rd }{\rd t} \left( \frac{( \nabla U_n(\bm{x}+t\bm{v}), \bm{v})}{| \nabla U_n(\bm{x}+t\bm{v})|^2}  \partial_1 U_n(\bm{x}+t\bm{v})
\right)
 \right|_{t=0}\\
 &= (\mathcal{R}\partial_x f)(\bm{x},\bm{v})v_1 - (\mathcal{R}\partial_v f)(\bm{x},\bm{v})\times \mathfrak{U}(\bm{x},\bm{v}),
\end{align*}
where
\begin{multline*}
\mathfrak{U}(\bm{x},\bm{v})
:= \left.\frac{\rd }{\rd t} \left( 2\frac{( \nabla U_n(\bm{x}+t\bm{v}), \bm{v})}{| \nabla U_n(\bm{x}+t\bm{v})|^2}  \partial_1 U_n(\bm{x}+t\bm{v})
\right)
 \right|_{t=0}\\
=\frac{2}{|\nabla U_n(\bm{x})|^2}\left\{
{ \left( \bm{v} , \nabla{U}^2_n(\bm{x}) \bm{v} \right) \partial_1 U_n(\bm{x})
+ \left(\nabla U_n(\bm{x}), \bm{v} \right)\sum_{j=1}^n \partial^2_{j,1} U_n(\bm{x}) v_j}
\right\}\\
-
\frac{1}{|\nabla U_n(\bm{x})|^4}\left\{
{ 2 \partial_1 U_n(\bm{x}) \left(\nabla U_n(\bm{x}), \nabla^2 U_n(\bm{x}) \bm{v} \right) \left(\nabla U_n(\bm{x}), \bm{v}\right)  }
\right\} ,
\end{multline*}
{and thus by Assumption~\ref{ass:potential}} 
{
\begin{align*}
\left|\mathfrak{U}(\bm{x},\bm{v}) \right|
&\leq \frac{2}{|\nabla U_n(\bm{x})|^2}\left\{
M|\nabla U_n(\bm{x})| |\bm{v}|^2+ M |\bm{v}|^2 |\nabla U_n(\bm{x})|
\right\}\\
&\qquad +
\frac{1}{|\nabla U_n(\bm{x})|^4}\left\{
2 M |\nabla U_n(\bm{x})|^3 |\bm{v}|^2
\right\},
\end{align*}
}
whence
\begin{multline*}
\left|\mathfrak{U}(\bm{x},\bm{v}) \max\{0, ( \nabla U_n(\bm{x}), \bm{v}))\} \right|
\leq \frac{C |\bm{v}|}{|\nabla U_n(\bm{x})|}\left\{
{M|\nabla U_n(\bm{x})| |\bm{v}|^2+ M |\bm{v}|^2 |\nabla U_n(\bm{x})|}
\right\}\\
+
\frac{ C|\bm{v}|}{|\nabla U_n(\bm{x})|^3}\left\{
M |\nabla U_n(\bm{x})|^3 |\bm{v}|^2
\right\}\leq {C M |\bm{v}|^3}.
\end{multline*}
Thus overall,
\begin{align*}
\left|\frac{\rd }{\rd t} \mathfrak{R}f\left(\bm{x}+t\bm{v},\bm{v}\right)\!\Big|_{t=0} \max\{0, ( \nabla U_n(\bm{x}), \bm{v}))\} \right|
&\leq \|\partial_x f\|_\infty |\bm{v}|^2 |\nabla U_n(\bm{x})|
+ C M\|\partial_v f\|_\infty |\bm{v}|^3.
\end{align*}
On the other hand
\begin{align*}
\left|\left.\frac{\rd }{\rd t} f\left(\bm{x}+t\bm{v},\bm{v}\right) \right|_{t=0}
 \max\{0, ( \nabla U_n(\bm{x}), \bm{v}))\}\right|
 &\leq \|\partial_x f\|_\infty |\nabla U(\bm{x})| |\bm{v}|^2.
\end{align*}
Thus overall we have that, using the fact that $(V_1, \dots, V_n)$ are i.i.d.\ standard Gaussians and Lemma~\ref{lemma:BPSexpectednbofbouncesupper} in the Appendix
\begin{align*}
\pi[|J_2|]
&\leq C \rE\left[|\nabla U_n(\bm{X})|\right]\rE\left[| \bm{V}^2| \right]
	+C M\rE\left[| \bm{V}|^3 \right]\\
&\leq {CM^{1/2} n^{3/2}
	+C M n^{3/2} = O(Mn^{3/2})}
\end{align*}
and thus we have that $\pi \left[ \left| \AAA_n^{(1)} \AAA_n^{(2)}f \right|\right] = O(Mn^{3/2})$.

For the final term, since $Qf(\bm{x},\bm{v}) = Qf(x_1, v_1)$ we have
\begin{align*}
 \AAA_n^{(1)} \AAA_n^{(3)}f (\bm{x},\bm{v})
 &= \left.\frac{\rd }{\rd t} \AAA_n^{(3)}f (\bm{x}+t\bm{v},\bm{v})\right|_{t=0}\\
 &=\lref \left.\frac{\rd }{\rd t} \left[ Qf(x_1+tv_1,v_1)-f(x_1+t v_1,v_1)\right]\right|_{t=0}\\
 &=\lref \left[ Q(\partial_x f)(x_1+tv_1,v_1)-\partial_x f(x_1,v_1)\right]v_1,
\end{align*}
by an application of dominated convergence. We can easily see from the above that
$\pi \big[ \AAA_n^{(1)} \AAA_n^{(3)}f \big] =O(1)$ as $n\to \infty$.

\paragraph{Case $i=2$.}
For the first term $\AAA_n^{(2)} \AAA_n^{(1)}f$, notice first that since $f(\bm{x},\bm{v}) = f(x_1, v_1)$ we have
$$\AAA_n^{(1)}f(\bm{x}, \bm{v}) = \partial_x f(x_1,v_1) v_1=:h(x_1, v_1).$$
Therefore
\begin{align*}
\mathfrak{R}_n h(\bm{x},\bm{v})
&= \partial_x f\left(x_1,v_1- 2 \frac{( \nabla U_n(\bm{x}), \bm{v})}{| \nabla U_n(\bm{x})|^2} \partial_1  U_n(\bm{x})\right)
\left(v_1- 2 \frac{( \nabla U_n(\bm{x}), \bm{v})}{| \nabla U_n(\bm{x})|^2} {\partial_1  U_n(\bm{x})}\right),
\end{align*}
whence
\begin{align*}
\mathfrak{R}_n h(\bm{x},\bm{v})- h(\bm{x},\bm{v})
&= v_1 \left[ \mathfrak{R}_n \partial_xf(\bm{x},\bm{v})- \partial_x f(\bm{x},\bm{v})\right] - 2  \mathfrak{R}_n \partial_x f(\bm{x},\bm{v}) \frac{( \nabla U_n(\bm{x}), \bm{v})}{| \nabla U_n(\bm{x})|^2} \partial_1  U_n(\bm{x}),
\end{align*}
and thus
\begin{align*}
 \rE\left|\AAA_n^{(2)} \AAA_n^{(1)}f (\bm{X},\bm{V})\right|
&\leq  \rE\left[ \left|  (\nabla U_n(\bm{X}), \bm{V})\right|\times  |V_1| \times \left|\mathfrak{R}_n \partial_xf(\bm{X},\bm{V})- \partial_x f(\bm{X},\bm{V}) \right|    \right]\\
&\qquad + 2\rE\left[ \left|  (\nabla U_n(\bm{X}), \bm{V})\right| \times  \left|\mathfrak{R}_n \partial_x f(\bm{X},\bm{V}) \frac{( \nabla U_n(\bm{X}), \bm{V})}{| \nabla U_n(\bm{X})|^2} {\partial_1  U_n(\bm{x})} \right|    \right]\\
&\leq\left( \|\mathfrak{R}_n \partial_x f\|_\infty  + \|\partial_x f\|_\infty \right)
{ \rE\left[ |\bm{V}| \times |V_1|\right] \rE \left[ |\nabla U_n(\bm{X})|   \right]}\\
&\qquad + 2\|\mathfrak{R}_n \partial_x f\|_\infty \rE\left[ \frac{( \nabla U_n(\bm{X}), \bm{V}) ^2}{| \nabla U_n(\bm{X})|^2} \left|{\partial_1  U_n(\bm{x})}\right| \right]\\
&\leq\left( \|\mathfrak{R}_n \partial_x f\|_\infty  + \|\partial_x f\|_\infty \right)
{ \rE\left[ |\bm{V}| \times |V_1|\right] \rE \left[ |\nabla U_n(\bm{X})|   \right]}\\
&\qquad + 2\|\mathfrak{R}_n \partial_x f\|_\infty \rE\left[  \left|{\partial_1  U_n(\bm{X})}\right| \right],
\end{align*}
where for the second term we used the tower law and the fact that conditionally on $\bm{X}$, $(\nabla U_n(\bm{X}), \bm{V})$ is Gaussian with mean 0 and variance $|\nabla U_n(\bm{X})|^2$.
Using the Cauchy-Schwarz inequality and Lemma~\ref{lemma:BPSexpectednbofbouncesupper} from the Appendix we have
\begin{align*}
 \rE\left|\AAA_n^{(2)} \AAA_n^{(1)}f (\bm{X},\bm{V})\right|
&\leq\left( \|\mathfrak{R}_n \partial_x f\|_\infty  + \|\partial_x f\|_\infty \right)
{C M\sqrt{n}  \rE\left[ |V_1|^2\right]^{1/2} \rE[|\bm{V}|^2]^{1/2} }\\
&\qquad + 2\|\mathfrak{R}_n \partial_x f\|_\infty   \rE\left[\left|\nabla U_n(\bm{X}) \right| \right] = {O(Mn)}.
\end{align*}

For the next term  $\AAA_n^{(2)} \AAA_n^{(2)}f$ first we write
\begin{align*}
\AAA_n^{(2)}  \AAA_n^{(2)} f (\bm{x},\bm{v})
&= \max\left\{ 0, ( \nabla U_n(\bm{x}), \bm{v}) \right\}
  \left[\mathfrak{R}_n \AAA_n^{(2)} f (\bm{x},\bm{v}) - \AAA_n^{(2)} f (\bm{x},\bm{v})  \right].
\end{align*}
Then notice that
\begin{align*}
\mathfrak{R}_n \AAA_n^{(2)} f (\bm{x},\bm{v})
&= \max \left\{0, \left( \nabla U_n(\bm{x}), \bm{v}- 2 \frac{( \nabla U_n(\bm{x}), \bm{v})}{| \nabla U_n(\bm{x})|^2} \nabla U_n(\bm{x}) \right) \right\}\\
&\qquad \times
\left[\mathfrak{R}_n f\left( x_1, v_1-  2 \frac{( \nabla U_n(\bm{x}), \bm{v})}{| \nabla U_n(\bm{x})|^2} {\partial_1  U_n(\bm{X})}\right) - f\left( x_1, v_1-  2 \frac{( \nabla U_n(\bm{x}), \bm{v})}{| \nabla U_n(\bm{x})|^2} {\partial_1  U_n(\bm{X})}\right) \right]
\\
&= \max \left\{0, \left( \nabla U_n(\bm{x}), -\bm{v}\right) \right\}\\
&\qquad \times
\left[\mathfrak{R}_n f\left( x_1, v_1-  2 \frac{( \nabla U_n(\bm{x}), \bm{v})}{| \nabla U_n(\bm{x})|^2} {\partial_1  U_n(\bm{X})}\right) - f\left( x_1, v_1-  2 \frac{( \nabla U_n(\bm{x}), \bm{v})}{| \nabla U_n(\bm{x})|^2} {\partial_1  U_n(\bm{X})}\right) \right],
\end{align*}
and therefore that
\begin{align*}
\left|\mathfrak{R}_n \AAA_n^{(2)} f (\bm{x},\bm{v}) - \AAA_n^{(2)} f (\bm{x},\bm{v})\right|
&\leq 2 \|f\|_\infty  \left| \left( \nabla U_n(\bm{x}), \bm{v}\right) \right|\\
\left|\AAA_n^{(2)}\AAA_n^{(2)} f (\bm{x},\bm{v})\right|
&\leq 2 \|f\|_\infty  \left( \nabla U_n(\bm{x}), \bm{v}\right) ^2.
\end{align*}
Thus
\begin{align*}
\rE\left|\AAA_n^{(2)}  \AAA_n^{(2)} f (\bm{X},\bm{V})\right|
&\leq C \|f\|_\infty \rE\left[ \left( \nabla U_n(\bm{X}), \bm{V}\right) ^2 \right]\\
&\leq C \|f\|_\infty \rE\left\{ \rE\left[ \left. \left( \nabla U_n(\bm{X}), \bm{V}\right) ^2 \right| \bm{X} \right] \right\}
\intertext{using the fact that conditionally on $\bm{X}$, $(\nabla U_n(\bm{X}), \bm{V})$ is Gaussian}
&= C \|f\|_\infty \rE\left\{ |\nabla U_n(\bm{X})|^2\right\} =  {O(Mn)}
\end{align*}
from Lemma~\ref{lemma:BPSexpectednbofbouncesupper} in the Appendix.

Next we consider the term $\AAA_n^{(2)} \AAA_n^{(3)}f$. Since $f$ is bounded, it easily follows that $\AAA_n^{(3)}f$ is also bounded and therefore that
\begin{align*}
 \left|\AAA_n^{(2)} \AAA_n^{(3)}f (\bm{x}, \bm{v})\right|
 &= \max\left\{ 0, ( \nabla U_n(\bm{x}), \bm{v}) \right\}
    \left|\mathfrak{R}_n \AAA_n^{(3)} f (\bm{x},\bm{v}) - \AAA_n^{(3)} f (\bm{x},\bm{v})  \right|\\
 &\leq 2 \lref \|f\|_\infty    \max\left\{ 0, ( \nabla U_n(\bm{x}), \bm{v}) \right\}.
\end{align*}
 Therefore
 \begin{align*}
  \rE \left|\AAA_n^{(2)} \AAA_n^{(3)}f (\bm{X}, \bm{V}) \right|
  &\leq C \rE\left| ( \nabla U_n(\bm{X}), \bm{V}) \right|
  \leq C \rE\left[ ( \nabla U_n(\bm{X}), \bm{V})^2 \right]^{1/2}
  = {O(M^{1/2}n^{1/2})},
 \end{align*}
{from Lemma~\ref{lemma:BPSexpectednbofbouncesupper} and calculations similar to the previous term.}
\paragraph{Case $i=3$.} The first term to consider is
\begin{align*}
 \AAA_n^{(3)} \AAA_n^{(1)}f(\bm{x}, \bm{v})
 &= \lref \left[Q \AAA_n^{(1)}f(\bm{x},\bm{v})-\AAA_n^{(1)}f(\bm{x},\bm{v}) \right] \\
 &= \lref  \int \left[\AAA_n^{(1)}f(x_1,\alpha v_1+\sqrt{1-\alpha^2}\xi)-\AAA_n^{(1)}f(\bm{x},\bm{v}) \right]\phi(\xi) \rd \xi\\
 &= \lref  \int \left[\partial_x f(x_1,\alpha v_1+\sqrt{1-\alpha^2}\xi)\left(\alpha v_1+\sqrt{1-\alpha^2}\xi\right)-\partial_x f(x_1,v_1)v_1 \right]\phi(\xi) \rd \xi,
\end{align*}
where $\phi$ denotes the standard normal density.
Since $\|\partial_x f\|_\infty <\infty$ we have
\begin{align*}
  \rE\left|\AAA_n^{(3)} \AAA_n^{(1)}f(\bm{X}, \bm{V})\right|
 &\leq \lref\|\partial_x f\|_\infty \rE\left[\left|\alpha V_1 + \sqrt{1-\alpha^2}\xi\right| + |V_1|\right] =O(1),
\end{align*}
as $n\to \infty$.

For the second term we have, using Jensen's inequality on the Markov kernel $Q$,
\begin{align*}
\rE\left|\AAA_n^{(3)} \AAA_n^{(2)}f(\bm{X},\bm{V})\right|
&\leq \lref \rE\left[\left| Q \AAA_n^{(2)}f(\bm{X},\bm{V}) \right|\right] + \lref \rE\left[\left|\AAA_n^{(2)}f(\bm{X},\bm{V})\right| \right] \\
&\leq \lref \rE\left[ Q \left(\big|\AAA_n^{(2)}f\big|\right)(\bm{X},\bm{V}) \right] + \lref \rE\left[\left|\AAA_n^{(2)}f(\bm{X},\bm{V})\right| \right].
\end{align*}
At this point notice that $Q$ is $\pi_n$-invariant and therefore
$$\rE\left[ Q \left(\big|\AAA_n^{(2)}f\big|\right)(\bm{X},\bm{V}) \right] = \rE\left[ \big|\AAA_n^{(2)}f (\bm{X},\bm{V})\big| \right],$$
whence we conclude that
\begin{align*}
\rE\left|\AAA_n^{(3)} \AAA_n^{(2)}f(\bm{X},\bm{V})\right|
&\leq 2\lref \rE\left[\left|\AAA_n^{(2)}f(\bm{X},\bm{V})\right| \right]\\
&\leq 4 \lref \|f\|_\infty \rE\left[\left|( \nabla U_n(\bm{X}), \bm{V}) \right|\right]\\
&= 4\sqrt{\frac{2}{\pi}} \lref \|f\|_\infty \rE\left[ |\nabla U_n(\bm{X})|\right]
={O(M^{1/2}n^{1/2})},
\end{align*}
using Lemma~\ref{lemma:BPSexpectednbofbouncesupper} and the fact that conditionally on $\bm{X}$, $(\nabla U_n(\bm{X}), \bm{V})$ is a mean zero Gaussian with variance $|\nabla U_n(\bm{X})|^2$.

Finally, by similar arguments as above the last term is given by
\begin{align*}
\rE\left|\AAA_n^{(3)} \AAA_n^{(3)}f(\bm{X},\bm{V})\right|
&\leq 2\lref \rE\left[\left|\AAA_n^{(3)}f(\bm{X},\bm{V})\right| \right]\\
&\leq 4\lref^2 \|f\|_\infty = O(1).
\end{align*}
Overall we have shown that the error term defined in \eqref{eq:Rndef} satisfies
\begin{equation}\label{eq:bndRn}\mathcal{R}_n= \frac{\epsilon_n}{2} \rE\left[ \left| \mathcal{A}_n \mathcal{A}_n f \left( \bm{Z}_n (0)\right) \right| \right]
= {O(M n^{3/2} \epsilon_n)} = o(1),
\end{equation}
since we have chosen $\epsilon_n$ such that $\epsilon_n n^2\to 0$, as $n\to \infty$.
\paragraph{Main term.}
Having controlled the error term, we now focus on the main term given by
$$\rE\left[ \left|\rE\left[\left. \widetilde{\AAA}_n f \left(\ZZZ_n(0) \right)\right| \GG_0^n\right]- \mathcal{A}f\left( Z_n^{(1)}(0) \right)\right| \right],$$
where we recall that $\widetilde{\AAA}_n$ is the extended generator.
Notice that for $f(\bm{x},\bm{v}) = f(x_1, v_1)$,
\begin{align*}
 \mathcal{\AAA}_n f\left(\bm{x},\bm{v} \right)
 &= \partial_x f(x_1, v_1) v_1 + \max\left\{ 0, ( \nabla U_n(\bm{x}), \bm{v}) \right\} \left[ \mathfrak{R}_n f(\bm{x}, \bm{v})-f(\bm{x}, \bm{v})\right] + \lref\left[ Q f(x_1, v_1) - f(x_1,v_1)\right]\\
\mathcal{\AAA} f\left(x_1,v_1 \right)
 &= \partial_x f(x_1, v_1) v_1 - \partial_v f(x_1,v_1) W'(x_1) + \lref\left[ Q f(x_1, v_1) - f(x_1, v_1)\right],
\end{align*}
and thus the first and third terms are in fact identical and will cancel out. We thus only have to consider the difference of the second terms. We apply a first order Taylor expansion
\begin{align*}
\lefteqn{\rE\left[ \left.  \max\left\{ 0, ( \nabla U_n(\bm{X}), \bm{V}) \right\} \left[ \mathfrak{R}_n f(\bm{X}, \bm{V})-f(\bm{X}, \bm{V})\right] \right| \GG_0^n \right] }\\
&= \rE\Bigg[   \max\left\{ 0, (\nabla U_n(\bm{X}), \bm{V}) \right\}\\
&\qquad \times\left[ f\left(X_1, V_1 - 2 \frac{(\nabla U_n(\bm{X}), \bm{V})}{|\nabla U_n(\bm{X})|^2}{\partial_1 U_n(\bm{X})}\right)
	-f(X_1, V_1)\right] \Bigg| \GG_0^n \Bigg] \\
&= \rE\Bigg[  \max\left\{ 0, ( \nabla U_n(\bm{X}), \bm{V}) \right\}\\
&\qquad \times
    \partial_v f(X_1, V_1) \left\{-2 \frac{(\nabla U_n(\bm{X}), \bm{V})}{|\nabla U_n(\bm{X})|^2}{\partial_1 U_n(\bm{X})}\right\}\Bigg| \GG_0^n \Bigg] + \mathcal{E}_1,
\end{align*}
where $\mathcal{E}_1$ is the remainder. At this point notice that, by the tower law and the fact that $\left(\nabla U_n(\bm{X}), \bm{V}\right)$ is Gaussian conditionally on $\bm{X}$,
\begin{align}
\nonumber \rE | \mathcal{E}_1|
&\leq \|\partial_v f\|_\infty \rE\left[ \frac{\left|\left(\nabla U_n(\bm{X}), \bm{V}\right)\right|^3
|\partial_1 U_n(\bm{X}|)}{|\nabla U_n(\bm{X})|^4} \right]\\
\nonumber&= \|\partial_v f\|_\infty \rE \left\{ \frac{|\partial_1 U_n(\bm{X})|}
{|\nabla U_n(\bm{X})|^4}  \rE\left[ \left.
\left|\left(\nabla U_n(\bm{X}), \bm{V}\right)\right|^3
\right| \bm{X}\right] \right\}\\
\nonumber
&\leq C\|\partial_v f\|_\infty \rE \left\{ \frac{|\partial_1 U_n(\bm{X})| |\nabla U_n(\bm{X})|^{3/2}}
{|\nabla U_n(\bm{X})|^4}   \right\}\\
\nonumber &\leq C\|\partial_v f\|_\infty \rE \left\{ \frac{ |\nabla U_n(\bm{X})|^{5/2}}
{|\nabla U_n(\bm{X})|^4}   \right\}
= C\|\partial_v f\|_\infty \rE \left\{ \frac{ 1}
{|\nabla U_n(\bm{X})|^{3/2}}   \right\}\\
&\leq {C\|\partial_v f\|_\infty
\left[\frac{1}{(nm)^{3/2}}+\left(\frac{\sqrt{M} }{m}\right)^{3/2} \exp\left(-\frac{n m^2}{4 M^2}\right)+\frac{m^{3/2}}{2^n M^{3/4} n}\right]}\label{eq:bndeps1}
\end{align}
by Lemma~\ref{lemma:inversemoment} in the Appendix, which tends to 0 as $n\to \infty$.


Finally, having controlled the error terms,
to complete the proof of \eqref{eq:fourthcondition}, it remains to show that the following term vanishes
\begin{align*}
&\rE_{\pi} \bigg[\left|\partial_v f(X_1, V_1)\, \right|\times \\
&\qquad \qquad \bigg| \rE\bigg[  \max\bigg\{ 0, {\left(\grad U_n(\bm{X}), \bm{V} \right)} \bigg\}
 \left( \frac{-2 {\left(\grad U_n(\bm{X}), \bm{V} \right)}}
 {{|\grad U_n(\bm{X})|^2}}\right) {\partial_1 U_n(\bm{X})}\bigg| X_1, V_1 \bigg] - W'(X_1) \bigg| \bigg].
\end{align*}
{First notice that, since $V_2, \dots, V_n$ are independent of $V_1$ and $\bm{X}$, we can write
\begin{align*}
I(X_1, V_1)
&:=\rE\bigg[  \max\bigg\{ 0, {\left(\grad U_n(\bm{X}), \bm{V} \right)} \bigg\}
 \left( \frac{{\left(\grad U_n(\bm{X}), \bm{V} \right)}}
 {{|\grad U_n(\bm{X})|^2}}\right) {\partial_1 U_n(\bm{X})}\bigg| X_1, V_1 \bigg]
\\
&=\rE\left\{ \frac{\partial_1 U_n(\bm{X})}
 {|\grad U_n(\bm{X})|^2}  \rE\left[ \left.
\max\left\{ 0, \left(\grad U_n(\bm{X}), \bm{V} \right)  \right\}^2
 \right| \bm{X}, V_1\right]
 \bigg| X_1, V_1 \right\}\\
 &= \rE\left\{ \frac{\partial_1 U_n(\bm{X})}
 {|\grad U_n(\bm{X})|^2} \rE\left[ \left.
\max\left\{ 0, \partial_1 U_n(\bm{X})V_1 + \sqrt{\sum_{j=2}^n [\partial_j U_n(\bm{X})]^2}\times \xi \right\}^2
 \right| \bm{X}, V_1\right]
 \bigg| X_1, V_1 \right\}\\
 &= \rE\left\{ \frac{\partial_1 U_n(\bm{X})}
 {|\grad U_n(\bm{X})|^2}
\max\left\{ 0, \partial_1 U_n(\bm{X})V_1 + \sqrt{\sum_{j=2}^n [\partial_j U_n(\bm{X})]^2}\times \xi \right\}^2
 \bigg| X_1, V_1 \right\},
\end{align*}
where $\xi$ is a standard Gaussian random variable, independent from $\bm{X}$ and $V_1$.
Continuing we have
\begin{align*}
I(X_1, V_1)
 &= \rE\left\{ \frac{\partial_1 U_n(\bm{X})}
 {|\grad U_n(\bm{X})|^2}
\max\left\{ 0, \sqrt{\sum_{j=2}^n [\partial_j U_n(\bm{X})]^2}\times \xi \right\}^2
 \bigg| X_1, V_1 \right\} + \mathcal{E}_2(X_1, V_1)
\end{align*}
where
\begin{align*}
\mathcal{E}_2(X_1, V_1)
&\leq C\rE\left\{ \frac{ |\partial_1 U_n(\bm{X})|^3}
 {|\grad U_n(\bm{X})|^2}
 \bigg| X_1, V_1 \right\} +
C \rE\left\{ \frac{|\partial_1 U_n(\bm{X})|^2}
 {|\grad U_n(\bm{X})|}
 \bigg| X_1, V_1 \right\}\\
 &=: \mathcal{E}_{2,1}(X_1, V_1) + \mathcal{E}_{2,2}(X_1, V_1).
\end{align*}
We control the first term using the Cauchy-Schwarz inequality as follows
\begin{align}
\nonumber\rE[\mathcal{E}_{2,1}(X_1, V_1)]
&\leq C \rE\left\{ \frac{ |\partial_1 U_n(\bm{X})|^4}
 {|\grad U_n(\bm{X})|^4}
  \right\}^{1/2}
 \rE\left\{ |\partial_1 U_n(\bm{X})|^2\right\}^{1/2}
 \intertext{and since $|\partial_1 U_n(\bm{x})|^2/|\nabla U_n(\bm{x})|^2 \leq 1$}
\nonumber&\leq  C \rE\left\{ \frac{ |\partial_1 U_n(\bm{X})|}
 {|\grad U_n(\bm{X})|} \right\}^{1/2}
 \rE\left\{ |\partial_1 U_n(\bm{X})|^2\right\}^{1/2}\\
 &\leq CM^{1/2}\left[
 \frac{M^2}{m^2 \sqrt{n}}+ \frac{M^2 \rE|X_1|}{m^{3/2}n^{1/2}}+ \frac{M}{m}\sqrt{\frac{\log n}{n}} + \frac{1}{n}\right]^{1/2}\label{eq:bndeps21}
\end{align}
by Lemmas~\ref{lem:Unpartial1divbynorm}, \ref{lemma:BPSexpectednbofbouncesupper} in the Appendix, where we used the fact that by Assumption~\ref{ass:potential} we have that $\rE|X_1|<\infty$ (this follows for example by the $\mathbb{L}^1$ Poincar\'e inequality applied on the function $f(X)=X_1$, see Corollary 1.9 of \cite{BakrySimpleProof}).
}

{
For the second error term we have, again using the Cauchy-Schwarz inequality
\begin{align}
\nonumber\rE \left[ \mathcal{E}_{2,2}(X_1, V_1)\right]
&= C \rE\left\{ \frac{|\partial_1 U_n(\bm{X})|^2}
 {|\grad U_n(\bm{X})|}
  \right\}\\
\nonumber&\leq  \rE\left\{ \frac{ |\partial_1 U_n(\bm{X})|^2}
 {|\grad U_n(\bm{X})|^2}
  \right\}^{1/2}
 \rE\left\{ |\partial_1 U_n(\bm{X})|^2\right\}^{1/2}\\
\nonumber&\leq  C \rE\left\{ \frac{ |\partial_1 U_n(\bm{X})|}
 {|\grad U_n(\bm{X})|} \right\}^{1/2}
 \rE\left\{ |\partial_1 U_n(\bm{X})|^2\right\}^{1/2}\\
 &\leq CM^{1/2}\left[
 \frac{M^2}{m^2 \sqrt{n}}+ \frac{M^2 \rE|X_1|}{m^{3/2}n^{1/2}}+ \frac{M}{m}\sqrt{\frac{\log n}{n}} + \frac{1}{n}\right]^{1/2}\label{eq:bndeps22}
 \end{align}
as before.
}

{Finally notice that
\begin{align*}
\lefteqn{\rE\bigg[  \max\bigg\{ 0, \sum_{j=2}^n \partial_j U_n(\bm{X}) V_j \bigg\}
 \left( \frac{-2 \sum_{j=2}^n \partial_j U_n(\bm{X}) V_j}
 {{|\grad U_n(\bm{X})|^2}}\right) {\partial_1 U_n(\bm{X})}\bigg| X_1, V_1 \bigg]}\\
 &=-2\rE\bigg\{ \rE\left[\left.  \frac{\max\{0, \sum_{j=2}^n \partial_j U_n(\bm{X}) V_j\}^2}
 {{|\grad U_n(\bm{X})|^2}} {\partial_1 U_n(\bm{X})} \right| \bm{X} \right]\bigg| X_1, V_1 \bigg\}\\
&=-2\rE\left\{ \left. \rE\left[\left. \mathbb{1}\{\xi >0\}
 \frac{ \left(\sum_{j=2}^n [\partial_j U_n(\bm{X})]^2 \right) \xi^2 }
 { |\grad U_n(\bm{X})|^2 } \partial_1 U_n(\bm{X}) \right| \bm{X} \right]\right| X_1, V_1 \right\}\\
 \intertext{where $\xi$ is an independent standard Gaussian}
&=-\rE\left\{ \left. {\partial_1 U_n(\bm{X})} \frac{ \left(\sum_{j=2}^n [\partial_j U_n(\bm{X})]^2 \right)  }
 { |\grad U_n(\bm{X})|^2 } \right| X_1, V_1 \right\}\\
&=-\rE\bigg\{ \partial_1 U_n(\bm{X}) \bigg| X_1, V_1 \bigg\} + \mathcal{E}_3(X_1, V_1),
\end{align*}
where
\begin{align}
\nonumber\rE[|\mathcal{E}_3(X_1, V_1)|]
&\leq \rE \left[ \frac{|\partial_1 U_n(\bm{X})|^3}{|\nabla U_n(\bm{X})|^2} \right] \\
&\leq CM^{1/2}\left[
 \frac{M^2}{m^2 \sqrt{n}}+ \frac{M^2 \rE|X_1|}{m^{3/2}n^{1/2}}+ \frac{M}{m}\sqrt{\frac{\log n}{n}} + \frac{1}{n}\right]^{1/2}, \label{eq:bndeps3}
\end{align}
by calculations similar to those for the error term $\mathcal{E}_{2,1}$.
Finally
\begin{align*}
-\rE\bigg\{{\partial_1 U_n(\bm{X})} \bigg| X_1, V_1 \bigg\}
&=-\frac{\int \frac{\partial}{\partial x_1} U_n(x_1; x_{2:n}) \re^{-U_n(x_1; x_{2:n})} \rd x_{2:n}}
{\int \re^{-U_n(x_1; x_{2:n})} \rd x_{2:n}}\\
&=-\frac{\frac{\partial}{\partial x_1}\int  U_n(x_1; x_{2:n}) \re^{-U_n(x_1; x_{2:n})} \rd x_{2:n}}
{\int \re^{-U_n(x_1; x_{2:n})} \rd x_{2:n}}\\
&=-\frac{\frac{\partial}{\partial x_1}\int   \re^{-U_n(x_1; x_{2:n})} \rd x_{2:n}}
{\int \re^{-U_n(x_1; x_{2:n})} \rd x_{2:n}}
=-\frac{\partial}{\partial x_1} \log \int   \re^{-U_n(x_1; x_{2:n})} \rd x_{2:n}\\
&=-\frac{\partial}{\partial x_1} \log \re^{-W(x_1)} = W'(x_1).
\end{align*}
}

{Overall we have shown that
\begin{align*}
\rE\bigg[  \max\bigg\{ 0, {\left(\grad U_n(\bm{X}), \bm{V} \right)} \bigg\}
 \left( \frac{-2 {\left(\grad U_n(\bm{X}), \bm{V} \right)}}
 {{|\grad U_n(\bm{X})|^2}}\right) {\partial_1 U_n(\bm{X})}\bigg| X_1, V_1 \bigg]
 &= W'(X_1) + \mathcal{E}_2+\mathcal{E}_3,
\end{align*}
where $\rE[ |\mathcal{E}_2|], \rE[ |\mathcal{E}_3|]\to 0$ as $n\to \infty$. Therefore we have
\begin{align*}
&\rE_{\pi} \bigg[\left|\partial_v f(X_1, V_1)\, \right|\times \\
&\qquad \qquad \bigg| \rE\bigg[  \max\bigg\{ 0, {\left(\grad U_n(\bm{X}), \bm{V} \right)} \bigg\}
 \left( \frac{-2 {\left(\grad U_n(\bm{X}), \bm{V} \right)}}
 {{|\grad U_n(\bm{X})|^2}}\right) {\partial_1 U_n(\bm{X})}\bigg| X_1, V_1 \bigg] - W'(X_1) \bigg| \bigg] \to 0
\end{align*}
as $n \to \infty$.
}

\subsubsection{Proof of \eqref{eq:834}.}
Next we need to verify \eqref{eq:834} for some $p>1$ for which we proceed as follows
\begin{align*}
\mathbb{E}\left[\left(\int_{0}^{T}|\phi_{n}(t)|^{p}\mathrm{d}t\right)^{1/p}\right]^p
&\leq\mathbb{E}\left[\int_{0}^{T}|\phi_{n}(t)|^{p}\mathrm{d}t\right]=\int_{0}^{T}\mathbb{E}\left[|\phi_{n}(t)|^{p}\right]\mathrm{d}t\\
&=\int_{0}^{T}\mathbb{E}\left[\left|
	\epsilon_n^{-1} \rE\left\{\left. f\left(Z_n^{(1)}(t+\epsilon_n)\right)- f\left(Z_n^{(1)}(t)\right)\right| \mathcal{G}_t^n \right\}
				\right|^{p}\right]\mathrm{d}t\\
&=\int_{0}^{T}\mathbb{E}\left[\left|
	\epsilon_n^{-1} \rE\left\{\left.
	 \int_0^{\epsilon_n} \left(\widetilde{\mathcal{A}}_n f\left(Z_n^{(1)}(t+s) \right) + R_{t+s} \right) \rd s
	\right| \mathcal{G}_t^n \right\}
				\right|^{p}\right]\mathrm{d}t\\
\intertext{and using the fact that $\rE[ R_{t+s} \mid \mathcal{G}_t^n] = \rE[ \rE[ R_{t+s} \mid \mathcal{F}_t^n ] \mid \mathcal{G}_t^n]=0$}
&=\int_{0}^{T}\mathbb{E}\left[\left|
	\epsilon_n^{-1} \int_0^{\epsilon_n}\rE\left\{\left.
	 \widetilde{\mathcal{A}}_n f\left(Z_n^{(1)}(t+s) \right)
	\right| \mathcal{G}_t^n \right\}\rd s
				\right|^{p}\right]\mathrm{d}t\\
\intertext{and by Jensen's inequality}
&\leq\int_{0}^{T}\mathbb{E}\left[
	\epsilon_n^{-1} \int_0^{\epsilon_n}\rE\left\{\left. \left|
	 \widetilde{\mathcal{A}}_n f\left(Z_n^{(1)}(t+s) \right) \right|^{p}
	\right| \mathcal{G}_t^n \right\}\rd s
				\right]\mathrm{d}t\\
&=	\int_{0}^{T}\epsilon_n^{-1} \int_0^{\epsilon_n}\mathbb{E}\left[
	\rE\left\{\left. \left|
	 \widetilde{\mathcal{A}}_n f\left(Z_n^{(1)}(t+s) \right) \right|^{p}
	\right| \mathcal{G}_t^n \right\}
				\right]\rd s\mathrm{d}t\\
&=	\int_{0}^{T}\epsilon_n^{-1} \int_0^{\epsilon_n}\mathbb{E}\left[
	 \left|
	 \widetilde{\mathcal{A}}_n f\left(Z_n^{(1)}(t+s) \right) \right|^{p}
				\right]\rd s\mathrm{d}t\\
&=T	\mathbb{E}\left[
	 \left|
	 \widetilde{\mathcal{A}}_n f\left(Z_n^{(1)}(0) \right) \right|^{p}
				\right],								
\end{align*}
by stationarity.
Next recalling the decomposition of $\widetilde{\mathcal{A}_n}$ into $\mathcal{A}_n^{(i)}$, $i=1,2,3$ notice that
$$\sup_{x,v} \left|\mathcal{A}_n^{(1)}f(x,v)\right| = \sup_{x,v}\left|\partial_x f(x,v) v\right|<\infty,$$
since $(x,v)\mapsto \partial_x f (x,v) v$ is continuous and has compact support, since $f$ has compact support. Similarly it follows easily that $\|\mathcal{A}_n^{(3)}f\|_\infty <\infty$ and therefore the only term we have to control corresponds to $\mathcal{A}_n^{(2)}$. For this term notice that
{
\begin{align*}
\lefteqn{\mathbb{E}\left[
	 \left| \mathcal{A}_n^{(2)} f\left(Z_n^{(1)}(0) \right) \right|^{p}
				\right]}\\
&= 	\mathbb{E}\left[
	 \left| \max\left\{0, \left(\nabla U_n(\bm{X}, \bm{V}\right)\right\}
	 \left[ f\left(X_1, V_1 - 2 \frac{\left(\nabla U_n(\bm{X}, \bm{V}\right)}{|\nabla U_n(\bm{X})|^2} \partial_1 U_n(\bm{X}) \right) -f(X_1, V_1)\right] \right|^{p}
				\right]\\
&\leq 2^p\|\partial_v f\|_\infty
	\mathbb{E}\left[
	 \frac{ \left| \left(\nabla U_n(\bm{X}, \bm{V}\right) \right|^{2p} |\partial_1 U_n(\bm{X})|^p}
	 { \left| \nabla U_n(\bm{X}) \right|^{2p} }\right]\\
&\leq 2^p\|\partial_v f\|_\infty
	\mathbb{E}\left\{ \rE\left[ \left.
	 \frac{ \left| \left(\nabla U_n(\bm{X}, \bm{V}\right) \right|^{2p} |\partial_1 U_n(\bm{X})|^p}
	 { \left| \nabla U_n(\bm{X}) \right|^{2p} }
	 \right| \bm{X} \right]
	 \right\}\\
&\leq 2^p\|\partial_v f\|_\infty
	\mathbb{E}\left\{ \rE\left[ \left.
	 \frac{ \left| \nabla U_n(\bm{X}) \right|^{p} |\partial_1 U_n(\bm{X})|^p}
	 { \left| \nabla U_n(\bm{X}) \right|^{2p} }
	 \right| \bm{X} \right]
	 \right\}\leq {2^p\|\partial_v f\|_\infty = O(1).}
\end{align*}
}

\subsubsection{Proof of \eqref{eq:firstcondition} and \eqref{eq:secondcondition}}
Notice that \eqref{eq:firstcondition} follows immediately since $\|f\|_\infty < \infty$, whereas \eqref{eq:secondcondition}  follows from calculations similar to the ones used to prove
\eqref{eq:834}.
\section{Proofs of Wasserstein rates}
\subsection{Proof of Theorem~\ref{thm:Wasserstein}}

Let $\tx(t):=X^{(2)}(t)-X^{(1)}(t)$ and $\tv(t):=V^{(2)}(t)-V^{(1)}(t)$ denote the differences between the two paths in position and momentum. Ignoring for the moment the refreshment events, $(\tx(t), \tv(t))$ will evolve according to the Hamiltonian dynamics, that is
\begin{equation}\label{eq:Hamcouplingdynamicsforward}
\begin{aligned}
\tx'(t)&=\tv(t),\\
\tv'(t)&=-(\grad U( X^{(2)}(t) )-\grad U( X^{(1)}(t) ))=-H(t) \tx(t),\text{ where}\\
\mathcal{H}(t)&:=\int_{s=0}^{1} \grad^2 U(sX^{(1)}(t)+ (1-s)X^{(2)}(t)) \rd s.
\end{aligned}
\end{equation}
By convexity, we can see that $\mathcal{H}(t)$ satisfies that $m I\preceq \mathcal{H}(t)\preceq MI$  where $I$ denotes the identity matrix, where we write
$A\preceq B$ to denote that $B-A$ is positive definite. The effect of the generator $L_{1,2}$ on $|\tx(t)|^2$, $\langle \tx(t), \tv(t)\rangle$ and $|\tv(t)|^2$ is given by
\begin{equation}\label{eq:L12forward}
\begin{aligned}
L_{1,2} |\tx(t)|^2&=2\inner{\tx(t)}{\tv(t)},\\
L_{1,2}{\tx(t)}^T{\tv(t)}&=|\tv(t)|^2-{\tx(t)}^T{\mathcal{H}(t)\tx(t)}-\lref (1-\alpha){\tx(t)}^T{\tv(t)},\\
L_{1,2} |\tv(t)|^2&=-2{\tv(t)}^T{\mathcal{H}(t)\tx(t)}-\lref (1-\alpha^2)|\tv(t)|^2.
\end{aligned}
\end{equation}
	The claim of Theorem~\ref{thm:Wasserstein} is equivalent to showing that $-\mu\cdot d_A^2(Z_1(t),Z_2(t))-L_{1,2} d_A^2(Z_1(t),Z_2(t))\ge 0$. This can be expressed as
	\begin{align*}
	&-\mu\cdot d_A^2(Z_1(t),Z_2(t))-L_{1,2} d_A^2(Z_1(t),Z_2(t))\\
	&=
	-\mu a |\tx(t)|^2 + 2[-\mu b+\lref (1-\alpha)b -a]{\tx(t)}^T{\tv(t)}+[-c \mu+\lref (1-\alpha^2)c-2b]|\tv(t)|^2\\
	&+2b {\tx(t)}^T{\mathcal{H}(t)\tx(t)}+2c {\tv(t)}^T{\mathcal{H}(t)\tx(t)}.
	\end{align*}
	
	Let \begin{align*}
	X&:=\l(\begin{matrix}|\tx(t)|^2 &{\tx(t)}^T{\tv(t)}\\ {\tx(t)}^T{\tv(t)} & |\tv(t)|^2\end{matrix}\r), \quad
	P:=\l(\begin{matrix}{\tx(t)}^T{\mathcal{H}(t)\tx(t)} &{\tv(t)}^T{\mathcal{H}(t)\tx(t)}\\ {\tv(t)}^T{\mathcal{H}(t)\tx(t)} &{\tv(t)}^T{\mathcal{H}(t)\tv(t)} \end{matrix}\r),\\
	V&:=\l(\begin{matrix}-\mu a & - a+b \lref(1-\alpha)-\mu b\\- a+b \lref(1-\alpha)-\mu b & -c \mu+c \lref (1-\alpha^2)-2b \end{matrix}\r),
	\quad W:=\l(\begin{matrix}2 b & c\\c & 0 \end{matrix}\r).
	\end{align*}
	We have
	\[-\mu\cdot d_A^2(Z_1(t),Z_2(t))-L_{1,2} d_A^2(Z_1(t),Z_2(t))=\Tr(VX+WP),\]
	so our goal is to show that $\Tr(VX+WP)\ge 0$ for all the possible $X, P$.
	Using the fact that $m I\preceq \mathcal{H}(t)\preceq M I$, we have $0\preceq m X\preceq P\preceq M X$.
	Let $Y:=P-mX$, and $Z:=M X-P$, then $Y\succeq 0$, $Z\succeq 0$, and for $M>m$, we have
	\begin{align*}&X=\frac{Y+Z}{M-m}, \quad P=\frac{MY + m Z}{M-m},
	\intertext{ and hence }
	&\Tr(VX+WP)=\frac{1}{M-m}\l(\Tr((V+MW)Y+(V+mW)Z) \r).
	\end{align*}
	When $M=m$, we have $\mathcal{H}(t)=M I$ and $P=MX$, hence
	\[\Tr(VX+WP)=\Tr((V+MW) X).\]
	Note that in both cases, $\Tr(VX+WP)\ge 0$ if both $V+MW\succeq 0$ and $V+mW\succeq 0$. This can be equivalently written as the following set of inequalities,
	\begin{align}
	\label{ineq1}-\mu a+2M b&\ge 0,\\
	-\mu a+2m b&\ge 0,\\
	-c \mu+c \lref (1-\alpha^2)-2b&\ge 0,\\
	(-a+b \lref(1-\alpha)-\mu b +M c)^2 &\le  (-\mu a+2M b)(-c \mu+c \lref (1-\alpha^2)-2b),\\
	\label{ineq5}(-a+b \lref(1-\alpha)-\mu b +m c)^2 &\le  (-\mu a+2m b)(-c \mu+c \lref (1-\alpha^2)-2b).
	\end{align}
	These inequalities correspond to the diagonal elements and the determinants of $V+mW$ and $V+MW$ being non-negative. As we have stated, let $\lref=\frac{1}{1-\alpha^2}\l(2\sqrt{M+m}-\frac{(1-\alpha) m}{\sqrt{M+m}}\r)$, $\mu=\frac{(1+\alpha)m}{\sqrt{M+m}}-\frac{\alpha m^{3/2}}{2(M+m)}$. Moreover, 	let 	
	\begin{equation}\label{eq:Wassabcdef}
	\begin{split}
	a&:=1,\\
	b&:=\frac{1+\alpha-\alpha \l(\frac{m}{M+m}\r)^{3/4}+\frac{3}{4} \frac{\alpha m}{M+m}}{2\sqrt{M+m}},\\
	c&:=\frac{1+\alpha- \frac{\alpha }{2}\l(\frac{m}{M+m}\r)^{1/2} }{M+m}.
	\end{split}
	\end{equation}
	Notice that by the change of variables $m\to 1$, $M\to M/m$, and updating $a$, $b$, $c$ and $\mu$ and $\lambda$ with these new values, inequalities \eqref{ineq1}-\eqref{ineq5} are kept 		invariant (they have this homogeneity property). Hence, without loss of generality, we can assume that $m=1$.
	For the choice of $a,b,c$ as in \eqref{eq:Wassabcdef}, the five inequalities can be shown to hold for every possible $0\le \alpha<1$ and $M$ using for example Mathematica.
	Hence the bound \eqref{eq:Wassersteincontraction1} follows.
	
	Now we are going to show the Wasserstein bounds. Note that the matrix $A$ satisfies that $\lambda_{\min}(A)=\frac{a+c-\sqrt{(a+c)^2-4(ac-b^2)}}{2}$ and $\lambda_{\max}(A)=\frac{a+c+\sqrt{(a+c)^2-4(ac-b^2)}}{2}$, hence by defining
	\[W_{2,d_A}(\nu_1,\nu_2)=\l(\inf_{X_1\sim \nu_1,X_2\sim \nu_2} d_A(X_1,X_2)^2 \r)^{1/2},\] then using the assumption $b^2<ac$, we have $W_2(\nu P^t,\pi)^2\le \frac{1}{\lambda_{\min}(A)} W_{2,d_A}(\nu P^t,\pi)^2$.
Let $Z_1(0), Z_2(0)$ be coupled according to the optimal coupling of $\nu$ and $\pi$ according to $W_2$ distance satisfying that $\rE(|Z_1(0)-Z_2(0)|^2)=W_2(\nu,\pi)^2$ (existence is shown by Theorem 4.1 of \cite{Villani2008}). Using \eqref{eq:Wassersteincontraction1} along with Gr\"onwall's lemma, and the definition of the Wasserstein distance, it follows that
	\begin{align*}
	W_2(\nu P^t,\pi)^2&\le \frac{1}{\lambda_{\min}(A)} W_{2,d_A}(\nu P^t,\pi)^2\le \frac{1}{\lambda_{\min}(A)} \rE( d_A^2(Z_1(t),Z_2(t)))\\
	&\le   \frac{e^{-\mu t}}{\lambda_{\min}(A)}\rE(d_A^2(Z_1(0), Z_2(0) )\le  \frac{\lambda_{\max}(A)}{\lambda_{\min}(A)} e^{-\mu t} W_{2}(\nu,\pi)^2, 	
	\end{align*}
	hence \eqref{eqWass2conv} follows.
	
	To show our $L^2$ bounds, we are also going to study the adjoint process $(P^t)^*$. Using the exact same coupling as before, the dynamics \eqref{eq:Hamcouplingdynamicsforward} ran backwards in time becomes
	\begin{equation}\label{eq:Hamcouplingdynamicsbackward}
	\begin{aligned}
	\tx'(t)&=-\tv(t),\\
	\tv'(t)&=\mathcal{H}(t) \tx(t),\\
	\end{aligned}
	\end{equation}
	with $\mathcal{H}(t)$ defined as in \eqref{eq:Hamcouplingdynamicsforward}. For the velocity updates, forward in time we had $v'=\alpha v+\sqrt{1-\alpha^2}Z$ where $Z\sim N(0,I_d)$. Since in stationary we have $v,v'\sim N(0,I_d)$ and $\rE(v (v')^T)=\rho I_d$, one can see that the updates backward in time are still the same.
	Hence the effect of the adjoint becomes
	\begin{equation}\label{eq:L12backward}
	\begin{aligned}
	L_{1,2}^* |\tx(t)|^2&=-2{\tx(t)}^T{\tv(t)},\\
	L_{1,2}^*{\tx(t)}^T{\tv(t)}&=|\tv(t)|^2-{\tx(t)}^T{\mathcal{H}(t)\tx(t)}+\lref (1-\alpha){\tx(t)}^T{\tv(t)},\\
	L_{1,2}^* |\tv(t)|^2&=2{\tv(t)}^T{\mathcal{H}(t)\tx(t)}-\lref (1-\alpha^2)|\tv(t)|^2.
	\end{aligned}
	\end{equation}
	Notice that this is very similar to the forward case \eqref{eq:L12forward}, except that we need to replace $\tv(t)$ by $-\tv(t)$. Based on this, by repeating the previous argument for $A':=\l(\begin{matrix}a &-b\\ -b &c\end{matrix}\r)$, we have
	\begin{equation}\label{eq:Wassersteincontraction2}
	L_{1,2}^* \ d_{A'}^2(Z_1(t),Z_2(t))\le -\mu\cdot d_{A'}^2(Z_1(t),Z_2(t)),
	\end{equation}
	where $a$, $b$ and $c$ are defined as in \eqref{eq:Wassabcdef}.
	
	Hence we have shown that the adjoint process is also a contraction with the same rate $\mu$, but with respect to a different metric $d_{A'}$ instead of $d_A$ used for the forward process. 
	Now we are going to show that $d_A^2$ and $d_{A'}^2$ are equivalent up to a constant factor $C:=\frac{ac+b^2+2\sqrt{ac b^2}}{ac-b^2}$.
	Notice that for any $z_1,z_2\in \mathbb{R}^{2d}$,
	\begin{equation}\label{eqdAApequivalence}
	d_A^2(z_1,z_2)/C\le d_{A'}^2(z_1,z_2)\le d_A^2(z_1,z_2)\cdot C,
	\end{equation}
	as long as $A\preceq C A'$ and $A' \preceq C A$, and by rearrangement, this is equivalent to
	\[\l(\begin{matrix}a(C-1)& -b(1+C)\\-b(1+C)& c(C-1) \end{matrix}\r)\succeq 0 \text{ and } \l(\begin{matrix}a(C-1)& b(C+1)\\b(C+1)& c(C-1) \end{matrix}\r)\succeq 0,\]
	which holds for $C$ defined as above.

	For $f:\mathbb{R}^{2d}\to \mathbb{R}$, let
	$$\|f\|_{\mathrm{Lip},d_A}:=\sup_{z_1,z_2\in \mathbb{R}^{2d}, z_1\ne z_2}\frac{|f(z_1)-f(z_2)|}{d_{A}(z_1,z_2)},$$
	be its Lipschitz coefficient with respect to the $d_A$ distance.
	Then based on \eqref{eq:Wassersteincontraction1},\eqref{eq:Wassersteincontraction2}, and \eqref{eqdAApequivalence}, for any $t\ge 0$, $f:\mathbb{R}^{2d}\to \mathbb{R}$, have
	\begin{align*}
	\|(P^t)^* P^t f \|_{\mathrm{Lip},d_A}&\le \sqrt{C}\|(P^t)^* P^t f \|_{\mathrm{Lip},d_{A'}}\le \sqrt{C}\exp\l(-\frac{\mu t}{2}\r) \|P^t f\|_{\mathrm{Lip},d_{A'}}\\
	&\le C\exp\l(-\frac{\mu t}{2}\r) \|P^t f\|_{\mathrm{Lip},d_{A}}
	\le C\exp\l(-\mu t\r) \|f\|_{\mathrm{Lip},d_{A}}.
	\end{align*}
	Based on Propositions 29 and 30 of \cite{OllivierJFA} with $\kappa=1-C\exp\l(-\mu t\r)$, it follows that for any $t>\frac{\log(C)}{\mu}$, the reversible kernel $(P^t)^* P^t$ has as spectral radius of at most $C\exp\l(-\mu t\r)$. Thus for every $f\in L_0^{2}(\pi)$, we have
\begin{equation}
\|P^t f \|^2=\inner{f}{(P^t)^* P^t f} \leq \|f\| \|(P^t)^* P^t f\|\le C e^{-\mu t} \|f\|^2,
\end{equation}
	and the claim of the Theorem follows by noticing that $\|P^t f \|^2\le \|f\|^2$ for every $t\ge 0$.
	
\begin{remark}
We note that for any given $\lref>0, \mu>0$, the contraction rate of $d_A^2(Z_1(t),Z_2(t))$ is at least $\mu$ as long as there are constants $a,b,c$ such that $a>0$, $c>0$, $b^2<ac$ and inequalities \eqref{ineq1}-\eqref{ineq5} hold. Unfortunately due to the non-linearity of these inequalities we did not manage to find an analytical expression for the largest possible $\mu$ for a given $\lref$ (and then the largest possible $\mu$ for any $\lref$). The reader can possibly slightly improve these rates by numerical optimization for a given $\alpha$, $m$ and $M$. Note however that in our numerical experiments, it seems that the choices of $\lref$ as stated leads to $\mu$ that is close to optimal in most of the domain $0\le \alpha<1$, and $0<m\le M$ (i.e. if we increase $\mu$ by a few percent, typically there is no longer a $\lref>0$ and parameters $a,b,c$ satisfying all of the inequalities).
\end{remark}
\subsection{Proof of Proposition~\ref{prop:WassersteinGauss}}
Assume without loss of generality that $m=1$ (the general case can be obtained from this  by rescaling). Let $D:=\l(\begin{matrix}aH & bI\\b I & c I \end{matrix}\r)$ be a block matrix.
Then
\[d_{D}^2(Z_1(t),Z_2(t))=a{\tx(t)}^T{H \tx(t)}+2b{\tx(t)}^T{\tv(t)}+ c|\tv(t)|^2,\]
and the effect of the generator on these terms equal
\begin{align}
L_{1,2} {\tx(t)}^T{H \tx(t)}&=2{\tx(t)}^T{H\tv(t)},\\
L_{1,2}{\tx(t)}^T{\tv(t)}&=|\tv(t)|^2-{\tx(t)}^T{H\tx(t)}-\lref (1-\alpha){\tx(t)}^T{\tv(t)},\\
L_{1,2} |\tv(t)|^2&=-2{\tv(t)}^T{H\tx(t)}-\lref (1-\alpha^2)|\tv(t)|^2.
\end{align}
We have
\begin{align*}
&-\mu\cdot d_D^2(Z_1(t),Z_2(t))-L_{1,2} d_D^2(Z_1(t),Z_2(t))\\
&=
2[-\mu b+\lref (1-\alpha)b]{\tx(t)}^T{\tv(t)}+[-c \mu+\lref (1-\alpha^2)c-2b]|\tv(t)|^2\\
&+(2b-\mu a) {\tx(t)}^T{H\tx(t)}+2(c-a) {\tv(t)}^T{H\tx(t)}.
\end{align*}
Let $X$ and $P$ defined as in the proof of Theorem \ref{thm:Wasserstein}, and let
\begin{align*}
V&:=\l(\begin{matrix}0 & b \lref(1-\alpha)-\mu b\\ b \lref(1-\alpha)-\mu b & -c \mu+c \lref (1-\alpha^2)-2b \end{matrix}\r),
\quad W:=\l(\begin{matrix}2 b-\mu a & c-a\\c-a & 0 \end{matrix}\r).
\end{align*}
Then we have $-\mu\cdot d_D^2(Z_1(t),Z_2(t))-L_{1,2} d_D^2(Z_1(t),Z_2(t))=\Tr(VX+WP)$, and using the same argument as in the proof of Theorem \ref{thm:Wasserstein},
it follows that $\Tr(VX+WP)\ge 0$ if both $V+MW\succeq 0$ and $V+mW\succeq 0$. This can be verified (for example by Mathematica) for the choices
$\lref=2\sqrt{m}/(1-\alpha)$,  $\mu=\frac{\sqrt{m}}{3}$, $a=1$, $b=\frac{1}{4}$, $c=1$.
The proof of \eqref{eqL2bndGaussian} is analogous to the proof of \eqref{eqL2bndWass}. First we show that for $D':=\l(\begin{matrix}aH & -bI\\-b I & c I \end{matrix}\r)$,
\begin{equation}
L_{1,2}^* \ d_{D'}^2(Z_1(t),Z_2(t))\le -\mu\cdot d_{D'}^2(Z_1(t),Z_2(t)),
\end{equation}
then use the same argument as previously.
\section{Proof of Theorem~\ref{thm:hypoco}}

The generator of the RHMC process will be denoted by $\mathcal{A}$ and it is given for smooth enough functions by
$$\mathcal{A}f(x,v) = \langle \nabla_x f, v\rangle - \langle \nabla U, \nabla_v f\rangle
+ \lref \left[Q_\alpha f(x,v) - f(x,v)\right],$$
where recall that $\alpha \in (0,1)$ and
$$Q_\alpha f(x,v):= \frac{1}{\sqrt{2\pi}^d}\int \re^{-\bm{\xi}'\bm{\xi}/2}
f\left(x,\alpha v+\sqrt{1-\alpha^2} \xi\right) \rd \bm{\xi}.$$

\paragraph{Hypo-coercivity, Exponential Convergence and Asymptotic Variance.}
In the context of MCMC one is interested in optimising the computational resources
needed to produce an estimate of a certain precision. For this reason we are also interested in understanding the asymptotic variance.
Geometric ergodicity is enough to show that a large class of functions, determined by the Lyapunov function, have finite asymptotic variance.
However, since the convergence rates are not explicit in the parameters of the process, geometric ergodicity often does not allow one to optimise the asymptotic variance.

{Usually controlling the asymptotic variance for a large enough class of functions is closely related to establishing a \textit{spectral gap}, that is showing that the $L^2(\pi)$ spectrum of the generator $\mathcal{L}$ lies in
$\{z\in \mathds{C}: \Re{z} \leq -\mu\}$, for some $\mu>0$.
In the reversible case, it is well known that geometric ergodicity is equivalent to having a spectral gap,
but in the non-reversible case this is no longer true, see \cite{kontoyiannis2012} and references therein (although it may be equivalent to a spectral gap on a different Banach space).
For reversible processes, an $L^2$-spectral gap is also equivalent to \textit{coercivity} of the associated Dirichlet form, that is $\langle -\mathcal{L} f, f\rangle \geq \mu \|f\|^2$, for all $f\in L_0^2(\pi)$. Moreover, coercivity is equivalent to $\|P^t f\|\leq \re^{-\mu t} \|f\|$, for all $f\in L_0^2(\pi)$, for \textit{all} Markov processes, whether reversible or not. For this reason, and perhaps abusively, coercivity is sometimes in the literature referred to as a spectral gap, or a spectral gap inequality. Another reason is that, an inequality of the form $\langle -\mathcal{L} f, f\rangle \geq \mu \|f\|^2$ is often easy to prove, e.g. for diffusions, by rewriting the Dirichlet form in a form involving the Sobolev norm and then applying a Poincar\'e inequality. }

{Interestingly enough however, for non-reversible processes it is possible that coercivity fails to hold, although we still have $\|P^t f\|\leq C\re^{-\mu t} \|f\|$, for all $f\in L_0^2(\pi)$, for some $C>1$. This is not possible for reversible processes, since one can use spectral calculus to show that $\|P^t f\|\leq C\re^{-\mu t} \|f\|$, for all $f\in L_0^2(\pi)$ also implies the same inequality with $C\equiv 1$.
This fact is actually observed for piecewise deterministic Markov processes such as the BPS and Zig-Zag samplers,
see \cite{peters2012rejection,bouncy2018,bierkens2016zig}. This class of processes also includes RHMC.
Although geometric ergodicity has been established for BPS (\cite{deligiannidis2017exponential,durmus2018geometric}), Zig-Zag (see \cite{bierkens2017ergodicity,fetique2017long}) and RHMC (\cite{RHMC}),
an easy calculation shows that, writing $\mathcal{L}$ for the generator of any of the above processes, we have $\langle \mathcal{L}f, f\rangle = 0$ for any function $f\in L^2(\pi)$ such that $f(x,v)= f(x)$, that is functions of the location only. The reason for this is that the Dirichlet form $\mathcal{E}(f,f):=\langle \mathcal{L}f, f \rangle$ only captures the symmetric part of the generator $\mathcal{L}$, which in these processes only affects the velocity component, whereas the location component is only affected by the anti-symmetric part of the generator. This means that although BPS, Zig-Zag and RHMC are geometrically ergodic, we certainly cannot have an inequality of the form $\langle -\mathcal{L} f, f\rangle \geq \mu \|f\|^2$ for all $f\in L_0^2(\pi)$. However, it may still be true that these processes admit a spectral gap in the classical sense, and in fact this has been shown for one-dimensional Zig-Zag in \citet{bierkens2019spectral}. Notice however, that in the non-reversible case, a classical spectral gap requires additional work, and potentially assumptions, to guarantee exponential decay of the semigroup, see \cite[Section~4.2]{bierkens2019spectral}.}

In fact this situation arises very often in so called kinetic equations which include for example the underdamped Langevin processes. For such processes a range of methods have been developed recently that are widely termed as \textit{hypocoercivity},
see \cite{nier2005hypoelliptic,dric2009hypocoercivity,dolbeault2015hypocoercivity} and references therein. In fact such methods have already been applied to piecewise deterministic Markov processes, see \cite{monmarche2014hypocoercive}.
Although this approach is often quite deep and involved, the underlying principle is that of adjusting the norm, or metric, in which the convergence is studied. This principle has been extremely successful recently, for example in the convergence of HMC when log-concavity fails locally in \cite{bou2018coupling}. In the case of hypocoercive estimates, the principle is to move from the $L^2$ norm to a stronger norm, usually some form of Sobolev norm.

\subsection{Strong continuity in $H^1(\pi)$. }
We will establish that the abstract Cauchy problem
\begin{align*}
\frac{\partial u(t,z) }{\partial t} &= \mathcal{A} u,\\
u(0,z) &= f,
\end{align*}
where the class of initial conditions $f$ will be specified in the sequel, admits a unique solution in $H^1(\pi)$ given by $u(t,z):= P^t f(z)$. This will justify computing the time derivatives of $\lann P^t f, P^t f \rann$.

Before we proceed we will need to introduce some additional notation. We decompose the generator $\mathcal{A}$ of RHMC into its symmetric and antisymmetric component as follows
$$\mathcal{A}f(x,v) = Bf(x,v) + \lref (-S) f,$$
where
\begin{equation}\label{eq:definitionofB}
Bf := \langle \nabla_x f, v\rangle - \langle \nabla_v f, \nabla U \rangle, \qquad Sf := [I-Q_\alpha]f.
\end{equation}
As before we write $\{P^t: t\geq 0\}$ for the semi-group of transition kernels of RHMC, but in this section we slightly change our point of view and consider it as a semigroup on $L^2(\pi)$, that is
$P^t: L^2(\pi)\to L^2(\pi)$. Its generator will be given by $\mathcal{A}$ for smooth enough functions.

In fact even more is true as we will next show that $P^t$ is also strongly continuous as a semi-group on $H^1(\pi)$. To see why, first recall that the anti-symmetric operator $B$ generates the Hamiltonian flow $z\mapsto \Xi(t,z)$ with respect to $H(\bm{x},\bm{v}) = U(\bm{x}) + |\bm{v}|^2/2$. Let us write $\{T^t:t\geq 0\}$ for the semigroup generated by $B$, that is
$T^t f(z) = f\left( \Xi(t,z)\right)$ for $z\in \mathcal{Z}$. Then
given a smooth function $f\in H^1(\pi)$, from the chain rule we have
$$\nabla T^t f(z) = \nabla f\left( \Xi(t,z)\right) \nabla_z\Xi(t,z).$$
From the variational equations of the Hamiltonian dynamics (see Section 6.1.2 of \cite{IntroductiontoHamiltonian}) and the upper bounds $M$ and $1$ of the Hessians of $U(\xx)$ and $\frac{\|\vv\|^2}{2}$ it follows that for $C=\max(1,M)$, we have
 $\|\nabla_z\Xi(t,z)\|\leq \re^{Ct}$ for every $t\ge 0$. Using this, we conclude that
 \begin{align*}
\|\nabla_x T^t f\|^2 + \|\nabla_v T^t f\|^2
&\leq \re^{2Ct} \iint \pi(\rd z)\left[ \left|\nabla_x f \left(\Xi(t,z)\right) \right|^2
	+\left|\nabla_v f \left(\Xi(t,z)\right) \right|^2\right]\\
&=	\re^{2Ct} \iint \pi(\rd z)\left[ \left|\nabla_x f \left(z\right) \right|^2
	+\left|\nabla_v f \left(z\right) \right|^2\right],
\end{align*}
by stationarity of the flow.  By an approximation argument we can further show that
$T^t: H^1(\pi)\to H^1(\pi)$ for all $t\geq 0$. Finally $\{T^t:t\geq 0\}$ is strongly continuous on $H^1(\pi)$, since
\begin{align}
\| \nabla T^s f - \nabla f\|^2
&= \int| \nabla f\left( \Xi(s, z) \right)\nabla_z \Xi(s,z) - \nabla f\left(z \right)|^2 \pi(\rd z)\notag\\
&\leq  \int| \nabla f\left( \Xi(s, z) \right) \left[ \nabla_z \Xi(s,z) - I\right] |^2 \pi(\rd z) \notag\\
&\qquad +   \int| \nabla f\left( \Xi(s, z) \right) - \nabla f\left(z \right)|^2 \pi(\rd z)\notag\\
&\leq  \int| \nabla f\left( \Xi(s, z) \right)|^2 |\nabla_z \Xi(s,z) - I |^2 \pi(\rd z) \notag\\
&\qquad +   \int| \nabla f\left( \Xi(s, z) \right) - \nabla f\left(z \right)|^2 \pi(\rd z)\notag\\
&\leq  \int| \nabla f\left( \Xi(s, z) \right)|^2 |\nabla_z \Xi(s,z) - I |^2 \pi(\rd z) \notag\\
&\qquad +   2\int|T^s \nabla f(z) - \nabla f(z)|^2 \pi(\rd z).\label{eq:strongcontinuityinH1}
\end{align}
Since $g:=\nabla f\in L^2(\pi)$, for every $\epsilon>0$ there is a smooth, compactly supported function $g_\epsilon$ such that $\|g-g_\epsilon\|_{L^2(\pi)}<\epsilon$. Then
\begin{align*}
\int|T^s g(z) - g(z)|^2 \pi(\rd z)
&= \int|T^s g(z) - T^s g_\epsilon(z)+ T^s g_\epsilon(z) -  g_\epsilon(z) + g_\epsilon(z)- g(z)|^2 \pi(\rd z)\\
&\leq \int \pi(\rd z) \left| g\left(\Xi(s,z)\right)-g_\epsilon\left(\Xi(s,z)\right)\right|^2
	+ \int \pi(\rd z) \left| g\left(z\right)-g_\epsilon\left(z\right)\right|^2\\
	&\qquad + \int \pi(\rd z) \left| g_\epsilon\left(\Xi(s,z)\right)-g_\epsilon\left(z\right)\right|^2\\
&= 2 \|g-g_\epsilon\|+  \int \pi(\rd z) \left| g_\epsilon\left(\Xi(s,z)\right)-g_\epsilon\left(z\right)\right|^2\\
&\leq 2 \epsilon +  \int \pi(\rd z) \left| g_\epsilon\left(\Xi(s,z)\right)-g_\epsilon\left(z\right)\right|^2.
\end{align*}
For every fixed $\epsilon >0$, the second term vanishes by bounded convergence. Since $\epsilon>0$ is arbitrary this shows that $\|T^s \nabla f - \nabla f\|^2  \to 0$ as $s \to 0$.

Going back to \eqref{eq:strongcontinuityinH1}, notice that the first term also vanishes  by the dominated convergence theorem, since $|\nabla_z \Xi(s,z) - I| \leq 2\re^{Cs}$ uniformly in $z$,  $|\nabla_z \Xi(s,z) -I | \to 0$ pointwise.
Thus $T^t$ is strongly continuous and therefore it admits a densely defined generator, which we denote by $B$,
$$B: \mathcal{D}(B) \subseteq H^1(\pi) \to H^1(\pi).$$
Again it is straightforward to check that $B$ has the expression given earlier.

In addition notice that $S$ is a bounded operator on $H^1(\pi)$. To see why first notice that an easy calculation, which will be provided later on in Section \ref{sec:Proofthmhypoco} for completeness,
shows that $\nabla_x Q_\alpha = Q_\alpha \nabla_x$ and $\nabla_v Q_\alpha = \alpha Q_\alpha \nabla_v$ whence
$$\|\nabla_x Q_\alpha f\|^2+ \|\nabla_v Q_\alpha f\|^2 \leq
\| Q_\alpha  \nabla_xf\|^2+ \alpha \|Q_\alpha \nabla_v f\|^2 \leq C \left( \|\nabla_x f\|^2+ \|\nabla_v f\|^2\right),$$
since $Q_\alpha$ is a contraction on $L^2(\pi)$.
Therefore, applying \cite[Theorem~3.2]{phillips1953perturbation}, the operator
$\mathcal{A}:=B+\lref(-S)$ has domain $\mathcal{D}(B)$ and generates a strongly continuous on $H^1(\pi)$, which we will denote again by $\{P^t:t\geq 0\}$. This implies that for every $f\in \mathcal{D}(B)$, $P^t f \in \mathcal{D}(\mathcal{A})$ for all $t\geq 0$ and $\mathcal{A}P^t f= P^t \mathcal{A} f$.
This essentially shows that given $f \in \mathcal{D}(B)$ the abstract Cauchy problem
\begin{align*}
\frac{\partial u(t,z) }{\partial t} &= \mathcal{A} u,\\
u(0,z) &= f,
\end{align*}
admits a unique solution in $H^1(\pi)$ given by $u(t,z):= P^t f(z)$.


\subsection{Proof of Theorem~\ref{thm:hypoco}.}\label{sec:Proofthmhypoco}
We introduce some additional notation to keep the presentation concise.  First recall the decomposition $\mathcal{A} = B +\lref(-S)$ where
$$
Bf = \langle \nabla_x f, v\rangle - \langle \nabla_v f, \nabla U \rangle, \qquad Sf = [I-Q_\alpha]f,$$
and let us define the Dirichlet form $\mathcal{E}(f,g) := \langle f, S g\rangle$.
We will also write $A:= \nabla_v$, $C:= \nabla_x$.
From \cite[p. 40]{dric2009hypocoercivity}, or an easy calculation,
we have
$$[A,B]=AB-BA=\nabla_x, \quad\text{ and} \qquad [B,C]= \nabla^2 U \cdot \nabla_v = \nabla^2 U \cdot A.$$

Since $P^t = \exp (t \mathcal{A})$, where $\mathcal{A}$ is the generator of the RHMC process, an easy calculation shows that for all $f, g \in \mathcal{D}(\mathcal{B})$ we have
\begin{align*}
\frac{\rd}{\rd t }\langle P^t f, P^t g \rangle  \Big|_{t=0}
&=\langle \mathcal{A}f, g\rangle +\langle f, \mathcal{A}g\rangle,
\end{align*}
This also implies that
$$\frac{\rd}{\rd t }\langle P^t f, P^t f\rangle \Big|_{t=0} = 2\langle \mathcal{A}f, f\rangle
=-2\lref \mathcal{E}(f,f).$$
since $B$ is antisymmetric, in the sense that
$\langle Bf, g\rangle=-\langle f, Bg\rangle$.

We want to compute $\rd \lann P^t f, P^t f\rann/\rd t|_{t=0}$.
To keep notation to a minimum we will write $h$ rather than $P^t f$.
We proceed by computing the
derivative of each term individually,
\begin{align*}
 \frac{\rd}{\rd t }\|A h\|^2
 &= 2\langle A h,  A \mathcal{A}h\rangle = -2\lref\langle A h,  A Sh\rangle + 2\langle A h,  A Bh\rangle, \\
 \frac{\rd}{\rd t }\langle C h, A h\rangle
 &= \langle C h, A (-\lref S+B) h\rangle + \langle C (-\lref S+B)h, A h\rangle,\\
  \frac{\rd}{\rd t }\|C h\|^2
 &= 2\langle C h, C \mathcal{A} h\rangle = -2\lref\langle C h, C S h\rangle + 2\langle C h, C B h\rangle.
\end{align*}
\paragraph{Term one. }
We now compute the first term which is given by
$$-2\lref\langle A h,  A Sh\rangle + 2\langle A h,  A Bh\rangle.$$
Notice that
\begin{align*}
\frac{\partial}{\partial v_i} Q_\alpha f(\bm{x},\bm{v})
&= \frac{\partial}{\partial v_i} \rE \left[  f\left(\bm{x},\alpha \bm{v}+ \sqrt{1-\alpha^2} \bm{\xi}\right)\right]\\
&=  \rE \left[ \alpha f_{v_i}\left(\bm{x},\alpha \bm{v}+ \sqrt{1-\alpha^2} \bm{\xi}\right)\right]\\
&= \alpha  \rE \left[ f_{v_i}\left(\bm{x},\alpha \bm{v}+ \sqrt{1-\alpha^2} \bm{\xi}\right)\right],
\end{align*}
where to keep notation clear we write
$\partial G(\bm{x},\bm{v}) /\partial v_i$ to denote the derivative of the expression $G(x,v)$ w.r.t.\ $v_i$, whereas we write $f_{v_i}\left(\bm{x},\alpha \bm{v}+ \sqrt{1-\alpha^2} \bm{\xi}\right)$ to denote the derivative
of $f$ w.r.t.\ $v_i$ evaluated at $\alpha \bm{v}+ \sqrt{1-\alpha^2} \bm{\xi}$.

The above calculation shows that $A Q_\alpha = \alpha Q_\alpha A$ and therefore
\begin{align*}
-\lref \langle A h,  A S h\rangle
&= \lref \langle A h, A (Q_\alpha -I) h\rangle\\
&= \lref \langle A h, A Q_\alpha h\rangle - \lref \langle A h, A h\rangle\\
&= \lref \alpha \langle A h, Q_\alpha A h\rangle - \lref \langle A h, A h\rangle\\
&=\lref\langle A h, (\alpha Q_\alpha -I) A  h\rangle\\
&=\lref\langle A h, \alpha( Q_\alpha -I) A  h\rangle
-(1-\alpha)\lref\langle A h,  A  h\rangle\\
&=-\lref\alpha\langle A h,  SA  h\rangle
-(1-\alpha)\lref\langle A h,  A  h\rangle.
\end{align*}
Continuing we have
\begin{align*}
 \langle A h,  A B h\rangle
 &= \langle A h,  (A B - BA) h\rangle + \langle A h,   B A h\rangle\\
 &= \langle A h,  [A,B] h\rangle + 0 = \langle A h,  C h\rangle,
\end{align*}
since by the anti-symmetry of $B$, it follows that $\langle g, Bg\rangle=0$ for any $g$.
\paragraph{Term two.}
We next compute the second term
$$\langle C f, A (-\lref S+B) f\rangle + \langle C (-\lref S+B)f, A f\rangle.$$
First we compute the derivative along $B$
\begin{align*}
 \langle A B h,  C  h\rangle + \langle A  h,  C B h\rangle
 &= \langle A B h,  C  h\rangle + \langle A  h,  BC h\rangle  + \langle A  h,  [C,B] h\rangle \\
 \intertext{and using that $B^\ast = -B$}
 &= \langle A B h,  C  h\rangle - \langle BA  h,  C h\rangle  + \langle A  h,  [C,B] h\rangle \\
 &= \langle [A, B] h,  C  h\rangle + \langle A  h,  [C,B] h\rangle \\
 &= \langle C h,  C  h\rangle + \langle A  h,  [C,B] h\rangle \\
 &= \|Ch\|^2- \langle A  h,  \nabla^2 U A h\rangle.
 \end{align*}
To compute the derivative along $S$ first notice that
$C Q_\alpha = Q_\alpha C$, where in the r.h.s.\ we tensorise $Q_\alpha$ allowing it to act on each component separately, in the sense that
$$\frac{\partial}{\partial x_i} \rE\left[ f\left(\bm{x}, \alpha \bm{v} + \sqrt{1-\alpha^2} \bm{\xi}\right)\right]
 =  \rE\left[ \frac{\partial}{\partial x_i}f\left(\bm{x}, \alpha \bm{v} + \sqrt{1-\alpha^2} \bm{\xi}\right)\right].$$
Therefore
\begin{align*}
\lefteqn{ -\lref\langle A S h,  C  h\rangle  -\lref\langle A  h,  C S h\rangle }\\
&= \lref \langle A (Q_\alpha -I) h, C h\rangle +  \lref \langle A  h, C (Q_\alpha -I)h\rangle\\
 &= \lref\langle (\alpha Q_\alpha -I) A  h,  C  h\rangle + \lref\langle A  h,
  (Q_\alpha -I) C  h\rangle\\
 &= \alpha\lref\langle ( Q_\alpha -I) A  h,  C  h\rangle + (\alpha-1) \lref\langle A  h,  C  h\rangle+ \lref\langle A  h, (Q_\alpha -I) C  h\rangle\\
 &= -(1+\alpha) \lref \langle SA h, C h\rangle -(1-\alpha)\lref \langle A h, C h\rangle,
 \end{align*}
where we used again the fact that $Q_\alpha$ is positive. 
\paragraph{Term three.}
Using the same arguments as before we have
\begin{align*}
 \langle C h,  C Q_\alpha h\rangle
 &= \sum_{i=1}^d \left\langle \frac{\partial}{\partial x_i}h, \frac{\partial}{\partial x_i} Q_\alpha h\right\rangle\\
 &= \sum_{i=1}^d \left\langle\frac{\partial}{\partial x_i}h, Q_\alpha \frac{\partial}{\partial x_i}  h\right\rangle
 = \langle C h, Q_\alpha C  h\rangle,
\end{align*}
where we are overloading the inner product by allowing it to take both vectors and scalars as arguments, in the case of scalars it integrates the product, in the case of vectors the vector inner product.
Therefore
\begin{align*}
 -\lref\langle C h,  C S h\rangle
 &= \lref \langle C h, C (Q_\alpha -I) h\rangle
 =-\lref \langle C h, SC h\rangle.
\end{align*}
The next one is
\begin{align*}
 \langle C h,  C B h\rangle
 &= \langle C h,  C B h\rangle\\
 &= \langle C h,  B C h\rangle - \langle C h,  [B,C] h\rangle  = 0- \langle C h,  \nabla^2 U\cdot A h\rangle.
\end{align*}
\paragraph{Combining all the terms.}
We now have the tools to compute the derivative of
$$\lann h, h\rann:=a\|A h\|^2 - 2b \inner{C h}{A h} + c\|C h\|^2,$$
which, after multiplying by $-1$, is given by
\begin{align*}
\lefteqn{-\frac{\rd }{\rd t} \lann h, h\rann}\\
&= - a\frac{\rd }{\rd t}\|A h\|^2 + 2b\frac{\rd }{\rd t} \inner{A h}{C h} - c\frac{\rd }{\rd t}\|C h\|^2\\
&=2 a \left[\lref(1-\alpha)\|A h\|^2 +\lref\alpha \inner{S Ah}{Ah}
-\inner{Ah}{Ch}\right]\\
&\qquad +2 b \left[ \|Ch\|^2 - \inner{\grad^2 U Ah} {Ah}
- (1+\alpha) \lref  \langle S^{1/2}Ah, S^{1/2}Ch\rangle
-(1-\alpha)\lref\inner{Ah}{Ch}\right]\\
&\qquad +2c \left[ \lref \inner{S C h}{Ch} + \inner{\grad^2 U Ah}{Ch}
\right]
\\
&= 2 a\lref(1-\alpha)\|A h\|^2 -2(a+(1-\alpha)b\lref)\inner{Ah}{Ch} +2 b \|Ch\|^2
\\
&\qquad
-2b \inner{\grad^2 U Ah} {Ah} + 2c\inner{\grad^2 U Ah}{Ch}\\
&\qquad + 2 a\lref\alpha \inner{S Ah}{Ah}+ 2 c\lref \inner{S C h}{Ch}- 2(1+\alpha) b\lref  \inner{S Ah}{Ch}.
\end{align*}
\begin{remark}
At this stage we can rewrite the above inequality as
\begin{align}
-\frac{1}{2}\frac{\rd }{\rd t} \lann h, h\rann
&\geq
\left[ a(1-\alpha) \lref - bM\right]\|A h\|^2 +b \|Ch\|^2  - \left\|J Ah\right\| \|Ch\|
\notag\\
&\qquad +a\alpha \lref \|S^{1/2}Ah\|^2 + c\lref \|S^{1/2}Ch\|^2 - (1+\alpha)b\lref
\|S^{1/2}Ah\| \|S^{1/2}Ch\|,\label{eq:quadraticformtooptimize}
\end{align}
where $S^{1/2}$ is the positive, self-adjoint square root of $S$, and
$$Jf:= \left( aI+(1-\alpha)b \lref I -c \nabla^2 U\right) f,$$
which is also self-adjoint, since $\nabla^2 U$ is symmetric, whence its norm is given by
\begin{align*}
\sup_{\|f\|=1} | \langle Jf, f\rangle|
&=\sup_{\|f\|=1} \left| \left[a+b\lref(1-\alpha)\right] \langle f, f\rangle - c \langle \nabla^2 U f, f\rangle \right|\\
&=\sup_{\|f\|=1}
\max \left\{ \left[a+b\lref(1-\alpha)\right] \langle f, f\rangle - c \langle \nabla^2 U f, f\rangle, c \langle \nabla^2 U f, f\rangle-\left[a+b\lref(1-\alpha)\right] \langle f, f\rangle \right\}\\
&\leq\sup_{\|f\|=1} \max\left\{ \left(a+(1-\alpha)\lref b\right) \|f\| - c m \|f\|, cM \|f\| - \left(a+(1-\alpha)\lref b\right) \|f\| \right\}\\
&= \max\left\{ a+(1-\alpha)\lref b - c m, cM - a-(1-\alpha)\lref b \right\}.
\end{align*}
Therefore, if we can find $a,b,c>0$, such that $b<\sqrt{4 a \alpha c}/(1+\alpha)$ and
\begin{align*}
4\left[ a(1-\alpha) \lref - bM\right]b &> \max\{cM-a-(1-\alpha)\lref b, a+(1-\alpha)b\lref - cm\}^2,
\end{align*}
then the RHS of \eqref{eq:quadraticformtooptimize} is a positive definite quadratic form. In principle this can be used to optimise the convergence rates among
norms of the form \eqref{eq:newnorm}.
\end{remark}
We take a slightly different approach.
Our goal is to show that for every $h$, we have $\frac{\rd }{\rd t} \lann h, h\rann\le -\mu \lann h, h\rann$, or equivalently
\[-\frac{\rd }{\rd t} \lann h, h\rann-\mu \lann h, h\rann\ge 0,\]
After rearrangement, we obtain that
\begin{align}
\nonumber&-\frac{\rd }{\rd t} \lann h, h\rann-\mu \lann h, h\rann\\
\nonumber&= a(2\lref(1-\alpha)-\mu)\|A h\|^2 -2(a+(1-\alpha)b\lref-\mu b)\inner{Ah}{Ch} +(2 b-c \mu) \|Ch\|^2\\
\nonumber&\quad -2b \inner{\grad^2 U Ah} {Ah} + 2c\inner{\grad^2 U Ah}{Ch}\\
&\quad + 2 a\lref\alpha \inner{S Ah}{Ah}+ 2 c\lref \inner{S C h}{Ch}- 2(1+\alpha) b\lref  \inner{S Ah}{Ch}.\label{eq:dperdthh}
\end{align}

We will use the following two lemmas.
\begin{lemma}\label{lem:VXWPZQ}
	If $V, W, Z, A \in \mathbb{R}^{2 \times 2}$ are symmetric matrices such that $0 \preceq A$, $-Z \preceq A$, $A \preceq V+mW$ and $A \preceq V+MW$, then $\Tr(VX + WP + ZQ) \ge 0$ for all symmetric matrices $X,P,Q$ such that $0 \preceq Q \preceq X$ and $mX \preceq P \preceq MX$.
\end{lemma}
\begin{proof}[Proof of Lemma~\ref{lem:VXWPZQ}]
First, suppose that $M=m$. By the assumptions we have $P=mX$, $A\succeq 0$, and $Z+A\succeq 0$. Note that if $S,T$ are symmetric positive semidefinite matrices, then $\Tr(ST) \ge 0$. Using this fact, it follows that
\begin{align*}    \Tr(VX + WP + ZQ)&=\Tr((V+mW)X + ZQ)\ge \Tr(AX + (Z+A)Q-AQ)\\
	&\ge \Tr(A(X-Q))\ge 0.
	\end{align*}
Now suppose that $M>m$.
Let
$$A_1 = Z+A,\quad  A_2 = A, \quad A_3 = \frac{1}{M-m}(V+MW-A), \quad A_4 = \frac{1}{M-m}(V+mW-A).$$
Then $A_1, A_2, A_3, A_4 \succeq 0$, and
$$V = A_2 - m A_3 + M A_4, \quad W = A_3 - A_4, \quad Z = A_1 - A_2.$$
So
\begin{align*}
	VX + WP + ZQ &= (A_2 - m A_3 + M A_4) X + (A_3 - A_4) P + (A_1-A_2) Q \\
	&= A_1 Q + A_2 (X-Q) + A_3 (P-mX) + A_4 (MX-P).
\end{align*}
Using positive definiteness of both terms in the matrix products, we have
$$\Tr(A_1 Q), \Tr(A_2(X-Q)), \Tr(A_3(P-mX)), \Tr(A_4(MX-P)) \ge 0,$$
and therefore
\begin{equation*}
\Tr(V X + W P + Z Q) \ge 0.\qedhere
\end{equation*}
\end{proof}

%
%

Now we are ready to complete the proof of Theorem~\ref{thm:hypoco}.
\begin{proof}[Proof of Theorem \ref{thm:hypoco}]
Let $a:=1$, and
\begin{align*}
b&:=\frac{1+\alpha-\alpha \l(\frac{m}{M+m}\r)^{3/4}+\frac{3}{4} \frac{\alpha m}{M+m}}{2\sqrt{M+m}},\\
c&:=\frac{1+\alpha- \frac{\alpha }{2}\l(\frac{m}{M+m}\r)^{1/2} }{M+m},\\
X&:=\l(\begin{matrix}\|A h\|^2 & \inner{Ah}{Ch}\\\inner{Ah}{Ch} & \|C h\|^2  \end{matrix}\r),\\
P&:=\l(\begin{matrix}\inner{\grad^2 U(x) A h}{A h} & \inner{\grad^2 U(x) Ah}{Ch}\\\inner{\grad^2 U(x) Ah}{Ch} & \inner{\grad^2 U(x) C h}{Ch} \end{matrix}\r),\\
Q&:=\l(\begin{matrix}\inner{S A h}{A h} & \inner{S Ah}{Ch}\\\inner{S Ah}{Ch} & \inner{S C h}{Ch} \end{matrix}\r),\\
V&:=\l(\begin{matrix} 2a(1-\alpha)\lref-a\mu & -a -(1-\alpha)b \lref+b\mu \\
-a -(1-\alpha)b \lref+b\mu  &   2 b-c \mu \end{matrix}\r),\\
W&:=\l(\begin{matrix}-2b & c\\c & 0  \end{matrix}\r),\\
Z&:=\l(\begin{matrix}2a \alpha \lref & -(1+\alpha)b \lref \\ -(1+\alpha)b \lref & 2 c\lref  \end{matrix}\r),\\
A&:= \l(\begin{matrix}\frac{4(-3+2m-2M)(-1+\alpha)}{3\sqrt{m+M}(1+\alpha)} & -\frac{(-3+2m-2M)(-1+\alpha)}{3(m+M)}\\-\frac{(-3+2m-2M)(-1+\alpha)}{3(m+M)} & -\frac{(-3+2m-2M)(-1+\alpha)(1+\alpha)}{3(m+M)^{3/2}}  \end{matrix}\r).
\end{align*}
{Using the fact that $x^{\sfrac{3}{4}}-\sfrac{3}{4}x \geq 0$ for $x\in [0,\sfrac{1}{2}]$, it is easy to check that $b^2<a c$.} Using the assumption that $m I \preceq \grad^2 U\preceq M I$, we have $m X\preceq P \preceq M X$. Moreover, using the fact that $0\preceq S \preceq I$, we have $0\preceq Q\preceq X$. Based on \eqref{eq:dperdthh} and the above definitions it follows that
\begin{equation}
-\frac{\rd }{\rd t} \lann h, h\rann-\mu \lann h, h\rann=\Tr(VX + WP + ZQ).
\end{equation}
One can check, for example using Mathematica, that for every $M\ge 1, 0\le \alpha<1$, the inequalities $0 \preceq A$, $-Z \preceq A$, $A \preceq V+mW$ and $A \preceq V+MW$ hold for $A$ defined as above. Therefore \eqref{eq:hypoco_rate} follows from Lemma \ref{lem:VXWPZQ}, and by Gr\"onwall's lemma, this implies that $\lann P^t f, P^t f\rann\le \exp(-\mu t)\lann f,f\rann$.

\subsubsection{From $H^1$ to $L^2$\label{sec:h1tol2}.}
To show our $L^2$ bound, we study the reversed process. Denote the variant of the scalar product $\lann\cdot,\cdot \rann$ when $b$ is replaced by $-b$ by $\lann \cdot,\cdot\rann'$, i.e.
\begin{equation}\label{eq:newnormp}
\lann h, h\rann':=a\|\nabla_v h\|^2 + 2b \langle \nabla_x h, \nabla_v h\rangle + c\|\nabla_x h\|^2.
\end{equation}
Then by repeating the same arguments as above with $v$ replaced by $-v$ everywhere, one can show that we have
\begin{equation}\frac{\rd }{\rd t} \lann (P^*)^t f, (P^*)^t f\rann'\le -\mu \lann (P^*)^t f, (P^*)^t f\rann',
\label{eq:hypoco_rate}
\end{equation}
and hence $\lann (P^*)^t f, (P^*)^t f\rann'\le \exp(-\mu t)\lann f,f\rann'$. Similarly to the previous proofs, we can show that $\lann\cdot, \cdot\rann$ and $\lann\cdot, \cdot\rann'$ are equivalent up to the same constant factor $C$, and
\[\lann (P^t)^* P^t f, (P^t)^* P^t f\rann\le C^2\exp(-2\mu t) \lann f,f\rann.\]
In addition, there exist constants $C_1, C_2>0$ such that
$\lann f, f\rann \leq C_1 \|\nabla f\|^2$ and $\|f\|^2\leq C_2 \lann f, f\rann$.
Thus, letting $f$ be $k$-Lipschitz we have
\begin{align*}
\|(P^t)^* P^t f\|^2
&\leq C_2 \lann (P^t)^* P^t f, (P^t)^* P^t f\rann'\\
&\leq C_2\exp(-2\mu t)\lann f,f\rann'\\
&\leq C_1 C_2 \exp(-2\mu t) \|\nabla f\|^2 \leq C_1 C_2 k^2 \exp(-2\mu t).
\end{align*}
Choose $t$ such that $C_1C_2 k^2 \re^{-2\mu t}=: 1-\kappa<1$ and define the self-adjoint operator $Q = (P^t)^* P^t$.
Iterating the above we have for $n\geq 1$ that
\begin{align*}
\|Q^n f\|^2
&\leq C_1 C_2 (1-\kappa)^{2n} k^2 =: C(f) (1-\kappa)^{2n}.
\end{align*}
The rest is similar to the proof of Proposition~2.8 from \citet{hairer2014spectral}.
Let $f$ be $k$-Lipschitz, and without loss of generality also assume that $\|f\|=1$.
Let $\nu_f$ be the spectral measure corresponding to the self-adjoint operator $Q$ applied to the function $f$. In particular, since $\|f\|=1$, $\nu_f$ is a probability measure.
Then
\begin{align*}
\|Q^n f\|^2
&= \int_{-1}^1 t^{2n} \nu_f (\rd t)\\
&= \int_{-1}^1 t^{2n (2n+2m)/(2n+2m)} \nu_f (\rd t)\\
&\leq \left(\int_{-1}^1 t^{2(n+m)}\nu_f (\rd t)\right)^{\frac{2n}{2(n+m)}}\\
&= \left(\|Q^{n+m} f\|^2 \right)^{\frac{2n}{2(n+m)}}\\
&\leq \left[C(f) (1-\kappa)^{2(n+m)} \right]^{\frac{2n}{2(n+m)}}\\
&\leq C(f)^{\frac{2n}{2(n+m)}} (1-\kappa)^{2n},
\end{align*}
and letting $m\to \infty$ we get for any $k$-Lipschitz $f$
\begin{align*}
\|Q^n f\|^2
&\leq \|f\|^2(1-\kappa)^{2n},
\end{align*}
noticing that the upper bound is independent of the Lipschitz constant.
Since Lipschitz functions are dense we conclude.
\end{proof}

\begin{remark}
Given any $\lref>0,\mu>0$, the contraction $\frac{\rd }{\rd t} \lann h, h\rann\le -\mu \lann h, h\rann$ holds as long as there exists coefficients $a,b,c\in \mathbb{R}$ and a $2\times 2$ real valued symmetric matrix $A$ such that $a>0$, $c>0$, $b^2<a c$ and
$0 \preceq A$, $-Z \preceq A$, $A \preceq V+mW$ and $A \preceq V+MW$ (with $V$ and $W$ defined as above). Note that as in the proof of Theorem~\ref{thm:Wasserstein}, due to the non-linearity of the constraints we did not manage to find an analytical expression for the largest possible $\mu$ for a given $\lref$, and the largest possible $\mu$ for any $\lref$. However, we believe that the choice of $\lref$ and $\mu$ as given here is close to optimal in most of the parameter range $0<m\le M<\infty$, $0\le \alpha<1$.
\end{remark}

\section*{Acknowledgements}
The authors would like to thank Peter Holderrieth for a careful reading of the manuscript and his invaluable suggestions and Philippe Gagnon for his insightful comments on the manuscript. G.D.\ would like to thank Gabriel Stoltz for many useful discussions.
The authors would also like to thank the anonymous referees for numerous suggestions that have greatly improved the content and the presentation of the paper.
This material is based upon work supported in part by the U.S. Army Research Laboratory and the U. S. Army Research Office, and by the U.K. Ministry of Defence (MoD) and the U.K. Engineering and Physical Research Council (EPSRC) under grant number EP/R013616/1 and by the EPSRC EP/R034710/1. A part of this research was done while A. Doucet, G. Deligiannidis and D. Paulin were hosted by the Institute for Mathematical Sciences in Singapore.
\nocite{*}
\bibliographystyle{plainnat}
\begin{appendix}
	\section{Auxiliary results}
Notice that using the independence of $X$ and $Z$, and the fact that the standard normal distribution is isotropic, we have
\[\rE_{X\sim \pi, Z\sim N(0,\mathbb{I}_d)}\left[ \left( \grad U(X),Z\right)_+\right]=\rE(|\grad U(X)|)  \rE\left[ \left( w ,Z\right)_+\right],\]
where $w$ is an arbitrary fixed $d$ dimensional unit vector. Now noticing that $\left( w ,Z\right)$ is a one dimensional standard normal random variable, it follows that
$\rE\left[ \left( w ,Z\right)_+\right]=\int_{x=0}^{\infty}\frac{1}{\sqrt{2\pi}} x \exp\l(-\frac{x^2}{2}\r) \mathrm{d}x=\frac{1}{\sqrt{2\pi}}$. Hence the key part of the proof is to find lower and upper bounds on
\[\rE(|\grad U(X)|)=\frac{\int_{x\in \rR^d} |\grad U(x)| e^{-U(x)} \mathrm{d}x }{\int_{x\in \rR^d}  e^{-U(x)} \mathrm{d}x }.\]
By shifting $U$, we can assume without loss of generality that $U(0)=0$ and $\grad U(0)=0$ (hence the minimum is taken in the origin 0).
Let $\mathbb{S}^d_1$ denote the $d$-dimensional unit sphere, then by writing the above integrals along half-lines, we have
\begin{equation}\label{eq:rsphereint}\rE(|\grad U(X)|)=\frac{\int_{u\in \mathbb{S}^d_1}  \int_{r=0}^{\infty} |\grad U(ru) | e^{-U(ru)} r^{d-1} \mathrm{d}r \, \mathrm{d} u }{\int_{u\in \mathbb{S}^d_1}  \int_{r=0}^{\infty} e^{-U(ru)} r^{d-1} \mathrm{d}r\, \mathrm{d} u }\ge \frac{\int_{u\in \mathbb{S}^d_1}  \int_{r=0}^{\infty} \l|\frac{\partial }{\partial r} U(ru) \r| e^{-U(ru)} r^{d-1} \mathrm{d}r \, \mathrm{d} u }{\int_{u\in \mathbb{S}^d_1}  \int_{r=0}^{\infty} e^{-U(ru)} r^{d-1} \mathrm{d}r\, \mathrm{d} u }
\end{equation}
If we could lower bound the ratios of the one dimensional integrals $$\frac{\int_{r=0}^{\infty} \l|\frac{\partial }{\partial r} U(ru) \r|  e^{-U(ru)} r^{d-1} \mathrm{d}r}{\int_{r=0}^{\infty} e^{-U(ru)} r^{d-1} \mathrm{d}r},$$ then a lower bound for $\rE(|\grad U(X)|)$ follows by rearrangement. This is shown in the following Lemma.

\begin{lemma}\label{lemma:BPSexpectednbofbounceslower}
	Let $d \in \rZ_{\ge 1}$, $m \in \rR_{>0}$, and let $V \colon \rR_{\ge 0} \to \rR$ be a differentiable function such that $x \mapsto V(x) - m \frac{x^2}{2}$ is convex, and $V'(0) = 0$.
	Let $A = \int_0^{\infty} x^{d-1} e^{-V(x)} dx$ and $B = \int_0^{\infty} V'(x) x^{d-1} e^{-V(x)} dx$.
	Then $B \ge \sqrt{2m} \frac{\Gamma(\frac{d+1}{2})}{\Gamma(\frac{d}{2})} A$.
\end{lemma}
\begin{proof}
	First let $d = 1$.
	Then $B = \int_0^{\infty} (-e^{-V(x)})' dx = e^{-V(0)}$.
	We have $V(x) \ge V(0) + m \frac{x^2}{2}$ for $x \ge 0$, so $A \le \int_0^{\infty} e^{-V(0)-m \frac{x^2}{2}} dx = e^{-V(0)} \sqrt{\frac{\pi}{2m}} = \sqrt{\frac{\pi}{2m}} B$, so $B \ge \sqrt{2m} \frac{\Gamma(1)}{\Gamma(\frac{1}{2})} A$.
	
	Now let $d \ge 2$.
	Then
	\[
	B = \int_0^{\infty} ((-x^{d-1} e^{-V(x)})' + (d-1) x^{d-2} e^{-V(x)}) dx = (d-1) \int_0^{\infty} x^{d-2} e^{-V(x)} \rd x,
	\]
	so the claim is equivalent to $\int_0^{\infty} (c-x) x^{d-2} e^{-V(x)} \rd x \ge 0$, where $c = \frac{\Gamma(\frac{d}{2})}{\Gamma(\frac{d-1}{2})} \cdot \frac{\sqrt{2}}{\sqrt{m}}$ (here we have used $\Gamma(\frac{d+1}{2}) = \frac{d-1}{2} \Gamma(\frac{d-1}{2})$).
	The function $x \mapsto V(x) - m \frac{x^2}{2}$ is convex, and its derivative at $x = 0$ is $0$, so this function is monotone increasing on $\rR_{\ge 0}$.
	Hence $V(x) \ge V(c) + \frac{m}{2}(x^2-c^2)$ if $x \ge c$, and $V(x) \le V(c) + \frac{m}{2}(x^2-c^2)$ if $x \le c$.
	Thus
	\begin{align*}
	\int_0^{\infty} (c-x) x^{d-2} e^{-V(x)} \rd x &\ge \int_0^{\infty} (c-x) x^{d-2} e^{-V(c) - \frac{m}{2}(x^2-c^2)} \rd x \\
	&= e^{\frac{m}{2} c^2 - V(c)} \int_0^{\infty} (c-x) x^{d-2} e^{-\frac{m}{2} x^2} \rd x.
	\end{align*}
	We have $\int_0^{\infty} x^{\alpha} e^{-\frac{m}{2} x^2} \rd x = \frac{1}{2} (\frac{m}{2})^{-\frac{\alpha+1}{2}} \Gamma(\frac{\alpha+1}{2})$ for every $m > 0$ and $\alpha > -1$.
	So
	\[
	\int_0^{\infty} (c-x) x^{d-2} e^{-V(x)} \rd x \ge e^{\frac{m}{2} c^2 - V(c)} \frac{1}{2} \l(c \l(\frac{m}{2}\r)^{-\frac{d-1}{2}} \Gamma\l(\frac{d-1}{2}\r) - \l(\frac{m}{2}\r)^{-\frac{d}{2}} \Gamma\l(\frac{d}{2}\r)\r) = 0.\qedhere
	\]	
\end{proof}
The following lemma will be used to find a simpler lower bound for the ratio $\frac{\Gamma(\frac{d+1}{2})}{\Gamma(\frac{d}{2})}$.
\begin{lemma}\label{lemma:Gammabnd}
	If $s > 0$, then $\frac{\Gamma(s+\frac{3}{4})}{\Gamma(s+\frac{1}{4})} > \sqrt{s}$.
\end{lemma}
\begin{proof}
	Let $\phi(s) := \frac{1}{\sqrt{s}} \frac{\Gamma(s+\frac{3}{4})}{\Gamma(s+\frac{1}{4})}$ for $s > 0$.
	Stirling's formula implies that $\lim_{s \to \infty} \phi(s) = 1$.
	For $s > 0$ we have $\phi(s) > 0$ and $(\frac{\phi(s+1)}{\phi(s)})^2 = 1 - \frac{1}{(1+s) (1+4s)^2} < 1$, so $\phi(s) > \phi(s+1)$.
	Thus $\phi(s) > \phi(s+1) \ge \phi(s+n)$ for every $n \in \rZ_{\ge 1}$, and taking $n \to \infty$ we get $\phi(s) > 1$.
\end{proof}

Taking $s = \frac{d-\frac{1}{2}}{2}$ for $d \in \rZ_{\ge 1}$, we get
\begin{equation}\label{eq:Gammalowerbnd}
\frac{\Gamma(\frac{d+1}{2})}{\Gamma(\frac{d}{2})} > \sqrt{\frac{d-\frac{1}{2}}{2}}.
\end{equation}

The next lemma will show the upper bound.

\begin{lemma}\label{lemma:BPSexpectednbofbouncesupper}
Suppose that the potential $U:\rR^n \to \rR$ satisfies Assumption \ref{ass:potential}. Then for every $1\le i\le n$, we have $\rE\l( (\partial_i U(X))^2 \r) \le M$, implying that $\rE\left( |\grad U_n(\bm{X}) |^2 \right) \leq n M$ and
$\rE\left(|\grad U_n(\bm{X})|\right)\le \sqrt{nM}$.
\end{lemma}

\begin{proof}
By Jensen's inequality, we have
\[\rE(|\grad U(X)|)\le \l[\rE(|\grad U(X)|^2)\r]^{1/2}=\l[\rE(\sum_{i=1}^{d} (\partial_i U(X))^2 \r]^{1/2}.
\]
Here
\[\rE\l( (\partial_i U(X))^2 \r) = \frac{\int_{x\in \rR^d} (\partial_i U(x))^2 \exp(-U(x)) \mathrm{d} x } {\int_{x\in \rR^d} \exp(-U(x)) \mathrm{d} x},\]
and from integration by parts, it follows that for every $1\le i\le d$, we have
\begin{align*}
&\int_{x\in \rR^d} (\partial_i U(x))^2 \exp(-U(x)) \mathrm{d} x=\int_{x_{-i}\in \rR^{d-1}} \int_{x_i\in \rR} (\partial_i U(x))^2 \exp(-U(x)) \mathrm{d}x_{i} \mathrm{d}x_{-i}\\
&=\int_{x_{-i}\in \rR^{d-1}} \left\{ \left[  -\partial_i U(x) \exp(-U(x)) \right]_{x_i={-\infty}}^{\infty} +  \int_{x_i\in \rR} \partial_i^2 U(x) \exp(-U(x)) \mathrm{d}x_{i} \right\} \mathrm{d}x_{-i}\\
&=\int_{x_{-i}\in \rR^{d-1}} \int_{x_i\in \rR} \partial_i^2 U(x) \exp(-U(x)) \mathrm{d}x_{i} \mathrm{d}x_{-i}\le M \int_{x\in \rR^d} \exp(-U(x)) \mathrm{d} x.
\end{align*}
The second and third claims now follow by summing up in $i$, and using Jensen's inequality.
\end{proof}

\begin{proof}[Proof of Proposition \ref{proposition:BPSexpectednbofbounces}]
The result follows from Lemmas \ref{lemma:BPSexpectednbofbounceslower}, \ref{lemma:Gammabnd} and \ref{lemma:BPSexpectednbofbouncesupper}.
\end{proof}

\begin{lemma}\label{lem:Unpartial1divbynorm}
Suppose that $U_n(\bm{X}):\rR^n\to \rR$ with $m \bm{I}_d\preceq\nabla^2 U_n(\bm{X}) \preceq M \bm{I}_d$. Then
\[\rE\left[ \left. \frac{\partial_1 U_n(\bm{X}) V_1}{\left(\sum_{j=1}^n [\partial_j U_n(\bm{X})]^2 \right)^{1/2}} \right| X_1, V_1 \right] \to 0\text{ as }n\to \infty.\]
\end{lemma}
\begin{proof}
\[\rE\left[ \left. \frac{\partial_1 U_n(\bm{X}) V_1}{\left(\sum_{j=1}^n [\partial_j U_n(\bm{X})]^2 \right)^{1/2}} \right| X_1, V_1 \right]\le \rE\left[ \left. \frac{|\partial_1 U_n(\bm{X})|}{|\nabla U_n(\bm{X})|}\right| X_1\right]\cdot |V_1|.\]
Let us denote $X_{-1}:=(X_2,\ldots,X_n)$, then $X_{-1}$ given $X_1$ has a conditional distribution with density that is proportional to $\exp(-U_n(X_{-1},X_1))$, which is a log-concave function of $X_{-1}$, with Hessian bounded between $m$ and $M$. By Theorem 5.2 of \cite{Ledouxconcentrationofmeasure},  $\mathcal{L}(X_{-1}|X_1)$ satisfies a log-Sobolev inequality with constant $C:=m^{-1}$. The functions $|\nabla U_n(\bm{X})|$ and $|\partial_1 U_n(\bm{X})|$ are $M$-Lipschitz in $X_{-1}$ given a fixed $X_1$, and hence  by Herbst's argument (see equation (5.8) on page page 95 of \cite{Ledouxconcentrationofmeasure}),
\begin{align*}
&\rP( |\partial_1 U_n(\bm{X})|-\rE(|\partial_1 U_n(\bm{X})||X_1)\ge t |X_1)\le \exp\left(-t^2\cdot \frac{2m}{M^2}\right),\\
&\rP( ||\nabla U_n(\bm{X})|-\rE(|\nabla U_n(\bm{X})| |X_1) \le -t |X_1)\le \exp\left(-t^2\cdot \frac{2m}{M^2}\right).
\end{align*}
Conditionally on $X_1$, define the event $G_t$ as
\[G_t:=\{ |\partial_1 U_n(\bm{X})|-\rE(|\partial_1 U_n(\bm{X})||X_1) < t \text{ and } |\nabla U_n(\bm{X})|-\rE(|\nabla U_n(\bm{X})| |X_1) > -t\},\]
then by the above bounds, we have $\rP(G_t|X_1)\ge 1-2\exp\left(-t^2\cdot \frac{2m}{M^2}\right)$ for every $t\ge 0$. Let $G_t^c$ denote the complement of $G_t$. Assuming that $0<t<\rE(|\nabla U_n(\bm{X})| |X_1)$, the quantity of interest can be bounded as
\begin{align}\nonumber&\rE\left[ \left. \frac{|\partial_1 U_n(\bm{X})|}{|\nabla U_n(\bm{X})|}\right| X_1\right]= \rE\left[ \left. \frac{|\partial_1 U_n(\bm{X})|}{|\nabla U_n(\bm{X})|} \cdot 1_{G_t}\right| X_1\right] + \rE\left[ \left. \frac{|\partial_1 U_n(\bm{X})|}{|\nabla U_n(\bm{X})|} \cdot 1_{G_t^c}\right| X_1\right]\\
&\le  \frac{ \rE(|\partial_1 U_n(\bm{X})||X_1) +t}{\rE(|\nabla U_n(\bm{X})| |X_1)-t} + 2\exp\left(-t^2\cdot \frac{2m}{M^2}\right),\label{eq:goodbadeventbound}
\end{align}
where we have used the fact that $\frac{|\partial_1 U_n(\bm{X})|}{|\nabla U_n(\bm{X})|}\le 1$. By Lemma 9 and equation (A.2), it follows that for any $n\ge 2$,
\[\rE(|\nabla U_n(\bm{X})| |X_1)\ge \rE(|\partial_{-1} U_n(\bm{X})| |X_1)\ge \sqrt{m (n-3/2)},\]
where $\partial_{-1}U_n(\bm{X})$ denotes the gradient vector without the first component. By Lemma 11,
\[\rE(|\partial_1 U_n(\bm{X})|)\le \sqrt{M}.\]
Note that $|\partial_{-1} U_n(\bm{X}_{-1},X_{1})-\partial_{-1} U_n(\bm{X}_{-1},X_{1}')|\le M |X_{1}-X'_1|$, and by Proposition 19 of \cite{vono2020efficient}, it follows that
\[W_1(\mathcal{L}(X_{-1}|X_1), \mathcal{L}(X_{-1}|X_1'))\le \frac{M}{m} |X_1-X_1'|,\]
therefore $g(X_1):=\rE(|\partial_1 U_n(\bm{X})||X_1)$ is $\frac{M^2}{m}$-Lipschitz in $X_1$. By log-Sobolev inequality and Herbst's argument, for any $s\ge 0$, we have
\[\rP(|X_1 -\rE(X_1)|\ge s)\le 2\exp\left(-s^2 \cdot 2m\right).\]
Therefore, it follows that
\begin{align*}
\sqrt{M}&\ge \rE[g]= \rE[g(X_1)-g(\rE(X_1))]+ g(\rE(X_1))\ge -\int_{r=0}^{\infty}\rP[g(X_1)-g(\rE(X_1))\le -r]\mathrm{d}r + g(\rE(X_1))\\
&\ge  -\int_{r=0}^{\infty}\rP\left[|X_1-\rE(X_1)|\ge r \frac{m}{M^2}\right] \mathrm{d}r + g(\rE(X_1))\ge -\int_{r=0}^{\infty}2 \exp\left(-r^2\frac{2m^3}{M^4}\right)\mathrm{d}r\\
&\ge g(\rE(X_1))-2 \sqrt{(\pi/2)M^4/m^3}\ge  g(\rE(X_1))-3 \frac{M^2}{m^{3/2}}.
\end{align*}
Thus $g(\rE(X_1))\le 4 \frac{M^2}{m^{3/2}}$, which implies by the Lipschitz property that implying that
\[\rE(|\partial_1 U_n(\bm{X})||X_1)=g(X_1)\le 4 \frac{M^2}{m^{3/2}}+\frac{M^2}{m} |X_1-\rE(X_1)|.\]
By simple algebra, $t=\sqrt{\log(n) M^2/(2m)}$ satisfies that for $n\ge 3/2+2\log(n) \frac{M^2}{m^2}$, we have $t\le \frac{1}{2}\sqrt{m(n-3/2)}$. By combining the above bound with \eqref{eq:goodbadeventbound} and using this $t$, we have
\[\rE\left[ \left. \frac{|\partial_1 U_n(\bm{X})|}{|\nabla U_n(\bm{X})|}\right| X_1\right]\le \frac{ 4 \frac{M^2}{m^{3/2}}+\frac{M^2}{m} |X_1-\rE(X_1)|+\sqrt{\log(n) M^2/(2m)}}{\frac{1}{2}\sqrt{m (n-3/2)}} + \frac{2}{n},\]
as long as $n\ge 3/2+2\log(n) \frac{M^2}{m^2}$. This tends to $0$ as $n\to \infty$.
\end{proof}

\begin{lemma}\label{lemma:inversemoment}
Suppose that $U_n$ satisfies Assumption~\ref{ass:potential} and let $\bm{X}\sim \bar{\pi}_n$. Then for any $\alpha >0$
$$\lim_{n\to \infty}\rE\left[\frac{1}{|\grad U_n(\bm{X})|^
\alpha} \right]= 0. $$
\end{lemma}
\begin{proof}[Proof of Lemma~\ref{lemma:inversemoment}]
We have
\begin{equation}\label{eq:expinvmomentbnd}\rE\left[\frac{1}{|\grad U_n(\bm{X})|^
\alpha} \right]=\int_{t=0}^{\infty}\rP\left[\frac{1}{|\grad U_n(\bm{X})|^
\alpha}\ge t \right] \mathrm{d}t=\rP\left[|\grad U_n(\bm{X})|\le t^{-1/\alpha} \right]\mathrm{d}t.
\end{equation}
The function $|\grad U_n(\bm{x})|$ is $M$-Lipschitz in $\bm{x}$, so by the log-Sobolev inequality and Herbst's argument (see \cite{Ledouxconcentrationofmeasure}), for any $s\ge 0$, we have
\[
\rP(|\grad U_n(\bm{x})|\le \rE(|\grad U_n(\bm{x})|)-s)\le \exp\left(-s^2 \cdot \frac{2m}{M^2}\right).
\]
In the proof of Proposition \ref{proposition:BPSexpectednbofbounces}, we have shown that $\rE(|\grad U_n(\bm{x})|)\ge \sqrt{n-\frac{1}{2}}\sqrt{m}$, hence for any $s\ge 0$,
\begin{equation}\label{eq:normgradlbnd1}
\rP\left(|\grad U_n(\bm{x})|\le \sqrt{n-\frac{1}{2}}\sqrt{m}-s\right)\le \exp\left(-s^2 \cdot \frac{2m}{M^2}\right).
\end{equation}
This bound will be used to control $\rP\left[|\grad U_n(\bm{X})|\le t^{-1/\alpha} \right]$ for small and intermediate values of $t$. However, for large $t$, the above concentration bound is not sufficiently sharp, as it does not tends to zero as $t\to \infty$. Hence we will use a different argument, that upper bounds the density of $\bar{\pi}_n$ and the volume of the space where $|\grad U_n(\bm{X})|\le r$.

First, note that by Assumption~\ref{ass:potential}, we have $U_n(0)=0$ and $U_n$ is minimized in $0$. Using the lower and upper bounds on the Hessian of $U_n$, it follows that
$\frac{m}{2} |x|^2\le U_n(\bm{x})\le \frac{M}{2} |\bm{x}|^2$. These bounds correspond to the log-likelihoods of Gaussian densities, so the normalising constant of $U_n$ can be bounded as
\begin{equation}
\frac{(2\pi)^{n/2}}{M^{n/2}}\le \int_{\bm{x}\in \rR^d} \exp(-U_n(\bm{x}))\mathrm{d}\bm{x}\le \frac{(2\pi)^{n/2}}{m^{n/2}}.
\end{equation}
Moreover, using the bounds on the Hessian of $U_n$, it follows that $|\grad U_n(\bm{X})|\le r$ implies that $|\bm{X}|\le \frac{r}{m}$. Since the volume of a ball of radius $\frac{r}{m}$ in $\rR^n$ is
\[V_n=\frac{\pi^{n/2}}{\Gamma(\frac{n}{2}+1)}\left(\frac{r}{m}\right)^{n}\le 6\left(\frac{r}{m}\right)^{n},\]
it follows that
\begin{equation}\label{eq:normgradlbnd2}
\rP(|\grad U_n(\bm{X})|\le r)\le \rP\left(|\bm{X}|\le \frac{r}{m}\right)\le  6 \frac{M^{n/2}}{(2\pi)^{n/2}} \left(\frac{r}{m}\right)^{n}.
\end{equation}
Let $a:=(\sqrt{n-\frac{1}{2}} \sqrt{m}/2)^{-\alpha}$, and $b=\left(\frac{m \sqrt{2\pi}}{2\sqrt{M}}\right)^{-\alpha}$. By upper bounding $\rP\left[|\grad U_n(\bm{X})|\le t^{-1/\alpha} \right]$ by $1$ for $0\le t\le a$, by $\exp\left(-\frac{(n-\frac{1}{2}) m^2}{2 M^2}\right)$ for $a<t\le b$ (using \eqref{eq:normgradlbnd1}), and by
$6 t^{-n/\alpha}\left(\frac{\sqrt{M}}{m\sqrt{2\pi}}\right)^{n}$ for $t>b$, by \eqref{eq:expinvmomentbnd}, for $n>\alpha$, we have
\begin{align*}\rE\left[\frac{1}{|\grad U_n(\bm{X})|^\alpha} \right]&\le \left(\sqrt{n-\frac{1}{2}} \sqrt{m}/2\right)^{-\alpha}+\left(\frac{m \sqrt{2\pi}}{2\sqrt{M}}\right)^{-\alpha}\cdot \exp\left(-\frac{(n-\frac{1}{2}) m^2}{2 M^2}\right)+6\left(\frac{\sqrt{M}}{m\sqrt{2\pi}}\right)^{n}\frac{b^{-\frac{n}{\alpha}+1}}{\frac{n}{\alpha}-1}\\
&\le \left(\sqrt{n-\frac{1}{2}} \sqrt{m}/2\right)^{-\alpha}+\left(\frac{m \sqrt{2\pi}}{2\sqrt{M}}\right)^{-\alpha}\cdot \exp\left(-\frac{(n-\frac{1}{2}) m^2}{2 M^2}\right)+6 \left(\frac{\sqrt{M}}{m}\right)^{\alpha} \frac{2^{-n}}{\frac{n}{\alpha}-1}
\end{align*}
which tends to $0$ as $n\to \infty$.
\end{proof}

\end{appendix}

\end{document}